\documentclass{article}
\usepackage{fullpage}
\usepackage{amsmath,amssymb,amsthm}
\usepackage[hidelinks]{hyperref}
\usepackage{tikz}
\usetikzlibrary{decorations.pathreplacing,patterns}
\usepackage{thmtools,thm-restate}
\usepackage{cleveref}
\crefname{lemma}{Lemma}{Lemmas}
\crefname{proposition}{Proposition}{Propositions}
\crefname{definition}{Definition}{Definitions}
\crefname{theorem}{Theorem}{Theorems}
\crefname{corollary}{Corollary}{Corollaries}
\crefname{claim}{Claim}{Claims}
\crefname{section}{Section}{Sections}
\crefname{appendix}{Appendix}{Appendices}
\crefname{figure}{Fig.}{Figs.}

\title{Span programs and quantum time complexity}
\author{
Arjan Cornelissen\thanks{QuSoft and University of Amsterdam, \texttt{arjan@cwi.nl}.}
\and
Stacey Jeffery\thanks{QuSoft and CWI, \texttt{jeffery@cwi.nl}. Supported by an NWO Veni Innovational Research Grant under project number 639.021.752, an NWO WISE Grant, and QuantERA project QuantAlgo 680-91-03. SJ is a CIFAR Fellow in the Quantum Information Science Program.}
\and
Maris Ozols\thanks{QuSoft and University of Amsterdam, \texttt{marozols@gmail.com}.
Supported by an NWO Vidi grant VI.Vidi.192.109.}
\and
Alvaro Piedrafita\thanks{QuSoft and CWI, \texttt{piedrafita@cwi.nl}.}
}
\date{\today}

\def\ket#1{{\lvert}#1\rangle}
\def\bra#1{{\langle}#1\rvert}
\def\braket#1#2{{{\langle}#1\vert}#2\rangle}
\newcommand{\norm}[1]{\left\|#1\right\|}
\newcommand{\Cal}[1]{\mathcal{#1}}
\renewcommand{\Re}{\text{Re}}
\newcommand{\A}{\mathcal{A}}
\newcommand{\B}{\mathcal{B}}
\newcommand{\C}{\mathbb{C}}
\renewcommand{\H}{\mathcal{H}}
\renewcommand{\L}{\mathcal{L}}
\newcommand{\N}{\mathbb{N}}
\renewcommand{\O}{\mathcal{O}}

\renewcommand{\S}{\mathcal{S}}
\newcommand{\U}{\mathcal{U}}
\newcommand{\V}{\mathcal{V}}
\newcommand{\W}{\mathcal{W}}

\newcommand{\I}{I}
\newcommand{\eps}{\varepsilon}
\newcommand{\Htrue}{\H_{\textnormal{true}}}
\newcommand{\Hfalse}{\H_{\textnormal{false}}}

\DeclareMathOperator{\Span}{span}
\DeclareMathOperator{\Ker}{ker}

\newtheorem{theorem}{Theorem}
\newtheorem{lemma}[theorem]{Lemma}

\newtheorem{definition}[theorem]{Definition}

\newtheorem*{conjecture*}{Conjecture}

\theoremstyle{definition}

\begin{document}

\maketitle

\begin{abstract}
Span programs are an important model of quantum computation due to their tight correspondence with quantum query complexity. For any decision problem $f$, the minimum complexity of a span program for $f$ is equal, up to a constant factor, to the quantum query complexity of $f$. Moreover, this correspondence is constructive. A span program for $f$ with complexity $C$ can be compiled into a bounded error quantum algorithm for $f$ with query complexity $\O(C)$, and vice versa.

In this work, we prove an analogous connection for quantum time complexity. In particular, we show how to convert a quantum algorithm for $f$ with time complexity $T$ into a span program for $f$ such that it compiles back into a quantum algorithm for $f$ with time complexity $\widetilde{\O}(T)$. While the query complexity of quantum algorithms obtained from span programs is well-understood, it is not generally clear how to implement certain query-independent operations in a time-efficient manner. We show that for span programs derived from algorithms with a time-efficient implementation, we can preserve the time efficiency when implementing the span program. This means in particular that span programs not only fully capture quantum query complexity, but also quantum time complexity.

One practical advantage of being able to convert quantum algorithms to span programs in a way that preserves time complexity is that span programs compose very nicely. We demonstrate this by improving Ambainis's variable-time quantum search result using our construction through a span program composition for the OR function.
\end{abstract}

\newpage

\tableofcontents

\newpage

\section{Introduction}
\label{sec:intro}

Span programs are a model of computation first introduced in the context of classical complexity theory~\cite{KW93}, and later introduced to the study of quantum algorithms by Reichardt and Sp\v{a}lek~\cite{RS12} who designed quantum algorithms for formula evaluation based on span programs. This connection to quantum algorithms proved to be particularly important when Reichardt showed that span programs are equivalent to dual solutions to the adversary bound, proving that the adversary bound is a tight lower bound on quantum query complexity, and span program complexity is a tight upper bound~\cite{Rei09}. In particular, this means that for any decision problem, it is possible to design a query-optimal quantum algorithm using the span program framework, although finding such an algorithm is generally hard in practice. Later work connecting quantum space complexity to span programs enriched this connection further, showing that span program algorithms can also have optimal space complexity~\cite{Jef20}, although again, there is no prescriptive recipe for finding such an optimal algorithm.

While finding optimal span programs is difficult in general, span programs have been used to design quantum algorithms for a variety of problems, including $st$-connectivity~\cite{BR12}, cycle detection and bipartiteness testing~\cite{Ari15,CMB16}, maximum bipartite matching~\cite{BT20}, graph connectivity~\cite{JJKP18},
$k$-distinctness~\cite{Bel12}, and formula evaluation~\cite{RS12,Rei09,JK2017}, some of which have optimal query complexity. Given a span program $P$ that decides a function $f$ and has complexity $C(P)$, there is a generic transformation that compiles the span program into a quantum query algorithm with query complexity $\O(C(P))$. Thus, analyzing the quantum query complexity of the algorithm is as simple as upper bounding the quantity $C(P)$.

This algorithm works by doing phase estimation to precision $C(P)^{-1}$ on a certain unitary $U$ associated with the span program (for details, see \cref{sec:span-program-alg}). While it is clear from the form of this unitary that it can be implemented using $\O(1)$ quantum queries to the input, it is not at all clear how one should time-efficiently implement the non-query parts of $U$. This will generally depend on the structure of the specific span program $P$ in question. While there has been some success in designing time-efficient implementations of $U$ for the case of span programs for certain graph problems, beginning with the work of~\cite{BR12}, for other algorithms, perhaps most notably the span program algorithm for $k$-distinctness~\cite{Bel12}, no time-efficient implementation is known. This is rather unsatisfying, since for a problem such as $k$-distinctness, where one wants to decide if the input string $x\in \{1,\dotsc,q\}^n$ has $k$ entries that are the same, it seems intuitive that queries to the input should be the dominating cost in any optimal algorithm, and so the time complexity should be within $\log n$ factors of the query complexity. And yet, while the best known upper bound on the query complexity is $\O(n^{\frac{3}{4}-\frac{1}{4(2^k-1)}})$, obtained via a span program in~\cite{Bel12}, a time analysis of this algorithm has proven illusive, and the best known upper bound on the time complexity of $k$-distinctness for $k>3$ is the significantly worse $\widetilde{\O}(n^{(k-1)/k})$~\cite{Jef14}. This difficulty in analyzing the time complexity of span program algorithms represents a major drawback to an otherwise powerful framework since query complexity does not fully represent the actual complexity of the algorithm.

In this work, we make progress in understanding the relationship between span programs and quantum time complexity by showing that for any decision problem, it is possible to design an almost time-optimal quantum algorithm (i.e., optimal up to polylogarithmic factors) using the span program framework. To do this, we give a construction that takes any quantum algorithm with time complexity $T$ and query complexity $S$, and maps it to a span program with complexity $\O(S)$, such that the unitary $U$ associated with the span program can be implemented in time $T/S$, up to polylog$(T)$ factors, meaning that the algorithm compiled from the span program has time complexity $\widetilde{\O}(T)$.

In the analysis of the time complexity of the newly-defined span program, we identify an input-dependent subspace of the state space which we are guaranteed to stay within throughout the execution of the span program algorithm. This allows us to drastically decrease the implementation cost of some of its subroutines. We refer to this subspace as the \textit{implementing subspace}, and we believe that this technique can be used to analyze the time complexity of a wider variety of algorithms than those considered in this text.

The problem of mapping an arbitrary quantum algorithm to a span program has been considered previously. In~\cite{Rei09}, Reichardt showed how to convert any quantum query algorithm with one-sided error\footnote{An algorithm with one-sided error must always output 1 on a 1-input, but may err with probability $1/3$ on a 0-input (or vice versa).} to a span program whose complexity matches the algorithm's query complexity. This was extended to the standard case of (two-sided) bounded error quantum query algorithms in \cite{Jef20}. We extend these results to time complexity, showing that a quantum algorithm with time complexity $T$ can be mapped to a span program that, if compiled back into an algorithm, can be implemented in time $\widetilde{\O}(T)$.

The major theoretical implication of this result is that for any decision problem, one can find a quantum algorithm that is optimal in not only space and query complexity, but also time complexity, using the span program framework. Moreover, using our construction we prove that these three flavors of optimality can be achieved simultaneously. Thus, we can definitively say that span programs \emph{are} quantum algorithms.

It is natural to ask if our result, and in particular our construction mapping quantum algorithms to span programs, is of practical relevance since normally quantum algorithms themselves are the end goal in designing span programs. One reason that it can be useful to convert a quantum algorithm into a span program is that span programs compose very nicely~\cite{Rei09} -- more so than quantum algorithms. It can thus be desirable to convert several quantum algorithms to span programs, compose them, and then convert the result back to a quantum algorithm.

To illustrate this, we improve a result of Ambainis~\cite{Amb06} for variable-time quantum search. Given~$n$ bounded-error quantum query algorithms evaluating Boolean functions $f_1, \dots, f_n$ with costs $C_1, \dots, C_n$, respectively, Ambainis provides a way to evaluate the function $f = \bigvee_{i=1}^n f_i$ with cost $\O(\sqrt{\sum_{i=1}^n C_i^2})$. We left the notion of \textit{cost} purposefully ambiguous here, as Ambainis's construction allows for defining any notion of cost associated with providing uniform access to the algorithms, i.e., the ability to apply the gate that is executed at any given time step in any of the algorithms. The resulting algorithm depends on the notion of cost selected, and from Ambainis's construction, it is not apparent how one would obtain the claimed scaling in multiple notions of cost simultaneously. Moreover, Ambainis's construction assumes that all instance-independent gates, i.e., all operations that are not part of the original algorithms, have cost zero, which means that a proper analysis of the time complexity of the resulting algorithm evaluating $f$ is lacking.

Our result improves on Ambainis's result in the following manner. If the $n$ original algorithms have query complexity $S_1, \dots, S_n$, time complexity $T_1, \dots, T_n$, and we have efficient uniform access to them, then we can evaluate $f$ with bounded error with $\widetilde{\O}(\sqrt{\sum_{i=1}^n S_i^2})$ queries and $\widetilde{\O}(\sqrt{\sum_{i=1}^n T_i^2})$ gates. Moreover, the number of auxiliary qubits introduced is at most polylogarithmic in $T_{\max} = \max_{i \in [n]} T_i$ and $n$. Thus, we achieve the desired scaling in the query and time complexities simultaneously, while also counting all instance-independent gates in our analysis of the time complexity of the resulting algorithm.

We achieve this result by converting the original algorithms into span programs, which we subsequently compose using techniques from \cite{Rei09}. We turn the resulting composed span program back into an algorithm, reusing some ideas from \cite{Amb06}, and using our technique of implementing subspaces.

Perhaps the most interesting future direction suggested by our work is to find new algorithm composition results by turning algorithms into span programs, taking advantage of the relative ease of span program composition, and then converting the result back into an algorithm.

\vspace{-1em}

\paragraph{Organization.} In \cref{sec:preliminaries}, we introduce the model of quantum query algorithms, our precise model for accessing individual gates of given subroutines, and the span program model. In \cref{sec:span-program-alg}, we describe how to compile a span program into a quantum algorithm, and reduce the problem of analyzing the time complexity of this algorithm to implementing and analyzing three subroutines. In \cref{sec:main_thm}, we describe our construction for turning an algorithm into a span program and show how to implement the resulting algorithm time-efficiently. Finally, in \cref{sec:variable_time_search}, we show how to combine the construction of \cref{sec:main_thm} with a span program composition to improve Ambainis's variable-time quantum search result.

\section{Preliminaries}
\label{sec:preliminaries}

\subsection{Quantum query algorithms}
\label{sec:quant_algorithm}

Let $n \in \N$, $X \subseteq \{0,1\}^n$ and $f : X \to \{0,1\}$ be a (partial) Boolean function. We study quantum algorithms that compute (or decide) the value of $f(x)$ given quantum query access to individual bits of the input $x \in X$.

Let $\A$ be a quantum algorithm that acts on a state space $\C^{[n] \times \W}$, where $[n] := \{1,\dotsc,n\}$ and $\W$ is a finite set that labels the workspace states. We denote the initial state of $\A$ by $\ket{\Psi_0} \in \C^{[n] \times \W}$ and the unitary transformations that $\A$ applies to the state space at the respective time steps by $U_1,\dots,U_T\in \U(\C^{[n] \times \W})$, where $T \in \N$ is the total number of time steps.

The algorithm $\A$ makes queries to an input string $x \in \{0,1\}^n$ by having a subset of the unitaries be (controlled) calls to an oracle $\O_x \in \U(\C^{[n] \times \W})$ defined by its action on the computational basis as
\begin{equation}
	\forall i \in [n], \forall j \in \mathcal{W}, \qquad \O_x : \ket{i,j} \mapsto (-1)^{x_i}\ket{i,j},
	\label{eq:Ox}
\end{equation}
where the two registers correspond to the input bit index and the workspace, respectively. The only dependence on $x$ of the unitaries that make up $\A$ is through some $U_t$'s being $\O_x$. We denote the set ${\cal S}\subset [T]$ to be the set that contains all $t\in {\cal S}$, such that $U_t = {\cal O}_x$. Then $S=|{\cal S}|$ is the query complexity of $\A$.

In the standard definition of a quantum query algorithm, every second unitary is a query, so ${\cal S}$ would be the set of odd indices. This is appropriate when we are only interested in the query complexity of the algorithm, since we can combine any consecutive non-query unitaries into a single unitary. However, since we are also interested in the time complexity, we want to restrict the non-query unitaries to some universal gate set. Thus, we do not assume that every other unitary is a query, and we explicitly allow for sequences of non-query unitaries between any two queries, as well as at the beginning and the end of the algorithm.

We take the initial state to be a computational basis state. We can assume that $U_1$ and $U_T$ are not queries without loss of generality. Indeed, if the first unitary is a query, then it only introduces a global phase and hence it is redundant. Similarly, we assume that any measurement at the end of the algorithm is a computational basis measurement, which implies that if $U_T$ is a query, then it is also redundant as it does not influence the measurement probabilities. Finally, we also assume without loss of generality that no two consecutive time steps are query time steps, as then the resulting operation on the state space would reduce to $\O_x^2 = \I$, rendering both queries redundant.

For every $x \in \{0,1\}^n$ we define the state of the system at time $t \in [T]_0:=\{0,\dots,T\}$ on input $x$ as
\begin{equation}
	\ket{\Psi_t(x)} := U_tU_{t-1} \cdots U_1\ket{\Psi_0},
	\label{eq:Psi}
\end{equation}
where $\ket{\Psi_0} \in \C^{[n] \times \W}$ is the initial state. Note that the right-hand side of \cref{eq:Psi} has an implicit dependence on $x$, since for some indices $t$, $U_t = \O_x$.

We can assume that there is a single-qubit answer register used to indicate the output of the computation. If $\Pi_b$ denotes the orthogonal projector onto states with $\ket{b}$ in the answer register, for $b\in\{0,1\}$, then $p_b(x):=\norm{\Pi_b\ket{\Psi_T(x)}}^2$ is the probability that the algorithm outputs $b$ on input $x$. We say that $\A$ computes a function $f:X\rightarrow\{0,1\}$, where $X\subseteq \{0,1\}^n$, with error probability $\varepsilon\in [0,1/2)$ if $p_{1-f(x)}(x)\leq \varepsilon$ for all $x\in X$.

In addition to the standard assumptions that we outlined above, we will also make some non-standard assumptions on the structure of quantum query algorithms. We refer to the query algorithms that satisfy both the standard and the non-standard assumptions as \textit{clean algorithms}. The formal definition is given in \cref{def:QA}, and we spend the rest of this section claiming that we can assume without loss of generality that every quantum query algorithm is clean.

\begin{definition}[Clean quantum algorithm]
\label{def:QA}
	Let $\A$ be a quantum query algorithm acting on $\C^{[n] \times \W} = \C^{[n] \times \W' \times \{0,1\}}$ with the last register being the answer register. Suppose that the time complexity of $\A$ is $T$, the query complexity is $S$, and the initial state has $\ket{0}$ in the answer register, so it can be expressed as $\ket{\Psi_0}=\ket{\psi_0}\ket{0}$ for some $\ket{\psi_0} \in \C^{[n] \times \W'}$. Define the \emph{final accepting state} as $\ket{\Psi_T} := \ket{\psi_0}\ket{1}$.	$\A$ is a \emph{clean quantum algorithm} if it satisfies the following properties.
	\begin{enumerate}
		\setlength\itemsep{-.3em}
		\item \emph{Consistency:} For all inputs $x \in \{0,1\}^n$,
		\begin{align*}
			\braket{\Psi_T}{\Psi_T(x)}&=p_1(x), \quad\mbox{and}\quad \bra{\Psi_T}(I\otimes X)\ket{\Psi_T(x)}=p_0(x),
		\end{align*}
		where $p_b(x)=\norm{(I\otimes\ket{b}\bra{b})\ket{\Psi_T(x)}}^2$ is the probability that $\A$ outputs $b$ on input $x$,
		and $X$ denotes the Pauli matrix implementing the logical $\mathrm{NOT}$.
		\item \emph{Commutation:} $(I\otimes X)$ commutes with every unitary $U_t$ of the algorithm, where $X$ acts on the answer register.
		\item \emph{Query-uniformity:} Two consecutive queries are not more than $\lfloor 3T/S \rfloor$ time steps apart, and the first and last queries are separated by at most $\lfloor 3T/S\rfloor$ time steps from the start and the finish of the algorithm, respectively.
	\end{enumerate}
\end{definition}

We proceed by showing that restricting our attention to clean algorithms only incurs a constant multiplicative overhead in the query and time complexities and constant additive overhead in the space complexity.

We prove this in two steps. First, we show that we can satisfy conditions 1 and 2 by modifying the algorithm in the following sense: we first run it once, then we copy out the answer register, and subsequently, we run it backwards. This constitutes \cref{lem:QA}. After that, we insert some queries and identity gates into the resulting algorithm, such that we in addition also satisfy condition 3, which is the objective of \cref{lem:Q_Uniformity}.

\begin{lemma}\label{lem:QA}
	Fix $f : X \subseteq \{0,1\}^n \to \{0,1\}$. Let $\A$ be a quantum query algorithm with initial state $\ket{\Psi_0} \in \C^{[n] \times \W}$, unitaries $U_1, \dots, U_T \in \U(\C^{[n] \times \W})$, time complexity $T$ and query complexity $S$ and suppose that it computes $f$ with error probability $\varepsilon > 0$. Now, let $\A'$ be a quantum algorithm acting on $\C^{[n] \times \W \times \{0,1\}}$, with initial state $\ket{\Psi_0'} = \ket{\Psi_0} \ket{0}$ and consisting of the following sequence of unitaries:
	\[
		(U_1^{\dagger} \otimes \I)	\cdots
		(U_T^{\dagger} \otimes \I)
		(\I \otimes \mathrm{CNOT})
		(U_T \otimes \I)	\cdots
		(U_1 \otimes \I),
	\]
	where the $\mathrm{CNOT}$ is a controlled-not gate with the answer qubit of $\A$ acting as control qubit and the last qubit of $\A'$ acting as the target. Then $\A'$ fulfills conditions 1 and 2 in \cref{def:QA} with final accepting state $\ket{\Psi'_{T'}} = \ket{\Psi_0} \ket{1}$, time complexity $T' = 2T+1 = \Theta(T)$, query complexity $S' = 2S = \Theta(S)$, uses one more qubit than $\A$ and evaluates $f$ with error probability $\varepsilon$.
\end{lemma}

\begin{proof}
	Since an $X$-gate on the target qubit of a $\mathrm{CNOT}$ gate commutes with the $\mathrm{CNOT}$-gate itself, we find that all operations in $\A'$ commute with $I \otimes X$, thus the commutation condition is fulfilled.

	Next we check the consistency condition. To that end, we let $\ket{\Psi_T(x)} = \ket{\Phi_0(x)} + \ket{\Phi_1(x)}$, where $\ket{\Phi_b(x)} = \Pi_b\ket{\Psi_T(x)}$ is the projection of $\ket{\Psi_T(x)}$ onto the part of the state with $\ket{b}$ in the answer register of $\A$. Then the state of $\A'$ after $T$ steps on input $x$ is
	\begin{align*}
		\ket{\Psi_T'(x)} &= \ket{\Psi_T(x)} \ket{0} = \ket{\Phi_0(x)}\ket{0}+\ket{\Phi_1(x)}\ket{0}
	\end{align*}
	and the state of $\A'$ after $T+1$ steps on input $x$ is
	\begin{align*}
		\ket{\Psi_{T+1}'(x)} &= \mathrm{CNOT}\ket{\Psi_T'(x)} = \ket{\Phi_0(x)}\ket{0} + \ket{\Phi_1(x)}\ket{1}.
	\end{align*}
	Let $U_{\A} = U_T \cdots U_1$, so that $\ket{\Psi_T(x)}=U_{\A}\ket{\Psi_0}$, and
	\begin{align}
		\ket{\Psi_{2T+1}'(x)}&=(U_{\A}^\dagger\otimes \I)\ket{\Psi_{T+1}'(x)} = (U_{\A}^{\dagger}\ket{\Phi_0(x)})\ket{0}+(U_{\A}^{\dagger}\ket{\Phi_1(x)})\ket{1}.\label{eq:QA2}
	\end{align}
	Since $U_{\A}^\dagger \ket{\Phi_b(x)} = U_{\A}^\dagger \Pi_b \ket{\Psi_T(x)}$, for $b \in \{0,1\}$,
	the success probability of $\A'$ is equal to the success probability of $\A$:
	\[\norm{(\I \otimes \ket{b}\bra{b})\ket{\Psi_{2T+1}'(x)}}^2 = \norm{U_\A^{\dagger} \Pi_b \ket{\Psi_T(x)}}^2 = \norm{\Pi_b\ket{\Psi_T(x)}}^2 = p_b(x).\]
	Moreover, from \cref{eq:QA2}, we have for all $b \in \{0,1\}$,
	\begin{align*}
		\braket{\Psi_0,b}{\Psi'_{2T+1}(x)} &
		= \bra{\Psi_0} U_{\A}^\dagger \Pi_b U_{\A}\ket{\Psi_0}
		= \norm{\Pi_b U_{\A}\ket{\Psi_0}}^2
		= \norm{\Pi_b \ket{\Psi_T(x)}}^2 = p_b(x).
	\end{align*}
	In particular that implies that
	\[\braket{\Psi'_{T'}}{\Psi_{2T+1}'(x)} = \braket{\Psi_0,1}{\Psi_{2T+1}'(x)} = p_1(x) \quad \text{and} \quad \bra{\Psi'_{T'}}(\I \otimes X)\ket{\Psi_{2T+1}'(x)} = \braket{\Psi_0,0}{\Psi_{2T+1}'(x)} = p_0(x).\]
	Hence, $\A'$ satisfies the consistency condition as well.
\end{proof}

\begin{lemma}\label{lem:Q_Uniformity}
	Fix $f : X \subseteq \{0,1\}^n \to \{0,1\}$. Let $\A$ be a quantum query algorithm with time complexity $T$ and query complexity $S$ that computes $f$ with error probability $\varepsilon > 0$. Then, there exists an algorithm $\A'$ with time complexity $T' = \Theta(T)$ and query complexity $S' \leq 3S$ such that two consecutive queries are no  more than $\lfloor 3T'/S' \rfloor$ times steps apart. In addition if $\A$ fulfills conditions 1 and 2 in \cref{def:QA}, then $\A'$ is a clean quantum algorithm evaluating $f$ with error $\varepsilon$.
\end{lemma}

\begin{proof}
	First, if $S \in \{1,2\}$, we note that $\lfloor 3T/S \rfloor > T$, and hence the third condition in \cref{def:QA} is trivially satisfied without any modifications to $\A$. Hence, we restrict to the case where $S \geq 3$. We insert a sequence of operations $I\O_xI\O_xI$ into $\A$ between time steps $\lceil kT/S\rceil$ and $\lceil kT/S\rceil + 1$ where $k \in [S-1]$. This increases the number of queries to $S' \leq 3S$ and the number of time steps to $T' = T + 5(S-1)$. The number of time steps between two consecutive queries is at most
	\[\left\lceil\frac{T}{S}\right\rceil + 2 \leq \frac{T}{S} + 3 = \frac{T' - 5(S-1)}{S} + 3 \leq 3\frac{T'}{S'} - 5 + \frac{5}{S} + 3 < 3\frac{T'}{S'}.\]
	As the left-hand side is an integer, we can just as well take the floor on the right-hand side. Similarly, the distance of the first query from the start is at most $\lceil T/S\rceil + 1 < 3T'/S'$, and the number of time steps between the last query and the end of the algorithm is at most $T - \lceil (S-1)T/S\rceil + 1 \leq T/S + 1 < 3T'/S'$. Thus, we have satisfied the query-uniformity condition from \cref{def:QA}.

	Furthermore, the second statement follows immediately from the fact that the unitaries that we are inserting amount to the identity, and hence if $\A$ evaluates $f$ with error probability $\varepsilon$, so does $\A'$. This completes the proof.
\end{proof}

By \cref{lem:QA} and \cref{lem:Q_Uniformity}, we can assume without loss of generality that any quantum algorithm is a clean quantum algorithm, namely, that it does a computation, copies out the answer, and then reverses the computation. The overhead of putting an algorithm into this form is only a constant factor in the query and time complexity, and a single auxiliary qubit in the space complexity.

For clarity, we emphasize that in a clean quantum algorithm with non-zero error, while in some sense the algorithm uncomputes everything but the answer, this uncomputation does not succeed fully -- we do not return the non-answer registers of the algorithm to the fixed state $\ket{\Psi_0}$. The weight of the final state $\ket{\Psi_T(x)}$ on $\ket{\Psi_0}$ in the non-answer registers is
\begin{align*}
	|\braket{\Psi_0,0}{\Psi_T(x)}|^2 + |\braket{\Psi_0,1}{\Psi_T(x)}|^2 = p_0(x)^2 + p_1(x)^2 = p_0(x)^2 + (1-p_0(x))^2,
\end{align*}
which is strictly less than 1 whenever $0 < p_0(x) < 1$.

\subsection{Accessing an algorithm as input}
\label{sec:model}

Throughout the rest of the paper, we will consider algorithms that, among other things, take other algorithms as input. This section concerns how we model this through several oracles. The model is essentially a generalization of the one used in \cite{Amb06}.

Let $m \in \N$ and let $\A = \{\A^{(1)}, \dotsc, \A^{(m)}\}$ be a set of quantum query algorithms. For every $j \in [m]$, let $T^{(j)}$ be the time complexity of $\A^{(j)}$, let $\S^{(j)} \subseteq [T^{(j)}]$ be the set of time steps at which $\A^{(j)}$ performs queries to the input, let $U_1^{(j)}, \dots, U_{T^{(j)}}^{(j)}$ be the sequence of unitaries in $\A^{(j)}$, and suppose that $\A^{(j)}$ evaluates a function $f_j : X^{(j)} \subseteq \{0,1\}^{n^{(j)}} \to \{0,1\}$ with bounded error. For convenience we define $T_{\max} = \max_{j \in [m]} T^{(j)}$ and $n_{\max} = \max_{j \in [m]} n^{(j)}$, and we assume that all unitaries $U_t^{(j)}$ act on some space $\C^{[n_{\max}] \times \W}$, where the first register is large enough to hold the input bit label for any of the Boolean functions $f_j$.

We define three different oracles associated with $\A$. First, the \textit{algorithm oracle}, sometimes referred to as \textbf{Select}, acts on $\C^{[m] \times [T_{\max}] \times [n_{\max}] \times \W}$ as
\[\forall j \in [m], t \in [T^{(j)}] \setminus \S^{(j)}, \ket{\psi} \in \C^{[n_{\max}] \times \W}, \qquad \O_\A : \ket{j}\ket{t}\ket{\psi} \mapsto \ket{j}\ket{t}U_t^{(j)}\ket{\psi}.\]
Second, the \textit{query time step oracle}, which allows us to determine whether a given algorithm $\A^{(j)}$ perform a query at a given time step $t$, acts on $\C^{[m] \times [T_{\max}]}$ as
\[\forall j \in [m], t \in [T^{(j)}], \qquad \O_\S : \ket{j}\ket{t} \mapsto \begin{cases}
	-\ket{j}\ket{t}, & \text{if } t \in \S^{(j)}, \\
	\ket{j}\ket{t}, & \text{otherwise}.
\end{cases}\]
Finally, given a list of inputs $x = (x^{(1)}, \dots, x^{(m)})$, where $x^{(j)} \in \{0,1\}^{n^{(j)}}$ is the input to function $f_j$, the \textit{input oracle} to $x$ acts on $\C^{[m] \times [n_{\max}]}$ as
\[\forall j \in [m], \qquad \O_x = \sum_{j=1}^m \ket{j}\bra{j} \otimes \O_{x^{(j)}}, \qquad \text{where} \qquad \forall i \in [n_j], \qquad \O_{x^{(j)}} : \ket{i} \mapsto (-1)^{x_i^{(j)}}\ket{i}.\]
On computational basis states that are not specified above, the behavior of the three oracles can be arbitrary.

By saying that we have \textit{uniform access} to the set of algorithms $\A$, we mean that we have access to these three oracles $\O_\A$, $\O_\S$ and $\O_x$. Moreover, if the time complexity of implementing the oracles $\O_\A$ and $\O_\S$ is polylogarithmic in $T_{\max} = \max_{j \in [m]} T^{(j)}$ and $m$, then we say that we have \textit{efficient uniform access} to $\A$.

Note that if $m = 1$, then the first register in all above oracles only contains one dimension and hence can be omitted. In that case, we drop all the superscripts and $\O_x$ reduces to the regular input oracle $\O_x$ that we defined in \cref{eq:Ox}.

These oracles fully capture the set $\A$ and provide an interface for the higher-level algorithms to execute the algorithms in $\A$ as subroutines. From a computer science point of view, one can think about these oracles as the endpoints for the user. To use our results for a particular set of algorithms $\A$, one has to provide implementations of these three oracles. The machinery we develop in the remainder of this text then takes care of the rest of the construction, and our analysis provides the number of calls made to these oracles, alongside with the number of extra gates used.

A natural question to ask is how difficult it is in general to implement these oracles. If the algorithms from $\A$ are very unstructured, then it is in general very time-consuming to implement these oracles. In that case, one could implement $\O_\A$ and $\O_\S$ by querying a quantum read-only random access memory (commonly referred to as QRAM) storing the algorithms $\A^{(j)}$ as lists of gates. A similar model, called \emph{quantum random access stored-program machines} was recently formalized in~\cite{WY20}.

However, quantum query algorithms that we encounter in practice can usually be described very succinctly, and we have some efficient constructive procedure to calculate what gate has to be applied in the $j$th algorithm at the $t$th time step and at what time steps the algorithms perform a query. These procedures can be used to implement the oracles $\O_\A$ and $\O_\S$ efficiently and provide us with efficient uniform access. Through a similar argument, one can usually provide an efficient implementation of $\O_x$ as well, based on the individual implementations of the $\O_{x^{(j)}}$'s. All of these constructions are always instance-dependent, though, and hence we cannot elaborate on them further without losing generality.

We conclude this section by remarking that this final argument is more generally applicable to oracular algorithms. The results about query complexity are in general most interesting and applicable in a setting where the oracles themselves can be substituted by efficient algorithms. The same goes for the uniform access model we consider here.

\subsection{Span programs}

Having discussed the general structure of quantum query algorithms, let us turn our attention to the other construct of interest in this work, namely, span programs. Following~\cite{IJ15}, we define a span program for evaluating a Boolean function as follows.

\begin{definition}[Span program]
A \emph{span program} $P=(\H,\V,A,\ket{\tau})$ on $\{0,1\}^n$ consists of
\begin{enumerate}
\setlength\itemsep{-.4em}
\item a finite-dimensional Hilbert space $\H$ that decomposes as $$\H=\H_1\oplus\cdots\oplus \H_n\oplus \Htrue \oplus \Hfalse$$ where each $\H_i$, $i\in[n]$, decomposes further as $\H_i = \H_{i,0} \oplus \H_{i,1}$,
\item a finite-dimensional Hilbert space $\V$,
\item a linear operator $A\in \L(\H,\V)$, and
\item a target vector $\ket{\tau} \in \V$.
\end{enumerate}
\end{definition}

With each string $x \in \{0,1\}^n$, we associate the subspace
\[\H(x) = \H_{1,x_1}\oplus\cdots\oplus \H_{n,x_n}\oplus \Htrue.\]
For any subspace $\H' \subseteq \H$, we write $\Pi_{\H'} \in \L(\H)$ to denote the projector onto $\H'$.

Intuitively, a span program encodes the decision problem ``Is $\ket{\tau}\in A\H(x)$?''. To answer this question in the affirmative, it is sufficient to provide a preimage of $\ket{\tau}$ under $A$ in $\H(x)$, called a \emph{positive witness}. In the negative case, one would like to find an object, called a \emph{negative witness}, that precludes the existence of such a positive witness. These concepts are defined rigorously as follows.

\begin{definition}[Positive and negative witnesses]
Fix a span program $P=(\H,\V,A,\ket{\tau})$ and an input $x \in \{0,1\}^n$. We call a vector $\ket{w} \in \H$ a \emph{positive witness} for $x$ if $\ket{w} \in \H(x)$, and $A\ket{w} = \ket{\tau}$. The \emph{positive witness size} of $x$ is
\[w_+(x,P)=w_+(x):=\min_{\ket{w}\in \H(x)}\{\|\ket{w}\|^2: A\ket{w}=\ket{\tau}\},\]
if there exists a positive witness for $x$, and $w_+(x) = \infty$ otherwise. We say that $\ket{\omega} \in \V$ is a \emph{negative witness} for $x$ if $\bra{\omega} A \Pi_{\H(x)} = 0$ and $\braket{\omega}{\tau} = 1$. The \emph{negative witness size} of $x$ is
\[w_-(x,P)=w_-(x):=\min_{\ket{\omega}\in\V}\{\|\bra{\omega}A\|^2: \bra{\omega} A\Pi_{\H(x)}=0, \braket{\omega}{\tau}=1\},\]
if there exists a negative witness, and $w_-(x)=\infty$ otherwise. We define the set of \emph{positive} and \emph{negative inputs} of $P$, respectively, as
\begin{align*}
  P_1&:=\{x \in X: w_+(x)<\infty\},&
  P_0&:=\{x \in X: w_-(x)<\infty\}.
\end{align*}
\end{definition}

One can think of a span program as a puzzle in which several pieces are supplied (the space $\H(x)$) together with assembly instructions (the map $A$) and a contour of a shape to be constructed (the target~$\ket{\tau}$). The larger the number of pieces required to construct the target, the harder it is to solve the puzzle. Alternatively, the larger the number of missing pieces required to declare the problem unsolvable, the harder it is to do so. This justifies the notion that larger witness sizes are indicative of harder span programs.

However useful these definitions are, requiring that the puzzle be solved exactly, without using any ``illegal'' pieces, can be too hard a constraint. As \cite{IJ15} illustrates, it can be advantageous to relax the constraints and simply require that the target be constructed with as few elements outside $\H(x)$ as possible, or that the negative witness $\bra{\omega} A$ overlaps with $\H(x)$ as little as possible. In this paper, we will use the second concept, in a relaxed form from \cite{Jef20}. For some $x\in \{0,1\}^n$, not necessarily in $P_0$, we will say that $\ket{\widetilde\omega}\in\V$ is an \emph{approximate negative witness} if  $\braket{\widetilde{\omega}}{\tau}=1$, and $\norm{\bra{\widetilde\omega}A\Pi_{\H(x)}}^2$, called the error of $\ket{\widetilde\omega}$, is not too large. If $x\in P_0$, then this quantity can be 0, but otherwise not. In particular, if $x\in P_1$, the minimum possible error of $\ket{\widetilde\omega}$ is $\frac{1}{w_+(x,P)}$ \cite[Theorem~9]{IJ15}. This is why inverse positive witness size is used as a point of reference for what constitutes ``small'' error for a negative witness in the following definition.

\begin{definition}[Span program complexity]
Let $P$ be a span program, $f : X \subseteq \{0,1\}^n \to \{0,1\}$ a Boolean function, and $\lambda \in [0,1)$. The \emph{positive} and \emph{approximate negative complexity} of $P$ (w.r.t. $f$) are, respectively,
\begin{align*}
	W_+(f,P) &= W_+(P) := \max_{x\in f^{-1}(1)} w_+(x,P), &
	\widetilde{W}_-(f,P) &= \widetilde{W}_-(P) := \max_{x\in f^{-1}(0)} \widetilde{w}_-(x,P),
\end{align*}
where $\widetilde{w}_-(x,P)$ is the following minimization over all \emph{approximate negative witnesses} $\ket{\widetilde\omega}\in\V$:
\begin{equation*}
  \widetilde{w}_-(x,P) := \min_{\ket{\widetilde\omega}\in\V}
	\left\{
	  \norm{\bra{\widetilde\omega}A}^2 :
		\braket{\widetilde\omega}{\tau} = 1,
		\norm{\bra{\widetilde\omega}A\Pi_{\H(x)}}^2 \leq \lambda/W_+(f,P)
	\right\}.
\end{equation*}
We say that $P$ \emph{positively $\lambda$-approximates} $f$ if $f^{-1}(1)\subseteq P_1$ and $\widetilde{w}_-(x,P) < \infty$ for all $x \in f^{-1}(0)$.
The \emph{complexity} of $P$ with respect to $f$ is
\begin{equation*}
	C(P) := \sqrt{W_+(f,P)\widetilde{W}_-(f,P)},
\end{equation*}
leaving $\lambda$ and $f$ implicit.
We say that \emph{$P$ decides $f$ exactly} if it $0$-approximates $f$.
\end{definition}

In other words, in a positively approximating span program, negative instances (w.r.t.\ $f$) are distinguished from positive instances by having approximate negative witnesses with smaller error. Since the minimum possible error of any approximate negative witness for $x\in f^{-1}(1)$ is $\frac{1}{w_+(x,P)}\geq \frac{1}{W_+(f,P)}$, we can use approximate negative witnesses with error strictly smaller than this as proof of membership in $f^{-1}(0)$. The ``gap'' in error between 1- and 0-inputs is characterized by the parameter $\lambda$ and will later be exploited to construct quantum algorithms that evaluate span programs.

\section{The time complexity of implementing a span program}
\label{sec:span-program-alg}

Span programs by themselves are not quantum objects. Nevertheless, the elements that define a span program can be combined to form a quantum algorithm. In this section, we describe such an algorithm and consider its implementation and time complexity in a general setting. Specifically, we describe its time complexity in terms of the time complexities of several operations. We do not show how to implement these in general -- such details will depend on the specific construction of the span program.

To turn a span program into a quantum algorithm, we first have to normalize it. This process is explained in \cref{sec:normalization}. From such a normalized span program, we can construct a span program algorithm that distinguishes between positive and negative instances, which is explained in \cref{sec:sp-alg}. Finally, in \cref{sec:implementing-subspace}, we analyze the time complexity of this span program algorithm, where we use the notion of implementing subspaces. The techniques from sections \cref{sec:normalization} and \cref{sec:sp-alg} are mainly taken from \cite{IJ15}, and our contribution starts in \cref{sec:implementing-subspace}.

\subsection{Span program normalization}
\label{sec:normalization}

Let $P = (\H,\V,A,\ket{\tau})$ be a span program. The \emph{minimal positive witness} of $P$ is defined as $\ket{w_0} = A^+\ket{\tau}$, where $A^+$ is the Moore-Penrose inverse of $A$. In other words, $\ket{w_0}$ is the shortest vector in $\H$ that is mapped to $\ket{\tau}$ by $A$. When $\norm{\ket{w_0}} = 1$, we say that the span program $P$ is \emph{normalized}.

In \cref{sec:sp-alg}, we will describe an algorithm that can be derived from a span program, first described in \cite{IJ15} (similar algorithms were given previously in \cite{Rei09}). This algorithm assumes that a span program is normalized -- in fact, $\ket{w_0}$ will be the algorithm's initial state. The process of converting any span program into a normalized span program is what we call normalization.

One very naive way of normalizing a span program is simply by rescaling $\ket{\tau}$. The minimal positive witness $A^+\ket{\tau}$ scales accordingly, and hence it is easy enough to obtain a normalized span program. However, scaling $\ket{\tau}$ with a factor $\alpha$ also scales the positive witness size $W_+$ by a factor of $\alpha$, and the negative witness size $W_-$ by a factor of $1/\alpha$. For reasons that will become apparent in the next section, the complexity of the span program scales most favorably if $W_+$ is $\O(1)$.

So, ideally, we would like to normalize our span program in such a way that we also modify our positive witnesses to have constant size. For this, we need a slightly more involved construction than simply scaling $\ket{\tau}$. This is the objective of the following theorem.

\begin{restatable}[Span program rescaling]{theorem}{thmscaling}
	\label{thm:scaling}
	Let $P=(\H,\V,A,\ket{\tau})$ be any span program on $\{0,1\}^n$ that positively $\lambda$-approximates $f$ for some $\lambda\in (0,1)$, and let $N=\|\ket{w_0}\|^2$ where $\ket{w_0}$ is the minimal positive witness of $P$. For $\beta > 0$, define the span program $P^\beta=(\H^\beta,\V^\beta,A^\beta,\ket{\tau^\beta})$ as follows:
	\begin{align}
		\forall i\in[n], b\in\{0,1\},\qquad \H^\beta_{i,b}=\H_{i,b}, \quad \Htrue^\beta=\Htrue\oplus\Span\{\ket{\hat1}\},& \quad \Hfalse^\beta=\Hfalse\oplus\Span\{\ket{\hat0}\},\nonumber \\
		\V^\beta=\V\oplus\Span\{\ket{\hat1}\}, \quad A^{\beta} = \beta A + \ket{\tau}\bra{\hat{0}} + \frac{\sqrt{\beta^2 + N}}{\beta}\ket{\hat{1}}\bra{\hat{1}},& \quad \ket{\tau^\beta} = \ket{\tau}+\ket{\hat{1}},\label{eq:VAt}
	\end{align}
	where $\ket{\hat0}$ and $\ket{\hat{1}}$ are unit vectors orthogonal to $\H$ and $\V$.
	Then the minimal witness of $P^\beta$ is
	\begin{equation}
	  \ket{w_0^\beta}=\frac{\beta}{\beta^2+N}\ket{w_0}+\frac{N}{\beta^2+N}\ket{\hat0}+\frac{\beta}{\sqrt{\beta^2+N}}\ket{\hat1},
		\label{eq:wb}
	\end{equation}
	and $\|\ket{w_0^\beta}\|=1$. Furthermore, if $\beta\geq \sqrt{W_+(P)}$ and for all $x\in f^{-1}(0)$ there is an approximate negative witness $\ket{\tilde\omega}\in\V$ with complexity $\norm{\bra{\tilde\omega}A}^2\leq \widetilde{W}_-(P)$ and error $\norm{\bra{\tilde\omega}A\Pi_{\H(x)}}^2\leq\lambda/\beta^2$,
	then $P^\beta$ positively $2\lambda$-approximates $f$ with $W_+(P^{\beta})\leq 2$, and $\widetilde{W}_-(P^{\beta})\leq \beta^2\widetilde{W}_-(P)+2$.
\end{restatable}

This theorem is a slight generalization of the normalization procedure outlined in \cite{IJ15}. The proof of \cref{thm:scaling} appears in \cref{app:scaling}. Note that indeed $P^{\beta}$ is a normalized span program, and if $\beta$ is sufficiently large we also have $W_+(P^{\beta}) = \O(1)$.

For technical reasons, we assume in the remainder of this text that we can reflect through the states $\ket{\hat{0}}$ and $\ket{\hat{1}}$ in $\O(1)$ gates and with only $\O(1)$ auxiliary qubits. One way to implement this is to make the state space of the system equal to $\H \otimes \C^{\{0,1,2\}}$, identify $\H$ with $\H \otimes \ket{2}$, and define $\ket{\hat{0}} = \ket{0,0}$ and $\ket{\hat{1}} = \ket{0,1}$. Now, we leave it for the reader to check that the unitary $\I_{\H} \otimes (2\ket{0}\bra{0} - \I_{\C^{\{0,1,2\}}})$ acts as $2\ket{\hat{0}}\bra{\hat{0}} - \I_{\H^{\beta}}$ on $\H^{\beta}$, and similarly for $\ket{\hat{1}}$. Moreover, these unitaries can be implemented with $\O(1)$ gates, and $\O(1)$ extra qubits.

\subsection{The span program algorithm}
\label{sec:sp-alg}

From the previous section we know how to convert a span program $P$ that positively $\lambda$-approximates a function $f$ into a normalized span program $P^{\beta}$ that $2\lambda$-approximates $f$. We proceed with explaining how one can convert this normalized span program $P^{\beta}$ into a quantum algorithm that evaluates $f$ with bounded error.

To that end, we define its \emph{span program unitary}, dependent on the input $x \in \{0,1\}^n$, as
\begin{equation}
	U=U(P^\beta,x)=(2\Pi_{\H^\beta(x)}-\I)\Bigl(2\bigl(\Pi_{\ker A^\beta}+\ket{w_0^\beta}\bra{w_0^\beta}\bigr)-\I\Bigr).
	\label{eq:U}
\end{equation}
The span program algorithm of \cite{IJ15} works by doing phase estimation of $U$ on initial state $\ket{w_0^\beta}$ to precision $\Theta$ and error probability at most $\eps$, and then estimating the amplitude of this process on a 0 in the phase register, using precision $\Theta'$, where
\[\Theta = \Omega\left(\sqrt{\frac{1-2\lambda}{W_+(P^{\beta})\widetilde{W}_-(P^{\beta})}}\right), \qquad \varepsilon = \Omega\left(\frac{1-2\lambda}{W_+(P^{\beta})}\right), \qquad \text{and} \qquad \Theta' = \Omega\left(\frac{1-2\lambda}{\sqrt{W_+(P^{\beta})}}\right).\]
Moreover, if we only have access to upper bounds on $W_+(P^{\beta})$, $\widetilde{W}_-(P^{\beta})$ and $\lambda$, then it suffices to use those instead of the true values of these quantities and update the above expressions accordingly. Thus, this algorithm requires constructing the state $\ket{w_0^{\beta}}$, and then making $\O(\frac{1}{\Theta'}\frac{1}{\Theta}\log\frac{1}{\eps})$ controlled calls to $U$. If we choose $\beta = \sqrt{W_+(P)}$, we find that $W_+(P^\beta) \leq 2$, and $\widetilde{W}_-(P^\beta) \leq W_+(P)\widetilde{W}_-(P) + 2$, so that
\begin{equation}
	\frac{1}{\Theta'}\frac{1}{\Theta}\log\frac{1}{\eps}=\O\left(\frac{W_+(P^\beta)\sqrt{\widetilde{W}_-(P^\beta)}}{(1-2\lambda)^{3/2}}\log\frac{W_+(P^\beta)}{1-2\lambda}\right)=\O\left(\frac{\sqrt{W_+(P)\widetilde{W}_-(P)}}{(1-2\lambda)^{3/2}}\log\frac{1}{1-2\lambda}\right).
	\label{eq:span-query-complexity}
\end{equation}
Similarly as above, if we only know upper bounds on $W_+(P)$, $W_-(P)$ and $\lambda$, we can substitute those both in the choice for $\beta$ and in the above expression. Furthermore, it is important to remark that the relevant limits in the big-$\O$-notation are when $\lambda \uparrow \frac12$ and $W_+(P), \widetilde{W}_-(P) \to \infty$. For a more detailed description of this algorithm, see \cite{IJ15}.

We refer to this algorithm as \emph{the algorithm compiled from $P$}, or the \textit{span program algorithm for $P$}. Its query complexity is straightforward to analyze. The initial state $\ket{w_0^\beta}$ is independent of the input, so it can be generated in 0 queries. Similarly, since $A^\beta$ does not depend on the input, the reflection $2(\Pi_{\ker A^\beta}+\ket{w_0^\beta}\bra{w_0^\beta})-\I$ can be implemented in 0 queries. The reflection $2\Pi_{\H^\beta(x)}-\I$ does depend on the input, but it can be implemented in 2 queries, see \cite[Lemma~2]{IJ15}. Thus, the query complexity of the algorithm compiled from $P$ is given by \cref{eq:span-query-complexity}.

\subsection{Time complexity of the span program algorithm}
\label{sec:implementing-subspace}

We now turn our attention to the time complexity of the algorithm span program algorithm for $P$. We express this time complexity in terms of the number of calls we perform to some black-box operations that can be defined directly in terms of the span program $P$, see \cref{thm:span-program-time}. In principle, the time complexity of any span program algorithm can be analyzed using this theorem, and we expect that its relevance is not restricted to the application we present in the subsequent sections.

Before we analyze the time complexity, though, we first introduce the concept of an \textit{implementing subspace}. This subspace depends on the particular input $x \in \{0,1\}^n$, and has the properties that it is often much smaller than the ambient Hilbert space $\H$, and that throughout the execution of the span program algorithm the state vector remains in this subspace, and hence all operations in the span program algorithm need only be defined in this subspace to ensure successful computation of the span program.

\begin{definition}[Implementing subspace]
	\label{def:implementing_subspace}
	Let $\lambda \in [0,1)$ and let $P = (\H,\V,A,\ket{\tau})$ be a span program that positively $\lambda$-approximates a Boolean function $f : X \subseteq \{0,1\}^n \to \{0,1\}$. Let $x \in X$ and let $\H_x$ be a subspace of $\H$ such that:
	\begin{enumerate}
		\setlength\itemsep{-.4em}
		\item $\Pi_{\Ker(A)} \H_x \subseteq \H_x$.
		\item $\Pi_{\H(x)} \H_x \subseteq \H_x$.
		\item $\ket{0} \in \H_x$, where $\ket{0}$ is the all-zeros computational basis state.
		\item $\ket{w_0} \in \H_x$, where $\ket{w_0} = A^+\ket{\tau}$ is the minimal witness for $P$.
	\end{enumerate}
	Then we refer to $\H_x$ as the \emph{implementing subspace of $P$ for $x$}.
\end{definition}

For any $x \in X$, a valid implementing subspace $\H_x$ of $P$ for $x$ is $\H$ itself. For example, we can always implement $2\ket{0}\bra{0}-\I_{\H}$ in complexity $\O(\log\dim\H)$, by simply checking that every qubit is in the state $\ket{0}$. However, for algorithms with large space complexity, such as the element distinctness algorithm \cite{Amb07}, this is very costly, especially if we have to do it many times. In some cases, as in our main theorem in \cref{sec:main_thm}, we can show that the span program has an implementing subspace in which implementing $2\ket{0}\bra{0}-\I$ is easy, thus circumventing an undesired $\log\dim\H$ overhead in the time complexity of the span program algorithm.

The notion of an implementing subspace is not exclusive to span program algorithms. Indeed, any algorithm with high space complexity would run into the same problem if it contains a reflection around any state (i.e., a one-dimensional subspace), even a computational basis state. This includes most quantum walk based algorithms. However, even algorithms with low space complexity could benefit from this technique.

Now, we can state the main result of this section.

\begin{theorem}\label{thm:span-program-time}
	Fix $\lambda \in [0,1/2)$. Suppose $P = (\H,\V,A,\ket{\tau})$ is a span program that positively $\lambda$-approximates a function $f : X \subseteq \{0,1\}^n \to \{0,1\}$. For all $x \in X$, let $\H_x$ be an implementing subspace for $P$. Suppose that we have access to the following subroutines and their controlled versions:
	\begin{enumerate}
		\setlength\itemsep{-.4em}
		\item A subroutine $\Cal R_{\Ker(A)}$ that acts on $\H_x$ as $2\Pi_{\Ker(A)} - \I$.
		\item A subroutine $\Cal C_{\ket{w_0}}$ that leaves $\H_x$ invariant and maps $\ket{0}$ to $\ket{w_0} / \norm{\ket{w_0}}$.
		\item A subroutine $\Cal R_{\H(x)}$ that acts on $\H_x$ as $2\Pi_{\H(x)} - \I$.
		\item A subroutine $\Cal R_{\ket{0}}$ that acts on $\H_x$ as $2\ket{0}\bra{0} - \I$.
	\end{enumerate}
	Then we can implement the span program algorithm for $P$ using a number of calls to the previous subroutines that satisfies
	\[\O\left(\frac{\sqrt{W_+(P)\widetilde{W}_-(P)}}{(1-2\lambda)^{3/2}}\log\frac{1}{1-2\lambda}\right).\]
	Moreover, the number of extra gates and auxiliary qubits used is $\O(\mathrm{polylog}(C(P),1/(1-2\lambda)))$.\footnote{By $f(x,y) = \O(\mathrm{polylog}(x,y))$, we mean that there exist constants $C_1,C_2 > 0$ such that $f(x,y) = \O(\log^{C_1}(x)\log^{C_2}(y))$, in the limit where $x,y \to \infty$.} Finally, it suffices to merely use upper bounds on $W_+(P)$, $W_-(P)$ and $\lambda$, if one substitutes these upper bounds in the relevant complexities.
\end{theorem}

The purpose of \cref{thm:span-program-time} is to enumerate the fundamental instance-dependent operations that have to be given by the user to compile a particular span program algorithm. In other words, if one wants to compile a time-efficient algorithm from a span program, it suffices to give time-efficient implementations of the four subroutines listed in \cref{thm:span-program-time}. Moreover, observe that there is no reference to the normalization of the span program $P$ in the above theorem statement. The proof of this theorem will take care of all the normalization in a very general way.

The remainder of this section is dedicated to proving \cref{thm:span-program-time}. The proof is divided into four lemmas which we prove first, followed by the proof of \cref{thm:span-program-time}. The first lemma deals with the preparation of states of a certain kind. As a corollary we obtain circuits to construct two states necessary for the unitary $U(P^\beta, x)$.

\begin{restatable}{lemma}{lemComb}
	\label{lem:state-comb}
	Let $\alpha_0, \alpha_1, \alpha_2 \in \mathbb{C}$ be such that $|\alpha_0|^2 + |\alpha_1|^2 + |\alpha_2|^2 = 1$, and let $\ket{w_0}$ be the minimal witness for the span program $P = (\H,\V,A,\ket{\tau})$. For all $x \in X$, let $\H_x$ be an implementing subspace. We define
	\[\ket{\eta} = \alpha_0\ket{\hat{0}} + \alpha_1\ket{\hat{1}} + \alpha_2\frac{\ket{w_0}}{\norm{\ket{w_0}}}.\]
	Assume that we have access to controlled versions of the following subroutines:
	\begin{enumerate}
		\setlength\itemsep{-.4em}
		\item A subroutine $\Cal C_{\ket{w_0}}$ that leaves $\H_x$ invariant and maps $\ket{0}$ to $\ket{w_0}/\norm{\ket{w_0}}$.
		\item A subroutine $\Cal R_{\ket{0}}$ that acts on $\H_x$ as $2\ket{0}\bra{0} - \I$.
	\end{enumerate}
	Let $\H_x^{\beta} = \H_x \oplus \Span\{\ket{\hat{0}},\ket{\hat{1}}\}$. Then we can implement a circuit $\Cal C_{\ket{\eta}}$ that leaves $\H_x^{\beta}$ invariant and maps $\ket{0}$ to $\ket{\eta}$, with one call to $\Cal C_{\ket{w_0}}$, two calls to $\Cal R_{\ket{0}}$, and $\O(1)$ extra gates and auxiliary qubits.
\end{restatable}

\begin{proof}
	Recall that we can encode $\ket{\hat 0}$ and $\ket{\hat 1}$ as $\ket{\hat{0}} = \ket{0,0}$, $\ket{\hat{1}} = \ket{0,1}$, and identify every $\ket{h} \in \H$ with $\ket{h} \otimes \ket{2}$. Our mapping $\Cal C_{\ket{\eta}}$ is supposed to map $\ket{0} \in \H \subseteq \H^{\beta}$ to $\ket{\eta} \in \H^{\beta}$, so it is supposed to implement $\ket{0,2} \mapsto \ket{\eta}$.

	First of all, we check if the first register is in state $\ket{0}$ by preparing an auxiliary qubit in the state $\ket{+}$, and then controlled on this auxiliary qubit calling the routine $\Cal R_{\ket{0}}$. If the first register was in the state $\ket{0}$, then we remain in $\ket{+}$, and if not we get a $\ket{-}$ in this qubit. Using a single Hadamard gate, we can now store in the auxiliary qubit whether the first register is in the $\ket{0}$-state.

	Next, controlled on the first register being in the $\ket{0}$-state, we apply the mapping $\ket{2} \mapsto \alpha_0\ket{0} + \alpha_1\ket{1} + \alpha_2\ket{2}$ to the second register. This can be implemented using $\O(1)$ gates, namely by implementing two controlled rotations, one in the plane $\Span\{\ket{1},\ket{2}\}$ and one in the plane $\Span\{\ket{0},\ket{1}\}$.

	Now, we uncompute the first part of our computation, i.e., we uncompute the auxiliary qubit that stored whether the first register was in state $\ket{0}$. This again takes one controlled call to $\Cal R_{\ket{0}}$ and $\O(1)$ extra gates. Observe that the total mapping has now only modified the second register when the first register was in the state $\ket{0}$. But as $\ket{0} \otimes \C^{\{0,1,2\}} = \Span\{\ket{0,2},\ket{\hat{0}},\ket{\hat{1}}\} \subseteq \H_x^{\beta}$, the mapping that we have implemented up to now leaves $\H_x^{\beta}$ invariant.

	Finally, controlled on the second register being in the state $\ket{2}$, we call the circuit $\Cal C_{\ket{w_0}}$. Checking whether the second register is in state $\ket{2}$ can be done in $\O(1)$ gates, and this takes one controlled call to $\Cal C_{\ket{w_0}}$. Moreover, as $\Cal C_{\ket{w_0}}$ leaves $\H_x$ invariant, we find that $\Cal C_{\ket{w_0}} \otimes \ket{2}\bra{2}$ also leaves $\H_x \otimes \ket{2} \subseteq \H_x^{\beta}$ invariant. This completes the proof.
\end{proof}

\begin{restatable}{corollary}{lemwb}
	\label{cor:TB-beta}
	Let $P$ be a span program, $\beta > 0$, and for all $x \in X$, let $\H_x$ be an implementing subspace. Suppose that we have access to the following controlled subroutines:
	\begin{enumerate}
		\setlength\itemsep{-.4em}
		\item A subroutine $\Cal C_{\ket{w_0}}$ that leaves $\H_x$ invariant and maps $\ket{0}$ to $\ket{w_0} / \norm{\ket{w_0}}$.
		\item A subroutine $\Cal R_{\ket{0}}$ that acts on $\H_x$ as $2\ket{0}\bra{0} - \I$.
	\end{enumerate}
	Let $\H_x^{\beta} = \H_x \oplus \Span\{\ket{\hat{0}},\ket{\hat{1}}\}$. Then, we can implement the mappings $\Cal C_{\ket{w_0^{\beta}}}$ and $\Cal C_{\ket{w_0} - \beta\ket{\hat{0}}}$ that leave $\H_x^{\beta}$ invariant while mapping the state $\ket{0}$ to $\ket{w_0^{\beta}}$ and $(\ket{w_0} - \beta\ket{\hat{0}}) / \norm{\ket{w_0} - \beta\ket{\hat{0}}}$, respectively, with one call to $\Cal C_{\ket{w_0}}$, two calls to $\Cal R_{\ket{0}}$ and $\O(1)$ extra gates and auxiliary qubits.
\end{restatable}

\begin{proof}
	From the form of $\ket{w_0^\beta}$, as presented in \cref{eq:wb}, we observe that the construction of $\Cal C_{\ket{w_0^{\beta}}}$ follows from \cref{lem:state-comb} by taking
	\[\alpha_0 = \frac{N}{\beta^2 + N}, \qquad \alpha_1 = \frac{\beta}{\sqrt{\beta^2 + N}}, \qquad \text{and} \qquad \alpha_2 = \frac{\beta\sqrt{N}}{\beta^2 + N}.\]
	Similarly, the construction of $\Cal C_{\ket{w_0} - \beta\ket{\hat{0}}}$ follows from \cref{lem:state-comb} by taking
	\[\alpha_0 = \frac{-\beta}{\sqrt{\beta^2 + N}}, \qquad \alpha_1 = 0, \qquad \text{and} \qquad \alpha_2 = \frac{\sqrt{N}}{\sqrt{\beta^2 + N}},\]
	completing the proof.
\end{proof}

The following lemma serves to construct the reflection around $\Ker(A^\beta)$ using the ability to reflect around $\Ker(A)$, $\ket{0}$ and generate $\ket{w_0}$.

\begin{lemma}\label{lem:TA-beta}
	Let $P = (\H,\V,A,\ket{\tau})$ be a span program, $\beta > 0$, and for all $x \in X$, let $\H_x$ be an implementing subspace. Suppose that we have access to the following subroutines and their controlled versions:
	\begin{enumerate}
		\setlength\itemsep{-.4em}
		\item A subroutine $\Cal R_{\Ker(A)}$ that acts on $\H_x$ as $2\Pi_{\Ker(A)} - \I$.
		\item A subroutine $\Cal C_{\ket{w_0}}$ that leaves $\H_x$ invariant and implements the mapping $\ket{0} \mapsto \ket{w_0} / \norm{\ket{w_0}}$.
		\item A subroutine $\Cal R_{\ket{0}}$ that acts on $\H_x$ as $2\ket{0}\bra{0} - \I$.
	\end{enumerate}
	Let $\H_x^{\beta} = \H_x \oplus \Span\{\ket{\hat{0}},\ket{\hat{1}}\}$. Then we can implement the circuit $\Cal R_{\Ker(A^{\beta})}$ that acts on $\H_x^{\beta}$ as $2\Pi_{\Ker(A^{\beta})} - \I$, using $\O(1)$ controlled calls to the subroutines, extra gates and auxiliary qubits.
\end{lemma}

\begin{proof}
	First, recall that $A\ket{w_0} = \ket{\tau}$, so from the definition of $A^{\beta}$ in \cref{eq:VAt} we find that $A^\beta \ket{w_0} = \beta \ket{\tau}$ and hence $\ket{w_0} - \beta\ket{\hat{0}} \in \ker A^\beta$.
	Since $\ket{w_0} - \beta\ket{\hat{0}}$ is orthogonal to $\Ker(A)$,
	\[\Ker(A^{\beta}) = \Ker(A) \oplus \Span\{\ket{w_0} - \beta\ket{\hat{0}}\}.\]
	Thus, we can implement the reflection through $\Ker(A^{\beta})$ up to a global phase as a product of the reflection through $\Ker(A)$ on the one hand, and $\Span\{\ket{w_0} - \beta\ket{\hat{0}}\}$ on the other, i.e.,
	\[2\Pi_{\Ker(A^{\beta})} - I = -\left(2\Pi_{\Ker(A)} - I\right)\left(2\Pi_{\Span\{\ket{w_0} - \beta\ket{\hat{0}}\}} - I\right).\]
	Thus, implementing the reflection through $\Ker(A^{\beta})$ comes down to implementing the reflection through $\Ker(A)$ and through $\ket{w_0} - \beta\ket{\hat{0}}$.

	Recall that we identify $\H$ with $\H \otimes \ket{2}$, $\ket{\hat{0}}$ with $\ket{0,0}$, and $\ket{\hat{1}}$ with $\ket{0,1}$. Thus, in order to implement the reflection around $\Ker(A)$ on $\H_x^{\beta}$, we apply $\Cal R_{\Ker(A)}$ on the first register, controlled on the second register being in the state $\ket{2}$, and we add a minus if the second register is not in the state $\ket{2}$. I.e., we apply the operation $\Cal R_{\Ker(A)} \otimes \ket{2}\bra{2} - I_{\H} \otimes \left(I - \ket{2}\bra{2}\right)$. As $\Cal R_{\Ker(A)}$ leaves $\H_x$ invariant, we easily check that this operation leaves $\H_x^{\beta}$ invariant. Moreover, we can recognize whether the second register is in state $\ket{2}$ using $\O(1)$ gates, so implementing this operation takes only $\O(1)$ gates and one call to $\Cal R_{\Ker(A)}$.

	Moreover, recall from \cref{cor:TB-beta} that we can implement the mapping $\Cal C = \Cal C_{\ket{w_0} - \beta\ket{\hat{0}}}$ with $\O(1)$ calls to the subroutines $\Cal C_{\ket{w_0}}$ and $\Cal R_{\ket{0}}$, extra gates, and auxiliary qubits. Moreover, observe that $\Cal R_{\ket{0}} \otimes \ket{2}\bra{2} - \I_{\H} \otimes \left(\I - \ket{2}\bra{2}\right)$ implements $2\ket{0}\bra{0} - \I$ on $\H_x^{\beta}$. As
	\[2\Pi_{\ket{w_0} - \beta\ket{\hat{0}}} - \I = C\left(2\ket{0}\bra{0} - \I\right)C^{\dagger},\]
	we can reflect through the state $\ket{w_0} - \beta\ket{\hat{0}}$ with $\O(1)$ calls to the subroutines, extra gates and auxiliary qubits.

	Thus, implementing the operations $2\Pi_{\ket{w_0} - \beta\ket{\hat{0}}} - \I$ and $2\Pi_{\Ker(A)} - \I$ consecutively allows for implementing the reflection around $\Ker(A^{\beta})$. As both individual reflections leave $\H_x^{\beta}$ invariant, so does their product, and the total number of calls to the subroutines, extra gates and auxiliary qubits are all $\O(1)$. Note that for the controlled implementation of the reflection through $\Ker(A^{\beta})$, we need to add an extra $Z$-gate to the control qubit to account for the global phase we neglected here. This completes the proof.
\end{proof}

Now that we know how to implement the reflection around $\Ker(A^{\beta})$, we proceed with analyzing the cost of reflecting around $\H^{\beta}(x)$. this is the objective of the following lemma.

\begin{lemma}\label{lem:TC-beta}
	Let $P$ be a span program, $\beta > 0$, and for all $x \in X$, let $\H_x$ be an implementing subspace. Suppose that we have controlled access to a subroutine $\Cal R_{\H(x)}$ that on $\H_x$ acts as $2\Pi_{\H(x)} - \I$. Then we can implement a circuit $\Cal R_{\H^{\beta}(x)}$ that on $\H_x^{\beta} = \H_x \oplus \Span\{\ket{\hat{0}}, \ket{\hat{1}}\}$ acts as $2\Pi_{\H^{\beta}(x)} - \I$, with one controlled call to $\Cal R_{\H(x)}$ and $\O(1)$ extra qubits and gates.
\end{lemma}

\begin{proof}
	From \cref{thm:scaling} we find that $\H^{\beta}(x) = \H(x) \oplus \Span\{\ket{\hat{1}}\}$. Since $\ket{\hat{1}}$ is orthogonal to $\H(x)$, the reflection through $\H^{\beta}(x)$ up to a global phase is merely the product of the reflections through $\H(x)$ and $\Span\{\ket{\hat{1}}\}$. Furthermore, the controlled implementation of $\Cal R_{\H^{\beta}(x)}$ has to have another $Z$-gate on the control qubit to account for the global phase that we neglect here.

	Recall that we identify $\H$ with $\H \otimes \ket{2}$, $\ket{\hat{0}}$ with $\ket{0,0}$ and $\ket{\hat{1}}$ with $\ket{0,1}$. Thus, we can implement the reflection through $\Span\{\ket{\hat{1}}\}$ in time $\O(1)$, because we can simply implement the operation $I_{\H} \otimes (2\ket{1}\bra{1} - \I)$ in $\O(1)$ gates.

	Similarly, we can apply the reflection through $\H(x)$ on $\H_x^{\beta}$ with one call to $\Cal R_{\H(x)}$, by implementing the operation $\Cal R_{\H(x)} \otimes \ket{2}\bra{2} - \I_{\H} \otimes \left(I - \ket{2}\bra{2}\right)$. This can be done with $\O(1)$ extra gates and auxiliary qubits, and one controlled call to $\Cal R_{\H(x)}$, completing the proof.
\end{proof}

Now we are ready to give the proof of the main theorem of this section.

\begin{proof}[Proof of \cref{thm:span-program-time}]
	As $\Pi_{\Ker(A)}$ and $\Pi_{\H_x}$ commute on $\H_x$, we have that $\Cal R_{\Ker(A)}$ leaves $\H_x$ invariant, as
	\[\left(2\Pi_{\Ker(A)} - \I\right)\Pi_{\H_x} = \Pi_{\H_x}\left(2\Pi_{\Ker(A)} - \I\right),\]
	and hence the image of $2\Pi_{\Ker(A)}-\I$ on $\H_x$ is also contained in $\H_x$. The same holds for $\Cal R_{\H(x)}$ and $\Cal R_{\ket{0}}$.

	By \cref{cor:TB-beta}, we can generate the initial state $\ket{w_0^{\beta}}$ with $\O(1)$ extra gates, auxiliary qubits and calls to $\Cal C_{\ket{w_0}}$ and $\Cal R_{\ket{0}}$, so it remains to implement the phase estimation part of the algorithm compiled from the span program. The span program unitary $U$, as defined in \cref{eq:U}, can be rewritten as
	\begin{equation}
		U = -\left(2\Pi_{\H^{\beta}(x)} - \I\right)\left(2\Pi_{\Ker(A^{\beta})} - \I\right)\left(2\ket{w_0^{\beta}}\bra{w_0^{\beta}} - \I\right).
		\label{eq:U2}
	\end{equation}
	Observe that all factors above leave $\H_x^{\beta}$ invariant, and hence $U$ as a whole leaves $\H_x^{\beta}$ invariant. As the initial state in the phase estimation algorithm is also an element of $\H_x^{\beta}$, the eigenanalysis of $U$ is unaltered.

	By \cref{lem:TC-beta,lem:TA-beta}, we can implement the first two factors in the above expression, respectively, with $\O(1)$ calls to the subroutines, extra gates and auxiliary qubits. Moreover, recall that
	\[2\ket{w_0^{\beta}}\bra{w_0^{\beta}} - \I = \Cal C_{\ket{w_0^{\beta}}}\left(2\ket{0}\bra{0} - \I\right)\Cal C_{\ket{w_0^{\beta}}}^{\dagger},\]
	and hence by virtue of \cref{cor:TB-beta} we can also implement the last term with $\O(1)$ calls to the subroutines, extra gates and auxiliary qubits. Thus, we conclude that we can implement $U$ with essentially the same cost, and remark that we can thus also implement a controlled-$U$ operation, where we have to add another $Z$-gate to the control qubit to account for the global phase in \cref{eq:U2}.

	Finally, recall that the total number of calls to controlled-$U$, and hence to the subroutines, in the algorithm compiled from the span program satisfies
	\[\O\left(\frac{\sqrt{W_+(P)\widetilde{W}_-(P)}}{(1-2\lambda)^{3/2}}\log\frac{1}{1-2\lambda}\right).\]
	Moreover, as the algorithm compiled from the span program implements phase estimation up to precision $\Theta$ with error probability at most $\varepsilon$, and amplitude estimation up to precision $\Theta'$, the number of extra gates and auxiliary qubits introduced by these algorithms satisfy
	\[\O\left(\mathrm{polylog}\left(\frac{1}{\Theta'}, \frac{1}{\Theta}\log\frac{1}{\varepsilon}\right)\right) = \O\left(\mathrm{polylog}\left(\frac{\sqrt{W_+(P)\widetilde{W}_-(P)}}{(1-2\lambda)^{3/2}}\right)\right) = \O\left(\mathrm{polylog}\left(C(P), \frac{1}{1-2\lambda}\right)\right).\]
	Finally, if we only know upper bounds to $W_+(P)$, $W_-(P)$ and $\lambda$, we are merely running the phase estimation and amplitude estimation routines with a better accuracy than strictly necessary. This completes the proof.
\end{proof}

\section{From algorithms to span programs}
\label{sec:main_thm}

Let $\A$ be a clean quantum algorithm that evaluates a function $f : X \subseteq \{0,1\}^n \to \{0,1\}$ with error probability $0 \leq \eps < 1/2$, as in \cref{def:QA}. Based on this algorithm, one can construct a span program that approximates the same function and whose complexity is equal to the query complexity of $\A$, up to a multiplicative constant. This construction was first introduced by Reichardt~\cite{Rei09} in the case where the algorithm has one-sided error, and extended to the case of bounded (two-sided) error in~\cite{Jef20}.

Our contribution is to extend this construction so that it not only preserves the query complexity of $\A$ but also the time complexity. Starting with a quantum algorithm $\A$ whose query complexity is $S$ and time complexity is $T$, we construct a corresponding span program $P_\A$ that accounts for individual gates of $\A$. If the span program is compiled back to a quantum algorithm, the resulting algorithm still solves the same problem, its query complexity is $\widetilde{\O}(S)$ and its time complexity remains $\widetilde{\O}(T)$. This requires modifications to the span program construction, but more importantly, an additional highly non-trivial analysis of the time complexity of the span program implementation.

\subsection{The span program of an algorithm}
\label{sec:span_program}

Recall from \cref{sec:model} that we can assume without loss of generality that there are no two consecutive queries in the algorithm $\A$, and that the first and last unitaries are not queries. We label the time steps where the algorithm queries the inputs by
\begin{equation}
	\S = \{q_1, \dots, q_S\} \subseteq [T],
	\label{eq:S}
\end{equation}
where $T$ is the total time complexity and $S$ denotes the total number of queries. For convenience, we also define $q_0 = 0$, $q_{S+1} = T+1$. We denote the $\ell$-th block of contiguous non-query time steps by $\B_{\ell} \subseteq [T]$, with $\ell \in [S+1]$. See \cref{fig:timesteps} for an overview of this notation.

\begin{figure}[ht!]
\centering
\begin{tikzpicture}[scale = .7]
	\footnotesize
	\foreach \x in {1,...,13} {
		\draw ({\x-.5},-.5) -- ({\x+.5},-.5) -- ({\x+.5},.5) -- ({\x-.5},.5) -- cycle;
	}
	\node[anchor=east] at (-1,1) {Time step};
	\node[anchor=east] at (-1,0) {Type};
	\node[anchor=east] at (-1,-1) {Label};
	\foreach \x in {0,...,7} {
		\node at (\x,1) {$\x$};
	}
	\node at (9,1) {$\cdots$};
	\node at (9,-1) {$\cdots$};
	\foreach \x in {1,2} {
		\node at ({13-\x},1) {$T\!\!-\!\x$};
	}
	\node at (14,1) {$T\!\!+\!1$};
	\node at (13,1) {$T$};
	\foreach \x in {3,7,11} {
		\draw[pattern=north west lines] ({\x-.5},-.5) -- ({\x+.5},-.5) -- ({\x+.5},.5) -- ({\x-.5},.5) -- cycle;
	}
	\foreach \x in {0,14} {
		\draw[dashed] ({\x-.5},-.5) -- ({\x+.5},-.5) -- ({\x+.5},.5) -- ({\x-.5},.5) -- cycle;
	}
	\node at (0,-1) {$q_0$};
	\node at (3,-1) {$q_1$};
	\node at (7,-1) {$q_2$};
	\node at (11,-1) {$q_S$};
	\node at (14,-1) {$q_{S\!+\!1}$};
	\draw[decoration={brace,mirror,raise=3pt},decorate] (.5,-.5) -- node[below=6pt] {$\B_1$} (2.5,-.5);
	\draw[decoration={brace,mirror,raise=3pt},decorate] (3.5,-.5) -- node[below=6pt] {$\B_2$} (6.5,-.5);
	\draw[decoration={brace,mirror,raise=3pt},decorate] (11.5,-.5) -- node[below=6pt] {$\B_{S\!+\!1}$} (13.5,-.5);
\end{tikzpicture}
\caption{Synopsis of the notation. The cells denote time steps of the algorithm $\A$ where time progresses to the right. They are indexed from $1$ to $T$. The hatched cells denote time steps in which a query to the input $x$ is performed. In all other time steps $t$ a unitary $U_t$ independent of $x$ is applied.}
\label{fig:timesteps}
\end{figure}
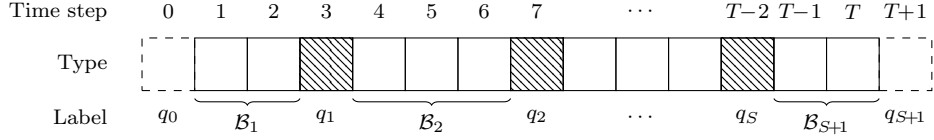

Recall that $\mathcal{W}$ is a finite set that labels the basis of the workspace of $\A$.
We define the following spaces:
\begin{align}
	\forall i \in [n], b \in \{0,1\}, \qquad \H_{i,b} &= \Span\{\ket{t,b,i,j} : t+1 \in \mathcal{S}, j \in \mathcal{W}\}, \nonumber\\
	\Htrue &= \Span\{\ket{t,0,i,j} : t+1 \in [T+1] \setminus \mathcal{S}, i \in [n], j \in \mathcal{W}\}, \label{eq:alg-span}\\
	\Hfalse &= \{0\}. \nonumber
\end{align}
As usual, the spaces $\H(x)$ and $\H$ are defined from these as:
\begin{equation}
	\forall x \in \{0,1\}^n, \qquad
	\H(x) = \biggl(\bigoplus_{i=1}^n \H_{i,x_i}\biggr) \oplus \Htrue \qquad \text{and} \qquad
	\H = \biggl(\bigoplus_{\substack{i \in [n]\\b \in \{0,1\}}} \H_{i,b}\biggr) \oplus \Htrue \oplus \Hfalse.
  \label{eq:H}
\end{equation}
For better intuition, we provide a graphical depiction of $\H$, $\Htrue$, $\H(x)$ and $\H_{i,b}$ in \cref{fig:Hilbert_spaces}.

\begin{figure}[ht!]\centering
\begin{tikzpicture}[thick, > = latex,
	 qbox/.style = {pattern = north east lines},
	 xbox/.style = {fill = gray},
	 circ/.style = {circle, draw, fill = white},
	 arrw/.style = {rounded corners, ->, > = to}
	]
	\newcommand{\gatebox}[1]{
	  \draw (#1) rectangle +(1,1);
	}
	\newcommand{\querybox}[2]{
	  \draw[qbox] (#1) rectangle +(1,1);
		\path (#1)+(0.5,0.5) node[fill = white] {$#2$};
	}
	\foreach \x in {0,2,3,5,7,8} {
	  \gatebox{\x,0}
		\draw (\x+.5,-.4) -- (\x+.5,-.6) node[below] {$\x$};
	}
	\foreach \x in {1,4,6} {
	  \querybox{\x,0}{0}
		\querybox{\x,1}{1}
		\draw (\x+.5,-.4) -- (\x+.5,-.6) node[below] {$\x$};
	}
	\foreach \x/\i in {2/1,5/2,7/3} {
	  \node at (\x+.5,-1.5) {$q_\i$};
	}
	\node at (8+.5,-1.5) {$T$};
	\draw[->] (-.75,-.5) -- (9.5,-.5) node[right] {$t$};
	\draw[->] (-.5,-.75) -- (-.5,2.5) node[above] {$b$};
	\draw (-.4,0.5) -- (-.6,0.5) node[left] {$0$};
	\draw (-.4,1.5) -- (-.6,1.5) node[left] {$1$};
  \def\r{0.2}
	\def\B{2.9}
	\node at (4.5,\B+0.5) {$\H$};
	\path[draw, radius = \r]
	  (0.0,\B-\r) arc [start angle = 180, end angle =  90] --
		(4.5-\r,\B) arc [start angle = -90, end angle =   0]
		            arc [start angle = 180, end angle = 270] --
    (9.0-\r,\B) arc [start angle =  90, end angle =   0];
	\node at (2.5,0.5) {$\Htrue$};
	\begin{scope}[xshift = 12.7cm, yshift = -1.3cm]
		\def\w{0.3}
		\def\h{0.2}
		\def\L{1.8}
		\querybox{-2.2,2*\L}{1}
		\querybox{-2.2,1*\L}{0}
		\gatebox {-2.2,0*\L}
		\fill[xbox] (0.05,0.1+2*\L+1*\h) rectangle +(3*\w,2*\h);
		\fill[xbox] (0.05,0.1+1*\L+3*\h) rectangle +(3*\w,1*\h);
		\fill[xbox] (0.05,0.1+1*\L+0*\h) rectangle +(3*\w,1*\h);
		\fill[xbox] (0.05,0.1+0*\L+0*\h) rectangle +(3*\w,4*\h);
		\foreach \k/\s in {1,2,3} {
			\begin{scope}[yshift = (\k-1)*\L cm]
				\foreach \i in {1,...,4} {
					\draw (0.05,-0.1+\h*\i) rectangle +(3*\w,\h);
					\draw (-0.1, 0.0+\h*\i) -- +(-0.1,0) node[left] {\tiny$\i$};
				}
				\foreach \j in {1,...,3} {
					\draw (-0.25+\w*\j,0.1) rectangle +(\w,4*\h);
					\draw (-0.10+\w*\j,-0.05) -- +(0,-0.1) node[below] {\tiny$\j$};
				}
				\draw[->] (-.30,-.10) -- (1.25,-.1) node[right] {$j$};
				\draw[->] (-.15,-.25) -- (-.15,1.2) node[left ] {$i$};
				\node at (-0.8,0.5) {$=$};
			\end{scope}
		}
		\node[inner sep = 2pt] (Hx) at (2.3,2.3) {$\H(x)$};
		\draw[arrw] (Hx) -- ++(0,1.8) -- ++(-1.2,0);
		\draw[arrw] (Hx) -- +(-1.2, 0.3);
		\draw[arrw] (Hx) -- +(-1.2,-0.3);
		\draw[arrw] (Hx) -- ++(0,-1.8) -- ++(-1.2,0);
	\end{scope}
\end{tikzpicture}
\caption{Graphical depiction of the relevant spaces when $T = 8$, $\mathcal{S} = \{2,5,7\}$, $n = 4$, $|\W| = 3$ and $x = 0110$. The total space $\H$ is a direct sum of all blocks on the left, where the block at position $(t,b) \in [T]_0 \times \{0,1\}$ denotes the subspace spanned by all computational basis states of the form $\ket{t,b,\cdot,\cdot}$. Every block is of one of three types, white, 0 or 1, shown on the right. The subspace $\Htrue$ is the direct sum of all white blocks. Each block further decomposes as a direct sum over computational basis states $\ket{i,j}$, $i \in [n]$, $j \in \W$. The gray cells of all blocks together span the space $\H(x)$.
Finally, for a given $i \in [n]$, the subspaces $\H_{i,0}$ and $\H_{i,1}$ consist of the $i$-th row within all 0 and 1 blocks, respectively.}
\label{fig:Hilbert_spaces}
\end{figure}
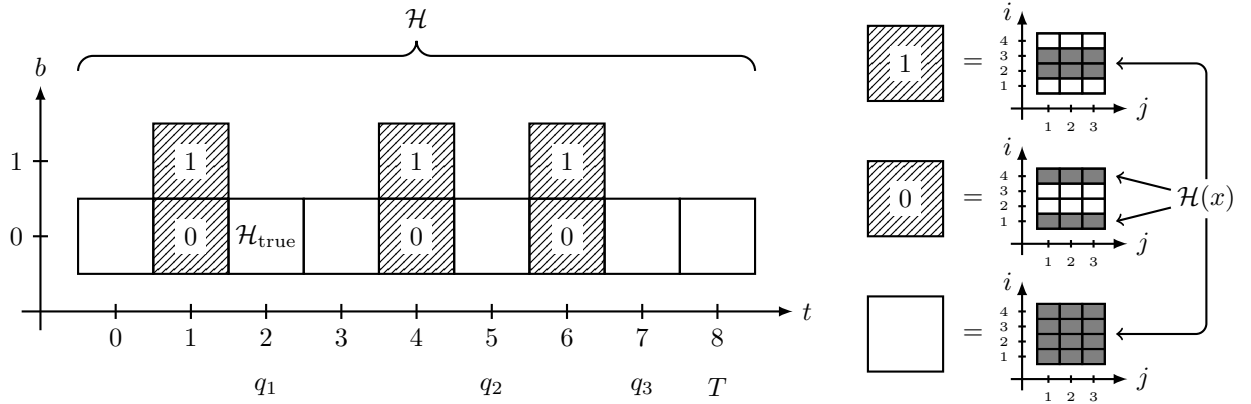

Let $[T]_0 := \{0,\dotsc,T\}$. We define the target space $\V$ and the target vector $\ket{\tau} \in \V$ as follows:
\begin{align}
	\V &= \Span\{\ket{t,i,j} : t \in [T]_0, i \in [n], j \in \mathcal{W}\}, &
	\ket{\tau} &= \ket{0}\ket{\Psi_0} - \ket{T}\ket{\Psi_T},
	\label{eq:tau}
\end{align}
where $\ket{\Psi_0}$ is the initial state of $\A$ (see \cref{eq:Psi}) and $\ket{\Psi_T}$ is the final accepting state (see \cref{def:QA}).

Recall that $S$ denotes the total number of queries and $\eps$ is the error probability of $\A$. Let
\begin{equation}
	a = \sqrt{\frac{\eps}{2S+1}} \qquad \text{and} \qquad M = \max_{\ell \in [S+1]} \sqrt{|\B_\ell|},
	\label{eq:aM}
\end{equation}
where $\B_\ell \subseteq [T]$ is the $\ell$-th contiguous block of non-query gates (see \cref{fig:timesteps}). By \cref{def:QA} and \cref{lem:QA}, we can assume that $M\leq \sqrt{3T/S}$. For all computational basis vectors $\ket{t,b,i,j}$ in $\H$, we define the action of the span program operator $A \in \L(\H,\V)$ as follows:
\begin{equation}
	A\ket{t,b,i,j}
=	\begin{cases}
		a\ket{T,i,j} & \text{if } t = T, \\
		M(\ket{t,i,j} - \ket{t+1}U_{t+1}\ket{i,j}) & \text{if } \exists \ell \in [S+1] : t + 1 \in \B_{\ell}, \\
		\ket{t,i,j} - (-1)^b\ket{t+1,i,j} & \text{if } \exists \ell \in [S] : t + 1 = q_{\ell}.
	\end{cases}
	\label{eq:A}
\end{equation}
The weights $a$ and $M$ are the main difference between our construction and that of \cite{Jef20}, and will enable the time-efficient implementation of the span program described in \cref{sec:implementation}.
The unitary $U_{t+1}$ is the $(t+1)$-th unitary of algorithm $\A$ as defined in \cref{sec:quant_algorithm}.

\begin{definition}[Span program of an algorithm]\label{def:PA}
	The \emph{span program of a quantum algorithm} $\A$ is $P_\A = (\H,\V,A,\ket{\tau})$, where $\H$ is defined in \cref{eq:alg-span,eq:H}, $\V$ and $\ket{\tau}$ in \cref{eq:tau}, and $A$ in \cref{eq:A}.
\end{definition}

We spend the remainder of this section proving various properties of span programs of this type. We start by analyzing the positive and negative witness sizes $W_+(P)$ and $W_-(P)$, and the approximation factor $\lambda$.

\begin{theorem}\label{thm:alg-span-complexity}
Let $\A$ be a clean  quantum algorithm for $f$ with error probability $0\leq \eps <1/5$, making $S$ queries, and let $P_{\A}$ the span program for $\A$ from \cref{def:PA}. Then $P_{\A}$ positively $5\eps$-approximates $f$ with complexities $W_+(P_{\A})= \O(S)$ and $\widetilde{W}_-(P_{\A})= \O(S)$.
\end{theorem}

\Cref{thm:alg-span-complexity} follows directly from \cref{lem:w+,lem:w-} below. The proofs are similar to those of \cite{Jef20}, which are themselves similar to \cite{Rei09}, with the difference that the operator $A$ of our span program now has slightly modified weights, see \cref{eq:A}.

\begin{restatable}{lemma}{lemwpos}
	\label{lem:w+}
	Let $\A$ be a clean quantum algorithm with query complexity $S$, time complexity $T$ and error probability $0 \leq \eps < 1/2$. Let $P_{\A}$ be the span program for $\A$ from \cref{def:PA}. Then,
	\[W_+(P_{\A}) \leq 3(2S + 1) = \O(S).\]
\end{restatable}

\begin{proof}
	Let $\mathcal{Z} = [n] \times \W$, so the state space of the algorithm $\A$ is $\C^{\mathcal{Z}}$.
	Recall from \cref{eq:Psi} that $\ket{\Psi_t(x)} \in \C^{\mathcal{Z}}$ denotes the state of $\A$ on input $x$ at time $t$, i.e., immediately after the application of $U_t$.
	We will construct a positive witness for every positive input $x\in f^{-1}(1)$ and upper bound its norm.

	Keeping \cref{eq:A} in mind, for every $t \in [T]_0$ we define
	\[\ket{\widehat{\Psi}_t(x)} = \begin{cases}
	\frac{1}{a}\ket{0}\ket{\Psi_T(x)} & \text{if } t = T, \\
	\frac{1}{M}\ket{0}\ket{\Psi_t(x)} & \text{if } \exists \ell \in [S+1] : t + 1 \in \B_{\ell}, \\
	L_x\ket{\Psi_t(x)} & \text{if } \exists \ell \in [S] : t + 1 = q_{\ell},
	\end{cases}\]
	where $L_x \in \L(\C^{\mathcal{Z}}, \C^{2} \otimes \C^{\mathcal{Z}})$ is defined on the computational basis vectors as follows:
	\[\forall i \in [n], j \in \mathcal{W}, \qquad L_x\ket{i,j} = \ket{x_i,i,j}.\]
	For all $t \in [T]_0$, we easily verify that $\ket{t}\ket{\widehat{\Psi}_t(x)} \in \H(x)$ by referring to \cref{eq:alg-span,eq:H}. Next, we define
	\[\ket{w_x} = \sum_{t=0}^{T-1} \ket{t}\ket{\widehat{\Psi}_t(x)} + \frac{1}{a}\ket{T}\ket{0}\left(\ket{\Psi_T(x)} - \ket{\Psi_T}\right),\]
	where $\ket{\Psi_T}$ is the final accepting state from \cref{def:QA}.
	As $\ket{T,0,z} \in \H(x)$ for all $z \in \mathcal{Z}$, we find by linearity that $\ket{w_x} \in \H(x)$.
  By splitting the time steps into query and non-query steps we find that
	\begin{align*}
		\ket{w_x}
		 &= \sum_{\ell=1}^S \ket{q_{\ell}-1} L_x \ket{\Psi_{q_{\ell}-1}(x)}
			+ \sum_{\ell=1}^{S+1} \sum_{t=q_{\ell-1}}^{q_{\ell}-2} \ket{t} \frac{1}{M} \ket{0} \ket{\Psi_t(x)}
			+ \frac{1}{a}\ket{T}\ket{0}\left(\ket{\Psi_T(x)} - \ket{\Psi_T}\right).
	\end{align*}
  Applying $A$ we get
	\begin{align*}
		A\ket{w_x}
		&= \sum_{\ell=1}^S \bigl[\ket{q_{\ell}-1} \ket{\Psi_{q_{\ell}-1}(x)} - \ket{q_{\ell}} \O_x\ket{\Psi_{q_{\ell}-1}(x)}\bigr] \\
		&\qquad + \sum_{\ell = 1}^{S+1} \sum_{t = q_{\ell-1}}^{q_{\ell}-2} M\left[\ket{t} \frac{1}{M} \ket{\Psi_t(x)} - \ket{t+1} \frac{1}{M} U_{t+1} \ket{\Psi_t(x)}\right] + \ket{T}\ket{\Psi_T(x)} - \ket{T}\ket{\Psi_T} \\
		&= \sum_{\ell=1}^S \left[\ket{q_{\ell}-1}\ket{\Psi_{q_{\ell}-1}(x)} - \ket{q_{\ell}}\ket{\Psi_{q_{\ell}}(x)}\right] + \sum_{\ell=1}^{S+1} \sum_{t=q_{\ell-1}}^{q_{\ell}-2} \left[\ket{t}\ket{\Psi_t(x)} - \ket{t+1}\ket{\Psi_{t+1}(x)}\right] \\
		&\qquad + \ket{T}\ket{\Psi_T(x)} - \ket{T}\ket{\Psi_T} \\
		&= \sum_{t=0}^{T-1} \left[\ket{t}\ket{\Psi_t(x)} - \ket{t+1}\ket{\Psi_{t+1}(x)}\right] + \ket{T}\ket{\Psi_T(x)} - \ket{T}\ket{\Psi_T} \\
		&= \ket{0}\ket{\Psi_0} - \ket{T}\ket{\Psi_T}
		 = \ket{\tau},
	\end{align*}
	where most terms cancel since the final sum is telescopic. In particular, we find that $\ket{w_x}$ is indeed a positive witness for $x$.
	We can use its size to bound the size of the minimum positive witness for $x$:
	\begin{align*}
		w_+(x,P_{\A}) &= \min\{\norm{\ket{w}}^2 : \ket{w} \in \H(x), A\ket{w} = \ket{\tau}\} \leq \norm{\ket{w_x}}^2 = \sum_{t=0}^{T-1} \norm{\ket{\widehat{\Psi}_t(x)}}^2 + \frac{1}{a^2} \norm{\ket{\Psi_T(x)} - \ket{\Psi_T}}^2 \\
		&= \sum_{\ell=1}^S \norm{\ket{\Psi_{q_{\ell}-1}(x)}}^2 + \sum_{\ell=1}^{S+1} \sum_{t=q_{\ell-1}}^{q_{\ell}-2} \frac{1}{M^2} \norm{\ket{\Psi_t(x)}}^2 + \frac{1}{a^2} \norm{\ket{\Psi_T(x)} - \ket{\Psi_T}}^2 \\
		&\leq S + (S + 1) + \frac{1}{a^2} \cdot 2\eps \leq 2S + 1 + \frac{2S+1}{\eps} \cdot 2\eps = 3(2S+1),
	\end{align*}
	where we used $M^2 \geq |\B_\ell| = q_\ell - q_{\ell-1} - 1$ from \cref{eq:aM} to bound the second term.
	To bound the third term, we used $a = \sqrt{\eps/(2S+1)}$ from \cref{eq:aM} and the inequality
	\begin{align*}
		\norm{\ket{\Psi_T(x)}-\ket{\Psi_T}}^2 = 2(1-\mathrm{Re}\braket{\Psi_T(x)}{\Psi_T})=2(1-p_1(x))\leq 2\eps
	\end{align*}
	which holds for any $x \in f^{-1}(1)$ (see \cref{lem:QA,def:QA}). Thus,
	\[W_+(P_{\A}) = \max_{x \in f^{-1}(1)} w_+(x,P_{\A}) \leq 3(2S+1),\]
	which completes the proof.
\end{proof}

\begin{restatable}{lemma}{lemwneg} \label{lem:w-}
	Let $\A$ be a clean quantum algorithm with query complexity $S$, time complexity $T$ and error probability $0 \leq \eps < \frac15$. Let $P_{\A}$ be the span program for $\A$ from \cref{def:PA}. Then, for all $x\in f^{-1}(0)$, there exists an approximate negative witness $\ket{\widetilde\omega_x}$ such that $\norm{\bra{\widetilde{\omega}_x}A\Pi_{\H(x)}}^2\leq 5\eps/(3(2S+1))$ and $\norm{\bra{\widetilde{\omega}_x}A}^2\leq 2(4S+1)$. Thus:
	\begin{enumerate}
		\setlength\itemsep{-.4em}
		\item $P_{\A}$ positively $\lambda$-approximates $f$ for $\lambda = 5\eps$.
		\item The approximate negative witness complexity of $P_{\A}$ is $\widetilde{W}_-(P_{\A}) = \O(S)$.
	\end{enumerate}
\end{restatable}

\begin{proof}
	Given a negative input $x$, we define an approximate negative witness and bound the negative error and minimum approximate negative witness size using this witness. To that end, let $x \in f^{-1}(0)$. Define
	\[\bra{\widetilde\omega_x} = \frac{1}{1 - \braket{\Psi_T(x)}{\Psi_T}}\sum_{t=0}^T \bra{t}\bra{\Psi_t(x)}.\]
	Note that this is well-defined as $x$ is a negative instance, and hence $|\braket{\Psi_T(x)}{\Psi_T}| \leq \eps < 1$ by \cref{def:QA,lem:QA}.
  Recalling from \cref{eq:tau} that $\ket{\tau} = \ket{0}\ket{\Psi_0} - \ket{T}\ket{\Psi_T}$, observe that
	\begin{equation}
		\braket{\widetilde\omega_x}{\tau} =  \frac{\braket{\Psi_0}{\Psi_0} - \braket{\Psi_T(x)}{\Psi_T}}{1 - \braket{\Psi_T(x)}{\Psi_T}} = 1.
		\label{eq:omegatau}
	\end{equation}
	Next, let $\ket{t,b,i,j}$ be a computational basis vector in $\H(x)$ and let $A$ be the operator defined in \cref{eq:A}. If $t + 1 = q_{\ell}$ for some $\ell \in [S]$, then $b = x_i$ and
	\begin{align*}
		\bra{\widetilde\omega_x}A\ket{t,b,i,j} &= \bra{\widetilde\omega_x}\bigl[\ket{t,i,j} - (-1)^{x_i}\ket{t+1,i,j}\bigr] \\
		&= \frac{1}{1 - \braket{\Psi_T(x)}{\Psi_T}} \bigl[\braket{\Psi_t(x)}{i,j} - (-1)^{x_i}\braket{\Psi_{t+1}(x)}{i,j}\bigr] \\
		&= \frac{1}{1 - \braket{\Psi_T(x)}{\Psi_T}} \left[\braket{\Psi_t(x)}{i,j} - \bra{\Psi_t(x)}\O_x^{\dagger}(-1)^{x_i}\ket{i,j}\right] = 0.
	\end{align*}
On the other hand, if $t + 1 \in \B_{\ell}$ for some $\ell \in [S+1]$, then $b = 0$ and
	\begin{align*}
		\bra{\widetilde\omega_x}A\ket{t,b,i,j} &= M\bra{\widetilde\omega_x}\bigl[\ket{t,i,j} - \ket{t+1}U_{t+1}\ket{i,j}\bigr] \\
		&= \frac{M}{1 - \braket{\Psi_T(x)}{\Psi_T}} \bigl[\braket{\Psi_t(x)}{i,j} - \bra{\Psi_{t+1}(x)}U_{t+1}\ket{i,j}\bigr] \\
		&= \frac{M}{1 - \braket{\Psi_T(x)}{\Psi_T}} \left[\bra{\Psi_t(x)} - \bra{\Psi_t(x)}U_{t+1}^{\dagger}U_{t+1}\right]\ket{i,j} = 0.
	\end{align*}
	Finally, if $t = T$ then
	\[\bra{\widetilde\omega_x}A\ket{T,0,i,j} = a\braket{\widetilde\omega_x}{T,i,j} = \frac{a\braket{\Psi_T(x)}{i,j}}{1 - \braket{\Psi_T(x)}{\Psi_T}}\]
where $a$ is defined in \cref{eq:aM}.
This might not evaluate to 0, potentially contributing to the negative witness error of $\bra{\widetilde{\omega}_x}$ for $x$.
Using $\braket{\Psi_T(x)}{\Psi_T}=p_1(x)\leq \eps < 1/5$ (see \cref{def:QA} and \cref{lem:QA}) we find that
	\begin{align}
		\norm{\bra{\widetilde\omega_x}A\Pi_{\H(x)}}^2 &= \sum_{i \in [n], j \in \mathcal{W}} \left|\bra{\widetilde\omega_x}A\ket{T,0,i,j}\right|^2 = \frac{a^2\norm{\ket{\Psi_T(x)}}^2}{\left|1 - \braket{\Psi_T(x)}{\Psi_T}\right|^2} = \frac{a^2}{\left(1 - p_1(x)\right)^2} \leq \frac{a^2}{\left(1 - \eps\right)^2} \nonumber\\
		&\leq \frac{\eps}{2S+1} \cdot \frac{1}{\left(1 - \frac{1}{5}\right)^2} < \frac{\frac{25}{15}\eps}{2S+1} = \frac{5\eps}{3(2S+1)} \leq \frac{5\eps}{W_+(P)},\label{eq:lambda-beta}
	\end{align}
	where in the last inequality we used \cref{lem:w+}. We find that $P$ positively $\lambda$-approximates $f$ with $\lambda = 5\eps$, completing the proof of the first claim.

	To prove the second claim, recall from \cref{eq:omegatau} that $\braket{\widetilde\omega_x}{\tau} = 1$. Hence, for any $x \in f^{-1}(0)$, we obtain using that $\mathcal{Z} = [n] \times \W$:
	\begin{align*}
		\widetilde{w}_-(x,P) &= \min_{\ket{\widetilde\omega} \in \V}\left\{\norm{\bra{\widetilde\omega}A}^2 : \braket{\widetilde\omega}{\tau} = 1, \norm{\bra{\widetilde\omega}A\Pi_{\H(x)}}^2 \leq \frac{\lambda}{W_+(P)}\right\} \leq \norm{\bra{\widetilde\omega_x}A}^2 \\
		&= \sum_{\substack{\ell \in [S]\\ b \in \{0,1\}, z \in \mathcal{Z}}} \left|\bra{\widetilde\omega_x}\bigl[\ket{q_{\ell}-1,z} - (-1)^b\ket{q_{\ell},z}\bigr]\right|^2 \\
		& \qquad + \sum_{\substack{\ell \in [S+1]\\ t \in \{q_{\ell-1}, \dots, q_{\ell}-2\}, z \in \mathcal{Z}}} \left|\bra{\widetilde\omega_x}M\bigl(\ket{t,z} - \ket{t+1}U_{t+1}\ket{z}\bigr)\right|^2 + \sum_{z \in \mathcal{Z}} \left|\bra{\widetilde\omega_x} a \ket{T,z}\right|^2 \\
		&\leq \frac{1}{\left|1 - \braket{\Psi_T(x)}{\Psi_T}\right|^2} \cdot \left[\sum_{\ell \in [S], z \in \mathcal{Z}} 2\left\{\left|\braket{\Psi_{q_{\ell}-1}(x)}{z}\right|^2 + \left|\braket{\Psi_{q_{\ell}}(x)}{z}\right|^2\right\}\right. \\
		&\qquad \left. + \sum_{\substack{\ell \in [S+1]\\ t \in \{q_{\ell-1}, \dots, q_{\ell}-2\}, z \in \mathcal{Z}}} M^2 \left|\left\{\bra{\Psi_t(x)} - \bra{\Psi_{t+1}(x)}U_{t+1}\right\}\ket{z}\right|^2 + \sum_{z \in \mathcal{Z}} a^2\left|\braket{\Psi_T(x)}{z}\right|^2\right] \\
		&= \frac{1}{\left|1 - \braket{\Psi_T(x)}{\Psi_T}\right|^2} \cdot \left[\sum_{\ell \in [S]} 2 \left\{\norm{\bra{\Psi_{q_{\ell}-1}(x)}}^2 + \norm{\bra{\Psi_{q_{\ell}}(x)}}^2\right\} + a^2\norm{\bra{\Psi_T(x)}}^2\right] \\
		&= \frac{4S + a^2}{\left|1 - \braket{\Psi_T(x)}{\Psi_T}\right|^2}  \leq \frac{4S + a^2}{\left(1 - \eps\right)^2} \leq \frac{4S+a^2}{\left(1 - \frac{1}{5}\right)^2} \\
		&\leq \left[4S + \frac{\eps}{2S+1}\right] \cdot \frac{25}{16} \leq 2(4S+1) = \O(S),
	\end{align*}
	which completes the proof.
\end{proof}

Together, \cref{lem:w+,lem:w-} prove \cref{thm:alg-span-complexity}, which in turn implies an upper bound on the query complexity of implementing the span program $P_{\cal A}$.

We conclude this section by characterizing in \cref{lem:kernel} the kernel of the span program operator $A$ and subsequently finding in \cref{lem:w0} the minimal witness size. These will prove relevant in the analysis of the time complexity of the algorithm compiled from $P_{\A}$, which we turn our attention to in \cref{sec:main-proof}.

\begin{restatable}{lemma}{lemkernel}
	\label{lem:kernel}
	Let $\A$ be a clean quantum query algorithm with error probability $0 \leq \varepsilon < 1$. Let $P_{\A} = (\H,\V,A,\ket{\tau})$ be the span program for $\A$. Let $\mathcal{Z} = [n] \times \mathcal{W}$. For $\ell \in \{2,\dots,S\}$, we define the linear map $\Phi_\ell$ from $\C^{\mathcal{Z}}$ to $\H$ as
	\[\Phi_{\ell}\ket{\psi} = \ket{q_{\ell-1}-1}\frac{\ket{-}}{\sqrt{2}}\ket{\psi} + \ket{q_{\ell}-1}\frac{\ket{+}}{\sqrt{2}} U_{q_{\ell}-1} \cdots U_{q_{\ell-1}+1}\ket{\psi} + \frac{1}{M}\sum_{t=q_{\ell-1}}^{q_\ell-2}\ket{t}\ket{0}U_t \cdots U_{q_{\ell-1}+1}\ket{\psi},\]
	where $\ket{\pm} = (\ket{0}\pm\ket{1})/\sqrt{2}$ and $M$ was defined in \cref{eq:aM}.
	We also define the linear map $\Phi_{S+1}$ from $\C^{\mathcal{Z}}$ to $\H$ as
	\[\Phi_{S+1}\ket{\psi} = \ket{q_S-1}\frac{\ket{-}}{\sqrt{2}}\ket{\psi} + \frac1M\sum_{t=q_S}^{T-1}\ket{t}\ket{0}U_t \cdots U_{q_S+1}\ket{\psi} + \frac1a\ket{T}\ket{0} U_T \cdots U_{q_S+1} \ket{\psi}.\]
	Then
	\[\Ker(A) = \bigoplus_{\ell=2}^{S+1} \Phi_{\ell}\left(\C^{\mathcal{Z}}\right).\]
\end{restatable}

The proof can be found in \cref{app:lemwzero}. Also proved in \cref{app:lemwzero} is the following lemma that provides an expression for the minimal positive witness which will be useful in our later analyses.

\begin{restatable}{lemma}{lemwzero}
	\label{lem:w0}
	Let $\A$ be a clean quantum query algorithm with error probability $0 \leq \eps < 1$. Let $P_{\A}$ be the span program for $\A$ from \cref{def:PA}. Then the minimal witness $\ket{w_0}=A^+\ket{\tau}$ is
	\begin{align*}
		\ket{w_0} &= \frac{1}{M}\sum_{t=0}^{q_1-2} \ket{t}\ket{0}U_t \cdots U_1\ket{\Psi_0} + \ket{q_1-1}\left(\frac12\ket{0} + \frac12\ket{1}\right)U_{q_1-1} \cdots U_1\ket{\Psi_0} \\
		&+ \frac{1}{Ca^2+1}\left[\ket{q_S-1}\left(\frac12\ket{0}-\frac12\ket{1}\right)U_{q_S+1}^{\dagger} \cdots U_T^{\dagger}\ket{\Psi_T} + \frac{1}{M} \sum_{t=q_S}^{T-1} \ket{t}\ket{0}U_{t+1}^{\dagger} \cdots U_T^{\dagger}\ket{\Psi_T}\right] \\
		&- \frac{Ca}{Ca^2+1} \ket{T}\ket{0} \ket{\Psi_T}, \qquad \text{where} \qquad C = \frac{T-q_S}{M^2} + \frac12
	\end{align*}
	and $a$ and $M$ are defined in \cref{eq:aM}.
	The squared norm of $\ket{w_0}$ is
	$N = \frac{q_1-1}{M^2} + \frac12 + \frac{C}{Ca^2+1}$.
\end{restatable}

\subsection{Time complexity of algorithms compiled from span programs of algorithms}
\label{sec:main-proof}

Now we analyze the time complexity of implementing $P_{\A}$. The main results are summarized in the following theorem and the remainder of this section is dedicated to proving it.

\begin{theorem}\label{thm:main-time}
	Let $\A$ be a clean quantum query algorithm that acts on $k$ qubits, has query complexity $S$, time complexity $T$, and evaluates a function $f : X \subseteq \{0,1\}^n \to \{0,1\}$ with bounded error. Let $P_{\A}$ be the span program for this algorithm, as in \cref{def:PA}. Then we can implement the algorithm compiled from $P_{\A}$ with:
	\begin{enumerate}
		\setlength\itemsep{-.4em}
		\item $\O(S\log(S))$ calls to $\O_x$.
		\item $\O(T\log(S))$ calls to $\O_\A$ and $\O_\S$, as defined in \cref{sec:model}.
		\item $\O(T\mathrm{polylog}(T))$ additional gates.
		\item $\O(\mathrm{polylog}(T) + k^{o(1)})$ auxiliary qubits.
	\end{enumerate}
	If we additionally require that the error probability of $\A$ is $o(1/S^2)$, then the $\log(S)$ factors and the $k^{o(1)}$ term can be removed. We can also drop the $k^{o(1)}$ term if we assume that $T = k^{1+\Omega(1)}$.
\end{theorem}

The proof leans on the structure of \cref{thm:span-program-time}. We first define a suitable implementing subspace in \cref{sec:consistent-space}, and subsequently, in \cref{sec:implementation}, provide efficient implementations of the four subroutines that are required to use \cref{thm:span-program-time}.

\subsubsection{Implementing subspace}
\label{sec:consistent-space}

Our construction relies on the fact that at intermediate steps the state of an algorithm compiled from $P_{\A}$ is not completely arbitrary but is guaranteed to live within an implementing subspace.

The main reason for this is that we need to reflect around the state $\ket{0}\ket{0}\ket{\Psi_0}$. Doing this in principle is very simple because we can assume that $\ket{\Psi_0}=\ket{0}^{\otimes \log(n)+\log(|{\cal W}|)}$. However, in general, the time complexity of this reflection is $\Theta(\log(n)+\log(|{\cal W}|))$, since it is necessary to check that every qubit is in the $\ket{0}$ state.\footnote{This detail has been neglected in previous work. For example, efficient implementation of such a reflection is not discussed in \cite{Amb06}. While this is not inconsistent with the stated results, since the result only claims to count oracle calls to ${\cal O}_x$ and ${\cal O}_{\A}$, this is not true of subsequent work that uses the results of \cite{Amb06}. We suspect that an argument like ours could also be made in previous work, but feel it is sufficiently non-trivial that it should not be taken for granted.} That would make such a reflection very costly for algorithms with high space complexity, such as the element distinctness algorithm of~\cite{Amb07}. Specifically, if the time complexity of $\A$ is polynomially related to the number of qubits used in $\A$, then we find that $\log(n) + \log(|\W|) = \Theta(\mathrm{poly}(T))$. In this section, we explain how we circumvent this polynomial dependence using implementing subspaces.

Unfortunately, we are not able to provide an exact implementing subspace. Instead, we will use an \textit{approximate implementing subspace}, i.e., we define a subspace $\H_x \subseteq \H$ and we prove that all operations map states in $\H_x$ to states that have high overlap with $\H_x$. The way we handle the propagation of errors is similar to standard approximation arguments: if the overlap with $\H_x^{\perp}$ after one approximate operation is at most $\delta$, then the combined error after $N$ such approximation operations is at most $\O(N\delta)$. Hence, if we make sure that $\delta < o(1/N)$, then the total cumulative ``lost amplitude'' is $o(1)$, and the influence on the final success probability of the algorithm is at most $o(1)$ as well.

Now, we work towards the formal definition of $\H_x$. For each $x\in\{0,1\}^n$ and $t\in \{0,\dots, T-1\}$, define
\begin{align*}
	\ket{\widetilde{\Psi}_t(x)} = U_{t+1}^\dagger\cdots U_T^\dagger\ket{\Psi_T},
\end{align*}
where $\ket{\Psi_T}$ is the final accepting state, and let $\ket{\widetilde{\Psi}_T(x)} = \ket{\Psi_T}$.
Intuitively, if $f(x)=1$, then $\ket{\Psi_t(x)}$ and $\ket{\widetilde{\Psi}_t(x)}$ are close because the final state of the algorithm is close to $\ket{\Psi_T}$; and if $f(x)=0$, $\ket{\Psi_t(x)}$ and $\ket{\widetilde{\Psi}_t(x)}$ are nearly orthogonal, but $\ket{\Psi_t(x)}$ is close to $(I\otimes X)\ket{\widetilde{\Psi}_t(x)}$, where $X$ acts on the answer register, because the final state of the algorithm is close to $(I\otimes X)\ket{\Psi_T}$. We formalize this in \cref{lem:almost-orthog}. Crucially, the states $\ket{\widetilde{\Psi}_t(x)}$, like the $\ket{\Psi_t(x)}$, have the property that $U_{t+1}\ket{\widetilde{\Psi}_t(x)} = \ket{\widetilde{\Psi}_{t+1}(x)}$.

\begin{lemma}\label{lem:almost-orthog}
	For any clean quantum algorithm $\A$, all $x \in \{0,1\}^n$, and all $t\in [T]_0$:
	\begin{align*}
	\braket{\Psi_t(x)}{\widetilde{\Psi}_t(x)}&=p_1(x)\\
	\mbox{and}\quad \bra{\Psi_t(x)}(I\otimes X)\ket{\widetilde{\Psi}_t(x)}&=p_0(x),
	\end{align*}
	where $p_b(x)$ is the probability that $\A$ outputs $b$ on input $x$, and the Pauli $X$ in $I\otimes X$ acts on the answer register of $\A$.
\end{lemma}

\begin{proof}
	The first statement follows from \cref{def:QA} since
	\begin{align*}
		p_1(x) &= \braket{\Psi_T(x)}{\Psi_T} = \bra{\Psi_t(x)}U_{t+1}^\dagger\dots U_T^\dagger \ket{\Psi_T} = \braket{\Psi_t(x)}{\widetilde{\Psi}_t(x)}.
	\end{align*}
	For the second statement, we use the facts, from \cref{def:QA}, that an $X$ gate on the answer register commutes with every $U_t$, and $\bra{\Psi_T}(I\otimes X)\ket{\Psi_T(x)}=p_0(x)$:
	\begin{align*}
		p_0(x) &= \bra{\Psi_T(x)}(I\otimes X)\ket{\Psi_T}=\bra{\Psi_t(x)}U_{t+1}^\dagger\dots U_T^\dagger (I\otimes X) U_T\dots U_{t+1}\ket{\widetilde{\Psi}_t(x)}\\
		&= \bra{\Psi_t(x)}(I\otimes X)\ket{\widetilde{\Psi}_t(x)},
	\end{align*}
	completing the second part of the proof.
\end{proof}

We bundle the vectors $\ket{\Psi_t(x)}$ and $\ket{\widetilde{\Psi}_t(x)}$ into a space $\H_x$ in the following definition.

\begin{definition}[Implemeting subspace for $\A$]
	\label{def:consistent_subspace}
	Let $x \in X$. We first define two subspaces of $\H$ as
	\begin{align*}
		\overline{\H}_x &= \Span\{\ket{t}\ket{0}\ket{\Psi_t(x)}:t\in [T]_0, t+1\not\in {\cal S}\}\oplus\Span\{\ket{t}\ket{+}\ket{\Psi_t(x)},\ket{t}\ket{-}\ket{\Psi_{t+1}(x)}: t\in [T]_0, t+1\in {\cal S}\}, \\
		\widetilde{\H}_x &= \Span\{\ket{t}\ket{0}\ket{\widetilde{\Psi}_t(x)}:t\in [T]_0, t+1\not\in {\cal S}\}\oplus\Span\{\ket{t}\ket{+}\ket{\widetilde{\Psi}_t(x)},\ket{t}\ket{-}\ket{\widetilde{\Psi}_{t+1}(x)}: t\in [T]_0, t+1\in {\cal S}\}.
	\end{align*}
	Next, we let
	\begin{equation}
		\H_x = \begin{cases}
			\overline{\H}_x, & \text{if } f(x) = 1, \\
			\overline{\H}_x + \widetilde{\H}_x, & \text{if } f(x) = 0.
		\end{cases}
		\label{eq:implementing_subspace}
	\end{equation}
\end{definition}

The intuition behind our use of implementing subspaces is that, for a given input $x \in X$, the state of the algorithm moves through the Hilbert space in a simple, one-dimensional path. So, given a time step $t \in [T]_0$ and an input $x \in X$, we can deduce what the corresponding state in algorithm $\A$ must be at that time step, and hence we can deduce the state in the last register of $\H$. The only difficulty arises at the query time steps, where we use the state before the query when the first two registers are in state $\ket{t}\ket{+}$, and the state after the query when the first two registers are in state $\ket{t}\ket{-}$.

Note that $\overline{\H}_x$ and $\widetilde{\H}_x$ are very close to one another when $x$ is a positive instance. This is because $\ket{\Psi_T(x)}$ is very close to $\ket{\Psi_T}$ when $f(x) = 1$. Hence, in this case it makes sense to only take one of the spaces as the implementing subspace in \cref{eq:implementing_subspace}. On the other hand, if $f(x) = 0$, the two spaces $\overline{\H}_x$ and $\widetilde{\H}_x$ are almost orthogonal, and hence we take both.

We now prove that the newly defined subspace $\H_x$ satisfies the first three properties of \cref{def:implementing_subspace}, and that it satisfies the fourth property approximately. We easily see that $\ket{0}\ket{0}\ket{\Psi_0} \in \overline{\H}_x \subseteq \H_x$, and hence the third condition is satisfied. Next, by the following lemma, we see that the reflection around $\Ker(A)$ leaves the implementing subspace invariant.

\begin{lemma}
	For all $x \in X$, we have $\Pi_{\Ker(A)} \H_x \subseteq \H_x$.
\end{lemma}

\begin{proof}
	Let $\ket{h} \in \H_x$. Observe from \cref{lem:kernel} that
	\[\Pi_{\Ker(A)} = \sum_{\ell=2}^S \frac{\Phi_{\ell}\Phi_{\ell}^{\dagger}}{1 + \frac{q_{\ell}-q_{\ell-1}-1}{M^2}} + \frac{\Phi_{S+1}\Phi_{S+1}^{\dagger}}{\frac12 + \frac{T-q_S}{M^2} + \frac{1}{a^2}}.\]
	Moreover, for any $\ell \in \{2,\dots,S+1\}$, we have that $\Phi_{\ell}\ket{\Psi_{q_{\ell-1}}(x)} \in \overline{\H}_x$, and $\Phi_{\ell}\ket{\widetilde{\Psi}_{q_{\ell-1}}(x)} \in \widetilde{\H}_x$. Thus, if $\ket{h} \in \overline{\H}_x$, we have for some constants $\alpha_t \in \C$ with $t \in [q_{\ell-1}-1, q_{\ell}-1]$,
	\begin{align*}
		\Phi_{\ell}^{\dagger}\ket{h} &= \Phi_{\ell}^{\dagger}\left[\alpha_{q_{\ell-1}-1}\ket{q_{\ell-1}-1}\frac{\ket{-}}{\sqrt{2}}\ket{\Psi_{q_{\ell-1}}(x)} + \sum_{t=q_{\ell-1}}^{q_{\ell}-2} \alpha_t \ket{t}\ket{0}\ket{\Psi_t(x)} + \alpha_{q_{\ell}-1}\ket{q_{\ell}-1}\frac{\ket{+}}{\sqrt{2}}\ket{\Psi_{q_{\ell}-1}(x)}\right] \\
		&= \left[\frac{1}{\sqrt{2}}\alpha_{q_{\ell-1}-1} + \frac{1}{M}\sum_{t=q_{\ell-1}}^{q_{\ell}-2} \alpha_t + \frac{1}{\sqrt{2}}\alpha_{q_{\ell}-1}\right]\ket{\Psi_{q_{\ell-1}}(x)},
	\end{align*}
	and similarly if $\ket{h} \in \widetilde{\H}_x$ then $\Phi_{\ell}^{\dagger}\ket{h}$ is a multiple of $\ket{\widetilde{\Psi}_{q_{\ell-1}}(x)}$. In particular, by linearity this implies that $\Pi_{\Ker(A)}\ket{h}$ is in $\overline{\H}_x$ if $\ket{h} \in \overline{\H}_x$ and in $\widetilde{\H}_x$ if $\ket{h} \in \widetilde{\H}_x$. Linearity for the case where $f(x) = 0$ completes the proof.
\end{proof}

Next, we show that the reflection around $\H(x)$ also leaves the implementing subspace invariant.

\begin{lemma}
	For all $x \in X$, we have $\Pi_{\H(x)}\H_x \subseteq \H_x$.
\end{lemma}

\begin{proof}
	Let $\ket{h} \in \H_x$. We will treat the cases where $\ket{h} \in \overline{\H}_x$ and $\ket{h} \in \widetilde{\H}_x$ separately. The result then follows by linearity.

	Suppose that $\ket{h} \in \overline{\H}_x$. If $\ket{h} = \ket{t}\ket{0}\ket{\Psi_t(x)}$ where $t + 1 \not\in \S$, then $\ket{h} \in \H(x)$ and hence $\Pi_{\H(x)}\ket{h} = \ket{h} \in \overline{\H}_x$. On the other hand, if $\ket{h} = \ket{t}\ket{+}\ket{\Psi_t(x)}$, then
	\begin{align*}
		\left(2\Pi_{\H(x)} - \I\right)\ket{t}\ket{+}\ket{\Psi_t(x)} &= \left(2\Pi_{\H(x)} - \I\right)\ket{t}\ket{+}\sum_{i \in [n], j \in \mathcal{W}}\alpha_{i,j}\ket{i,j} = \ket{t}\sum_{i \in [n], j \in \mathcal{W}} \left(\sqrt{2}\ket{x_i} - \ket{+}\right)\ket{i,j} \\
		&= \ket{t}\sum_{i \in [n], j \in \mathcal{W}} (-1)^{x_i}\ket{-}\ket{i,j} = \ket{t}\ket{-} \sum_{i \in [n], j \in \mathcal{W}} \O_x \ket{i,j} = \ket{t}\ket{-}(\O_x \otimes \I)\ket{\Psi_t(x)} \\
		&= \ket{t}\ket{-}\ket{\Psi_{t+1}(x)} \in \overline{\H}_x.
	\end{align*}
	Similarly, $(2\Pi_{\H(x)}-\I)\ket{t}\ket{-}\ket{\Psi_{t+1}(x)} = \ket{t}\ket{+}\ket{\Psi_t(x)} \in \overline{\H}_x$. Hence, the image of $\overline{\H}_x$ under $2\Pi_{\H(x)} - \I$ is contained in $\overline{\H}_x$, from which we can deduce that this also holds for $\Pi_{\H(x)}$.
	With a similar argument, we can also prove that $\Pi_{\H(x)}$ maps $\widetilde{\H}_x$ to $\widetilde{\H}_x$. By linearity, we can conclude the proof.
\end{proof}

Now, we turn to the fourth property of implementing subspaces from \cref{def:implementing_subspace} and prove in \cref{lem:approximate_implementing_subspace} that we only satisfy it approximately. This approximation will be used in \cref{lem:TB} to construct a subroutine $\Cal C_{\ket{w_0}}$ that prepares a normalized version of the witness $\ket{w_0}$.

\begin{lemma}
	\label{lem:approximate_implementing_subspace}
	$\norm{\Pi_{\H_x^{\perp}}\ket{w_0}} \leq \sqrt{2\varepsilon}\norm{\ket{w_0}}$.
\end{lemma}

\begin{proof}
	If $f(x) = 0$, then $\ket{w_0}$ (as given by \cref{lem:w0}) is exactly contained in $\H_x$ and the left-hand side is zero.
	If $f(x) = 1$, we can find the following state $\ket{w_0'}  \in \overline{\H}_x$ that is sufficiently close to $\ket{w_0}$:
	\begin{align*}
		\ket{w_0'} &= \frac{1}{M}\sum_{t=0}^{q_1-2} \ket{t}\ket{0}\ket{\Psi_t(x)} + \ket{q_1-1}\left(\frac12\ket{0} + \frac12\ket{1}\right)\ket{\Psi_{q_1-1}(x)}, \\
		&+ \frac{1}{Ca^2+1}\left[\ket{q_S-1}\left(\frac12\ket{0}-\frac12\ket{1}\right)\ket{\Psi_{q_S}(x)} + \frac{1}{M} \sum_{t=q_S}^{T-1} \ket{t}\ket{0}\ket{\Psi_t(x)}\right] \\
		&- \frac{Ca}{Ca^2+1} \ket{T}\ket{0} \ket{\Psi_T(x)} \qquad \text{where} \qquad C = \frac{T-q_S}{M^2} + \frac12.
	\end{align*}
	We immediately see that $\ket{w_0'} \in \overline{\H}_x$, as every term is an element of $\overline{\H}_x$. Moreover,
	\begin{align*}
		\norm{\ket{w_0'} - \ket{w_0}}^2 &= \left(\frac{1}{Ca^2 + 1}\right)^2 \cdot \left[\frac12\norm{\ket{\Psi_{q_S}(x)} - \ket{\widetilde{\Psi}_{q_S}(x)}}^2 + \frac{1}{M^2}\sum_{t=q_S}^{T-1} \norm{\ket{\Psi_t(x)} - \ket{\widetilde{\Psi}_t(x)}}^2\right] \\
		&\;\;\;\; + \left(\frac{Ca}{Ca^2 + 1}\right)^2 \norm{\ket{\Psi_T(x)} - \ket{\widetilde{\Psi}_T(x)}}^2 \\
		&= \left\{\left(\frac{1}{Ca^2+1}\right)^2 \cdot \left[\frac12 + \frac{T-q_S}{M^2}\right] + \left(\frac{Ca}{Ca^2+1}\right)^2\right\}\left(2-2p_1(x)\right) \leq 2\varepsilon\norm{\ket{w_0}}^2,
	\end{align*}
	completing the proof.
\end{proof}

\subsubsection{Implementation of the subroutines}
\label{sec:implementation}

The following four lemmas give sufficiently precise and efficient implementations of the subroutines in \cref{thm:span-program-time}. The proof of \cref{thm:main-time} will follow thereafter.

\begin{lemma}
	\label{lem:TA}
	\begin{restatable}{Iwanttodisregardthisfield}{lemTA}
		Let $\A$ be a clean quantum query algorithm with query complexity $S$ and time complexity $T$. Let $P_{\A}= (\H,\V,A,\ket{\tau})$ be the span program for this algorithm, as in \cref{def:PA}. Then the reflection $2\Pi_{\ker A}-\I$ can be implemented to precision $\delta > 0$ with $\O(T/S)$ calls to $\O_\A$ and $\O_\S$, $\O(\mathrm{polylog}(T))$ auxiliary qubits and a number of extra gates that satisfies
		\[\O\left(\frac{T}{S}\mathrm{polylog}\left(T,\frac{1}{\delta}\right)\right).\]
	\end{restatable}
\end{lemma}

\cref{lem:TA} is proven in \cref{app:TA}.

\begin{lemma}
	\label{lem:TC}
	\begin{restatable}{Iwanttodisregardthisfield}{lemTC}
		Let $\A$ be a clean quantum query algorithm with query complexity $S$, time complexity $T$, and error probability $\varepsilon$. Let $P_{\A}= (\H,\V,A,\ket{\tau})$ be the span program for this algorithm, as in \cref{def:PA}.
		Then the reflection $2\Pi_{\H(x)}-\I$ can be implemented with $\O(1)$ calls to $\O_x$ and $\O_\S$ and auxiliary qubits, and $\O(\mathrm{polylog}(T))$ extra gates.
	\end{restatable}
\end{lemma}

\cref{lem:TC} is proven in \cref{app:TC}.

\begin{lemma}
	\label{lem:TB}
	\begin{restatable}{Iwanttodisregardthisfield}{lemTB}
		Let $\A$ be a clean quantum query algorithm with query complexity $S$, time complexity $T$, and error probability $\varepsilon$. Let $P_{\A}= (\H,\V,A,\ket{\tau})$ be the span program for this algorithm, as in \cref{def:PA}. Let $\Cal C_{\ket{w_0}}$ be a unitary that maps $\ket{0}\ket{0}\ket{\Psi_0}$ to $\ket{w_0}/\norm{\ket{w_0}}$ and approximately preserves the implementing subspace in the sense that
		\[\sup_{\substack{\ket{h} \in \H_x \\ \norm{\ket{h}} = 1}} \norm{\Pi_{\H_x^{\perp}} \Cal C_{\ket{w_0}} \ket{h}} \leq 2\sqrt{2\varepsilon}.\]
		We can implement such $\Cal C_{\ket{w_0}}$ up to error $\delta > 0$ in the operator norm with $\O(T/S)$ calls to $\O_{\A}$ and $\O_{\S}$, $\O(1)$ auxiliary qubits and a number of gates that satisfies
		\[\O\left(\frac{T}{S}\mathrm{polylog}\left(T,\frac{1}{\delta}\right)\right).\]
	\end{restatable}
\end{lemma}

\cref{lem:TB} is proven in \cref{app:TB}.

\begin{restatable}{lemma}{lemConRef}\label{lem:consistent-refl}
	Let $\A$ be a clean quantum query algorithm with query complexity $S$, time complexity $T$, and error probability $\varepsilon\in [0,1/2)$. Let $P_{\A}= (\H,\V,A,\ket{\tau})$ be the span program for this algorithm, as in \cref{def:PA}. We can implement a map $G$ that, when restricted to $\H_x$, is a $4\sqrt{2\varepsilon}$-approximation of $\Cal R_{\ket{0}} = (2\ket{0}\bra{0}-\I)$ in operator norm, with $\O(1)$ auxiliary qubits and $\O\left(\mathrm{polylog}(T)\right)$ gates.
\end{restatable}

\cref{lem:consistent-refl} is proven in \cref{app:ConRef}. Now, we are ready to prove \cref{thm:main-time}.

\begin{proof}[Proof of \cref{thm:main-time}]
	Assume, for now, that $\varepsilon = o(1/S^2)$. In the final paragraph of this proof, we will lift this restriction.

	From \cref{thm:alg-span-complexity}, we know that $P_{\A}$ positively $5\varepsilon$-approximates $f$ with complexity $C(P_{\A}) = \O(S)$. Hence, we deduce from \cref{thm:span-program-time} that we can implement the algorithm compiled from $P_{\A}$ with a number of calls to the subroutines $\Cal R_{\Ker(A)}$, $\Cal C_{\ket{w_0}}$, $\Cal R_{\H(x)}$ and $\Cal R_{\ket{0}}$ that goes like
	\[\O\left(\frac{C(P_{\A})}{(1-2\lambda)^{3/2}} \log\frac{1}{1-2\lambda}\right) = \O(S).\]
	For $\Cal R_{\Ker(A)}$ and $\Cal C_{\ket{w_0}}$ we choose precision parameter $\delta = \Theta(\sqrt{\varepsilon})$, which implies that $\log(1/\delta) = \O(\log(S))$. Given their respective query, space and time complexities in \cref{lem:TA,lem:TB,lem:TC,lem:consistent-refl}, we can implement the span program algorithm with $\O(S)$ calls to $\O_x$, $\O(T)$ calls $\O_\A$ and $\O_\S$, $\O(\mathrm{polylog}(T))$ extra qubits, and $\O\left(T\mathrm{polylog}(T)\right)$ extra gates.

	We proceed by analyzing the error introduced by our approximate implementation of the subroutines. First, an error is introduced due to the mapping $\Cal C_{\ket{w_0}}$ not leaving $\H_x$ exactly invariant. Observe that whenever we call $\Cal C_{\ket{w_0}}$, we are moving a part of the state outside of $\H_x$ that has amplitude at most $2\sqrt{2\varepsilon}$. So, there exists a state in $\H_x$ that is $2\sqrt{2\varepsilon}$-close to the state that we used in the analysis of the algorithm in \cref{thm:span-program-time}, and we map to a state that is in turn $2\sqrt{2\varepsilon}$-close to this state. Hence, the total error introduced per call of $\Cal C_{\ket{w_0}}$ is $4\sqrt{2\varepsilon}$. Thus, the total error introduced is $\O(S \cdot 4\sqrt{2\varepsilon}) = o(1)$.

	Additional error is introduced by the approximate implementations of $\Cal R_{\ket{0}}$ and $\Cal R_{\Ker(A)}$. Both are implemented up to precision $\O(\sqrt{\varepsilon})$ in the operator norm, which means that the total cumulative error is at most $\O(S \cdot \sqrt{\varepsilon}) = o(1)$ as well. This completes the proof for the case where $\varepsilon = o(1/S^2)$.

	Finally, if the initial algorithm does not have error probability $\varepsilon = o(1/S^2)$, then we can boost the success probability first. One possible way to do this is to run amplitude estimation to determine with probability at least $1/S^2$ whether $p_0(x)$ is bigger than $2/3$ or smaller than $1/3$. This can be done with $\O(\log(S))$ calls to the original algorithm and reflections through the all-zeros state. For the reflection around the all-zeros state, which needs to be implemented on $\O(k)$ qubits, we now cannot make use of the implementing subspace, as it is not precise enough. But we only need to do this a total of $\O(\log(S))$ times, so we can afford to spend $\O(T\log^C(T))$ gates for some constant $C > 0$. If $T = k^{1+\Omega(1)}$, then we can construct such a circuit with $\O(\mathrm{polylog}(T))$ auxiliary qubits using a divide-and-conquer approach. In the worst case, i.e., when $T = \Theta(k)$, such a circuit can be constructed with $\O(k^{\xi(k)})$ auxiliary qubits, where $\xi(k) = \log\log\log^C(k)/\log(k) + 1/\log\log^C(k) = o(1)$. This completes the proof.
\end{proof}

\section{Application to variable time search}
\label{sec:variable_time_search}

One reason for converting quantum algorithms to span programs is that span programs compose very nicely (see \cite{Rei09} for a number of examples). We illustrate this by describing a construction that, given $n$ span programs for $n$ functions $\{f_j:\{0,1\}^{m_j}\rightarrow\{0,1\}\}_{j=1}^n$, outputs a span program for the logical OR of their output: $f(x^{(1)},\dotsc,x^{(n)}) = \bigvee_{j=1}^n f_j(x^{(j)})$. In short, we show that given query-, time- and space-efficient quantum implementations for each $f_j$, the resulting span program can also be implemented query-, time- and space-efficiently. The full theorem statement is provided below.
Note that throughout this section for the sake of simplicity we write $f_j$ as functions on $\{0,1\}^{m_j}$ even though the results also hold for partial Boolean functions with arbitrary domains $X_j \subseteq \{0,1\}^{m_j}$.

\begin{theorem}[Variable-time quantum search]\label{thm:variable-time}
	Let $\A = \{{\cal A}^{(j)}\}_{j=1}^n$ be a finite set of quantum algorithms, where ${\cal A}^{(j)}$ acts on $k_j \leq k_{\max}$ qubits and decides $f_j : \{0,1\}^{m_j} \rightarrow \{0,1\}$ with bounded error with query complexity $S_j$ and time complexity $T_j \leq T_{\max}$. Suppose that we have uniform access to the algorithms in $\A$ through the oracles $\O_\A$, $\O_\S$ and $\O_x$, as elaborated upon in \cref{sec:model}. Then we can implement a quantum algorithm that decides $f = \bigvee_{j=1}^n f_j$ with bounded error, with the following properties:
	\begin{enumerate}
		\setlength\itemsep{-.4em}
		\item The number of calls to $\O_x$ is $\O\left(\sqrt{\sum_{j=1}^n S_j^2} \cdot \log\left(\sum_{j=1}^n S_j^2\right)\right)$.
		\item The number of calls to $\O_{\A}$ and $\O_{\S}$ is $\O\left(\sqrt{\sum_{j=1}^n T_j^2} \cdot \log\left(\sum_{j=1}^n S_j^2\right)\right)$.
		\item The number of extra gates is $\O\left(\sqrt{\sum_{j=1}^n T_j^2} \cdot \mathrm{polylog}(T_{\max},n)\right)$.
		\item The number of auxiliary qubits is $\O\left(\mathrm{polylog}(T_{\max},n) + k_{\max}^{o(1)}\right)$.
	\end{enumerate}
	If we additionally require that the error probabilities of the $\A^{(j)}$'s are all $o(1/\sum_{j=1}^n S_j^2)$, then the $\log(\sum_{j=1}^n S_j^2)$ factors and the $k_{\max}^{o(1)}$ term can be dropped. We can also drop the term $k_{\max}^{o(1)}$ if $T_j = k_j^{1+\Omega(1)}$ for all $j \in [n]$.
\end{theorem}

A similar result was reached by Ambainis in \cite{Amb06}. Let us discuss how our result compares to that of Ambainis.

First, we assume the uniform access model described in \cref{sec:model}. This is a slight generalization of the model considered by Ambainis, as explained in \cite[Appendix A]{Amb06}, because we differentiate between query and non-query time steps in the algorithms $\A^{(j)}$, whereas Ambainis does not. Therefore, Ambainis only considers the algorithm oracle $\O_\A$ and includes the queries to $x$ as part of this oracle, whereas we also assume to have explicit access to the oracles $\O_\S$ and $\O_x$.

One can obtain some of our results using Ambainis's construction and subsequently converting the resulting algorithm back to our setting. For instance, if one counts every query in the original algorithms as having unit cost, then Ambainis's construction yields an algorithm that evaluates $f = \bigvee_{j=1}^n f_j$ with $\O(\sqrt{\sum_{j=1}^n S_j^2})$ queries to $\O_x$. This is a logarithmic factor better than our result, but the number of calls to $\O_\A$ and $\O_\S$ is unclear, and the time and space complexities are not analyzed.

Alternatively, if one assigns a unit cost to every gate in the original algorithms, then the algorithm that follows from Ambainis's construction performs $\O(\sqrt{\sum_{j=1}^n T_j^2})$ calls to $\O_\A$, $\O_\S$ and $\O_x$. Similarly as before, this is a logarithmic factor better in the scaling of the query complexity to $\O_\A$, but worse in the query complexity to $\O_x$ and again the time and space complexities are not analyzed.

Our improvement over Ambainis's work consists of the following elements. First, we show that one can attain both desired scalings in the number of calls to $\O_x$, $\O_\S$ and $\O_\A$ simultaneously, up to a single logarithmic factor. Second, our construction is also efficient with respect to the time and space complexities, as we show that we only suffer from polylogarithmic overhead in the number of extra gates and auxiliary qubits.

There are, however, some aspects to Ambainis's work that we did not reproduce. Ambainis proved a version of his theorem for the search problem: find $j$ such that $f_j(x^{(j)})=1$, whereas we only consider a decision version. By a standard reduction from the search version to the decision version, we also recover the analogous search result, but with an extra factor of $\log(n)$ overhead in the query and time complexities.

Ambainis also gives a result for the case where the costs of the original algorithms are unknown. It would be interesting to figure out whether our results can be similarly modified in the case where we do not know $\{T_j\}_{j=1}^n$ and/or $\{S_j\}_{j=1}^n$, but it is not immediately clear to us how one would go about this. We leave this for future research.

From \cref{thm:variable-time} we easily deduce that if we have efficient uniform access to a set of algorithms, i.e., the oracles $\O_{\A}$ and $\O_{\S}$ can be implemented in time logarithmic in $T_{\max}$ and $n$, then the algorithm compiled from $P$ has query complexity $\widetilde{\O}(\sqrt{\sum_{j=1}^n S_j^2})$ and time complexity $\widetilde{\O}(\sqrt{\sum_{j=1}^n T_j^2})$.

The remainder of this section is dedicated to proving \cref{thm:variable-time}. In \cref{sec:or-span} we describe how we can merge $n$ span programs $P^{(1)}, \dots, P^{(n)}$ evaluating functions $f_1, \dots, f_n$, respectively, into one span program $P$ evaluating the OR of these functions, $f_1 \lor \cdots \lor f_n$. Subsequently, in \cref{sec:or-span-implementation}, we relate the implementation of the algorithm compiled from $P$ to the implementation of the algorithms compiled from the individual $P^{(j)}$'s. Finally, in \cref{sec:var-time-search}, we specialize the span programs $P^{(j)}$ to be span programs of algorithms, and we relate the implementation of the required subroutines to the constructions in \cref{sec:main_thm}, completing the proof of \cref{thm:variable-time}.

\subsection{The OR of span programs}
\label{sec:or-span}

Fix $\lambda\in (0,1/n)$. For $j\in [n]$, let $P^{(j)} = (\H^{(j)},\V^{(j)},A^{(j)},\ket{\tau^{(j)}})$ be a span program on $\{0,1\}^{m_j}$ that positively $\lambda$-approximates $f_j : \{0,1\}^{m_j} \rightarrow \{0,1\}$.\footnote{We require $\lambda$ to be quite small here. One way to achieve this from an arbitrary span program is to convert it to an algorithm, reduce the error to $\O(1/n)$ at the expense of a $\O(\log n)$ multiplicative factor, and then convert that back to a span program using the construction in \cref{sec:span_program}. Furthermore, we can just as well use partial functions here, but we don't for notational simplicity.} Let $W_+^{(j)}$ and $W_-^{(j)}$ be some upper bounds on $W_+(P^{(j)})$ and $\widetilde{W}_-(P^{(j)})$ respectively, and assume that every $x \in f_j^{-1}(0)$ has an approximate negative witness $\ket{\tilde\omega^{(j)}} \in \V^{(j)}$ with $\norm{\bra{\tilde\omega^{(j)}} A^{(j)} \Pi_{\H^{(j)}(x)}}^2 \leq \lambda/W_+^{(j)}$ and $\norm{\bra{\tilde\omega^{(j)}} A^{(j)}}^2 \leq {W}_-^{(j)}$. Let $C_j = \sqrt{W_+^{(j)}W_-^{(j)}}$.

Assume, by applying an appropriate basis change, that $\ket{\tau^{(j)}}=\ket{0}$ for every $j \in [n]$. For each $j$, extend
$\ket{\tau^{(j)}}=\ket{0}$ to an orthonormal basis $\{\ket{0},\ket{j,1},\dots,\ket{j,\dim(\V^{(j)})-1}\}$ for $\V^{(j)}$ so that, aside from the single overlapping dimension $\ket{0}$, the subspaces $\V^{(j)}$ are orthogonal to one another. Let $\overline{\V}^{(j)}=\mathrm{span}\{\ket{j,1},\dots,\ket{j,\dim(\V^{(j)})-1}\}$, so that $\V^{(j)}=\mathrm{span}\{\ket{0}\}\oplus \overline{\V}^{(j)}$.

Let $f:\{0,1\}^{m_1+\dots+m_n}\rightarrow \{0,1\}$ be the function defined by $f(x^{(1)},\dots,x^{(n)})=\bigvee_{j=1}^nf_j(x^{(j)})$. We can define a span program $P$ on $\{0,1\}^{m_1+\dots+m_n}$ that decides $f$ as follows:
\begin{align}
	\forall j\in[n], \ell\in [m_j], b\in\{0,1\}, \quad \H_{j,\ell,b} =\Span\{\ket{j}\}\otimes \H_{\ell,b}^{(j)}, \quad \H_{\text{true}} = \bigoplus_{j=1}^n \H_{\text{true}}^{(j)}, \quad \H_{\text{false}} &= \Span\{\ket{0,0}\} \nonumber \\
	\V =\mathrm{span}\{\ket{0}\}\oplus \bigoplus_{j=1}^n\overline{\V}^{(j)}, \qquad \ket{\tau} = \ket{0}, \qquad A = \sum_{j=1}^n\sqrt{W_+^{(j)}}\bra{j}\otimes A^{(j)}.& \label{eq:or-span}
\end{align}
Above, we are indexing into an input $x\in \{0,1\}^{m_1+\dots+m_n}$ by using a pair of indices, $j\in [n]$ and $\ell\in [m_j]$, in the obvious way.
From this definition of $P$, we get:
\begin{align}
	\H(x) &= \bigoplus_{j \in [n]} \Span\{\ket{j}\}\otimes \H^{(j)}(x^{(j)}), \qquad \mbox{where} \qquad \forall j \in [n], \quad \H^{(j)}(x^{(j)}) = \bigoplus_{\ell\in[m_j]} \H^{(j)}_{\ell,x^{(j)}_{\ell}}.\label{eq:Hx-form}
\end{align}

\begin{definition}\label{def:or-span-program}
	Let $\{P^{(j)}\}_{j=1}^n$ be a set of span programs, where $P^{(j)} = (\H^{(j)},\V^{(j)},A^{(j)},\ket{\tau^{(j)}})$. Then we let $P = (\H,\V,A,\ket{\tau})$ be the \emph{OR span program} of these span programs, where $\H$, $\V$, $A$ and $\ket{\tau}$ are defined in \cref{eq:or-span,eq:Hx-form}.
\end{definition}

We proceed by proving various properties of the newly-defined OR span program. First, we prove that it indeed evaluates $f$ in the following theorem.

\begin{theorem}\label{thm:or-span-program}
The span program $P$ positively $n\lambda$-approximates $f$ with complexity $C(P)\leq\sqrt{\sum_{j=1}^nC_j^2}$.
\end{theorem}

The proof will follow from \cref{lem:or-positive,lem:or-negative}. First, we show that if $f(x)=1$, $P$ accepts $x$, and give an upper bound on the positive witness complexity.

\begin{lemma}\label{lem:or-positive}
	If $f(x)=1$, then the span program $P$ accepts $x$, with positive witness complexity $w_+(x)\leq 1$. Thus $W_+(P)\leq 1$.
\end{lemma}

\begin{proof}
	If $f(x)=1$, then there exists $j\in [n]$ such that $f_j(x^{(j)})=1$, so let $\ket{w^{(j)}} \in \H^{(j)}(x^{(j)})$ be a positive witness for $x^{(j)}$ in $P^{(j)}$ with $\norm{\ket{w^{(j)}}}^2 \leq W_+^{(j)}$. Then let $\ket{w}=\frac{1}{\sqrt{W_+^{(j)}}}\ket{j}\ket{w^{(j)}} \in \V$. Then $A\ket{w} = A^{(j)}\ket{w^{(j)}} = \ket{0}$. Furthermore, since $\H(x) = \bigoplus_{j=1}^n \Span\{\ket{j}\} \otimes \H^{(j)}(x^{(j)})$ by \cref{eq:Hx-form}, and $\ket{w^{(j)}} \in \H^{(j)}(x^{(j)})$, we have $\ket{w} \in \H(x)$, so $\ket{w}$ is a positive witness for $x$. Since $\norm{\ket{w^{(j)}}}^2 \leq W_+^{(j)}$, $\ket{w}$ has complexity $\norm{\ket{w}}^2 \leq 1$.
\end{proof}

We complete the proof of \cref{thm:or-span-program} by exhibiting approximate negative witnesses.

\begin{lemma}\label{lem:or-negative}
	If $f(x) = 0$, then there is an approximate negative witness $\ket{\widetilde\omega}$ with $\norm{\bra{\widetilde\omega}A\Pi_{\H(x)}}^2\leq n\lambda/W_+(P)$ and $\norm{\bra{\widetilde\omega}A}^2\leq \sum_{j=1}^n C_j^2$, so $P$ positively $n\lambda$-approximates $f$, and $\widetilde{W}_-(P)\leq \sum_{j=1}^n C_j^2$.
\end{lemma}

\begin{proof}
	If $f(x) = 0$, then it must be the case that for all $j \in [n]$, $f_j(x^{(j)}) = 0$, so for each $j$, let $\ket{\widetilde\omega^{(j)}}$ be an approximate negative witness for $x^{(j)}$ in $P^{(j)}$ with $\norm{\bra{\widetilde\omega^{(j)}} A^{(j)} \Pi_{\H^{(j)}(x^{(j)})}}^2 \leq \lambda/W_+^{(j)}$, and $\norm{\bra{\widetilde\omega^{(j)}}A^{(j)}}^2 \leq W_-^{(j)}$. For each $j$, we can write $\bra{\widetilde\omega^{(j)}} = \bra{0} + \bra{\overline\omega^{(j)}}$ for some $\ket{\overline\omega^{(j)}} \in \overline{\V}^{(j)}$. We define $\bra{\widetilde\omega} = \bra{0} + \sum_{j=1}^n \bra{\overline\omega^{(j)}}$. Then $\braket{\widetilde\omega}{\tau} = \braket{\widetilde\omega}{0} = 1$. Furthermore, for each $j$, since the column space of $A^{(j)}$ is in $\V^{(j)} = \mathrm{span}\{\ket{0}\} \oplus \overline{\V}^{(j)}$, we have $\bra{\widetilde\omega} A^{(j)} = (\bra{0} + \bra{\overline\omega^{(j)}}) A^{(j)} = \bra{\widetilde\omega^{(j)}} A^{(j)}$. Since $\H(x) = \bigoplus_{j=1}^n \Span\{\ket{j}\} \otimes \H^{(j)}(x^{(j)})$ by \cref{eq:Hx-form},
	\[\norm{\bra{\widetilde\omega} A\Pi_{\H(x)}}^2=\norm{\sum_{j=1}^n \sqrt{W_+^{(j)}}\bra{j}\otimes (\bra{\widetilde\omega^{(j)}} A^{(j)} \Pi_{\H^{(j)}(x^{(j)})})}^2 \leq \sum_{j=1}^n W_+^{(j)} \frac{\lambda}{W_+^{(j)}} = n\lambda,\]
	so $P$ positively $n\lambda$-approximates $f$. Finally, we conclude $\widetilde{W}_-(P) \leq \sum_{j=1}^n C_j^2$ by observing:
	\begin{equation*}
		\norm{\bra{\widetilde\omega}A}^2 = \sum_{j=1}^n W_+^{(j)} \norm{\bra{\widetilde\omega^{(j)}} A^{(j)}}^2 \leq \sum_{j=1}^n W_+^{(j)}W_-^{(j)} = \sum_{j=1}^n C_j^2.\qedhere
	\end{equation*}
\end{proof}

We conclude this section by characterizing the minimal positive witness $\ket{w_0}$ and the kernel of $A$ in the following two lemmas. The proofs are straightforward and can be found in Appendix \ref{app:OR}.

\begin{restatable}{lemma}{ORw}
\label{lem:or-w0}
	The minimal positive witness of $P$ is given by
	\[\ket{w_0} = \frac{1}{\norm{\alpha}^2} \sum_{j=1}^n \alpha_j \ket{j} \otimes \frac{\ket{w_0^{(j)}}}{\norm{\ket{w_0^{(j)}}}}, \qquad \text{where} \qquad \alpha_j = \frac{\sqrt{W_+^{(j)}}}{\norm{\ket{w_0^{(j)}}}}\]
	and $\ket{w_0^{(j)}}$ are the minimal witnesses of $P^{(j)}$. Moreover, the minimal witness size is $N = 1/\norm{\alpha}^2$.
\end{restatable}

The $\ker(A)$ is not just the union of kernels of each $A^{(j)}$ because just as we can combine the minimal witnesses $\ket{w_0^{(j)}}$ to map to $\ket{0}$, we can make a combination that maps to $0$. The following lemma characterizes such combinations of individual minimal witnesses and finds that they are orthogonal to the minimal witness for $P$.

\begin{restatable}{lemma}{ORkernel}
\label{lem:or-kernel}
	Let $K = \Span\{\ket{j}\ket{w_0^{(j)}}:j\in[n]\}\cap \Span\{\ket{w_0}\}^{\bot}$.
	The kernel of $A$ is given by
	\[\Ker(A) = \Span\{\ket{0,0}\} \oplus K \oplus \bigoplus_{j=1}^n \Span\{\ket{j}\} \otimes \Ker(A^{(j)})\]
\end{restatable}

\subsection{Implementation of the OR span program}
\label{sec:or-span-implementation}

Now that we have formally defined the OR span program in \cref{def:or-span-program}, we proceed by analyzing the implementation cost of the algorithm compiled from it. To that end, we first of all assume that all of the spaces $\H^{(j)}$ correspond to $M$ qubits, i.e., $\H^{(j)} = \C^{2^M}$ for all $j \in [n]$. This is not much of a restriction, as we can always simply pad the smaller $\H^{(j)}$'s with extra qubits that we don't touch until our space is as big as the largest state space of the individual span programs.

The main idea of this section will be to use \cref{thm:span-program-time}, and give implementations to the required subroutines in terms of the individual span programs $P^{(j)}$. This sometimes requires running several subroutines associated with the individual $P^{(j)}$'s concurrently. We first formalize this idea in the definition below.

\begin{definition}
	\label{def:concurrent-access}
	Let $\Cal{C}^{(1)}, \dots, \Cal{C}^{(n)}$ be quantum subroutines, all acting on the same Hilbert space $\H$. We say that a subroutine $\mathcal{C}$ provides \emph{concurrent access} to $\{\Cal{C}^{(j)}\}_{j=1}^n$ if it performs the following action on $\C^{[n]} \otimes \H$:
	\[\mathcal{C} = \sum_{j=1}^n \ket{j}\bra{j} \otimes \Cal C^{(j)}.\]
\end{definition}

Next, we present the main theorem relating the cost of implementing the span program compiled from $P$ to the cost of implementing the subroutines that are associated with the individual $P^{(j)}$'s.

\begin{lemma}\label{lem:or-span-time}
	Let $\lambda \in [0,1/(2n))$ and let $\{P^{(j)}\}_{j=1}^n$ be a set of positively $\lambda$-approximating span programs. For all $j \in [n]$, let $\ket{w_0^{(j)}}$ be a minimal positive witness for $P^{(j)}$ and let $W_+^{(j)} \geq W_+(P^{(j)})$ and $W_-^{(j)} \geq \widetilde{W}_-(P^{(j)})$ be upper bounds on the positive and negative complexities. Furthermore, for each $j \in [n]$ and $x^{(j)} \in \{0,1\}^{m_j}$, let $\H_{x^{(j)}}$ be an implementing subspace of $P^{(j)}$ for $x^{(j)}$. Let $P$ be the span program described in \cref{eq:or-span} and suppose that we have concurrent access to the following four sets of subroutines (as defined in \cref{def:concurrent-access}):
	\begin{enumerate}
		\setlength\itemsep{-.2em}
		\item A circuit $\mathcal{R}_A$, providing concurrent access to the subroutines $\{\Cal{R}_{\Ker(A^{(j)})}\}_{j=1}^n$, where $\Cal{R}_{\Ker(A^{(j)})}$ acts on $\H_{x^{(j)}}$ as $2\Pi_{\Ker(A^{(j)})} - \I$.
		\item A circuit $\mathcal{C}$, providing concurrent access to the subroutines $\{\Cal C_{\ket{w_0^{(j)}}}\}_{j=1}^n$, where $\Cal C_{\ket{w_0^{(j)}}}$ leaves $\H_{x^{(j)}}$ invariant and maps $\ket{0}$ to $\ket{w_0^{(j)}} / \|\ket{w_0^{(j)}}\|$.
		\item A circuit $\mathcal{R}_\H$, providing concurrent access to the subroutines $\{\Cal R_{\H(x^{(j)})}\}_{j=1}^n$, where $\Cal R_{\H(x^{(j)})}$ acts on $\H_{x^{(j)}}$ as $2\Pi_{\H(x^{(j)})} - \I$.
		\item A circuit $\mathcal{R}_0$, providing concurrent access to the subroutines $\{\Cal R_{\ket{0}}^{(j)}\}_{j=1}^n$, where $\Cal R_{\ket{0}}^{(j)}$ acts on $\H_{x^{(j)}}$ as $2\ket{0}\bra{0} - \I$.
		\item A circuit $\Cal{C}_{\alpha}$ that prepares the superposition
		\[\Cal{C}_{\alpha} : \ket{0} \mapsto \frac{1}{\norm{\alpha}}\sum_{j=1}^n \alpha_j \ket{j} \qquad \text{where} \qquad \alpha_j = \frac{\sqrt{W_+^{(j)}}}{\norm{\ket{w_0^{(j)}}}}.\]
	\end{enumerate}
	Then, we can implement the span program algorithm for $P$ with a number of calls to the aforementioned circuits that satisfies
	\[\O\left(\frac{\sqrt{\sum_{j=1}^n C_j^2}}{(1-2n\lambda)^{3/2}} \log \frac{1}{1-2n\lambda}\right) \qquad \text{where} \qquad C_j = \sqrt{W_+^{(j)}W_-^{(j)}},\]
	and a number of extra gates and auxiliary qubits that satisfies $\O(\mathrm{polylog}(\sqrt{\sum_{j=1}^n C_j^2},1/(1-2n\lambda),n))$.
\end{lemma}

\begin{proof}
	We apply \cref{thm:span-program-time} to $P$. As $P$ is positively $\lambda'$-approximating with $\lambda' = n\lambda < 1/2$, the first requirement is satisfied.

	Next, we define the implementing subspace that we use. We take $\H_x$ to be
	\[\H_x = \Span\{\ket{0,0}\} \oplus \bigoplus_{j=1}^n \Span\{\ket{j}\} \otimes \H_{x^{(j)}}^{(j)},\]
	i.e., we have one orthogonal direction that contains all scalar multiples of the all-zeros state, and all the implementing subspaces associated with the individual $P^{(j)}$'s labeled by $j$. We refer to the first and second registers as the \textit{label register} and \textit{data register}, respectively.

	Now, we turn our attention to the implementation of the four subroutines listed in \cref{thm:span-program-time}. First, we implement the reflection through the $\ket{0}$-state, $\Cal R_{\ket{0}}$. Observe that the all-zeros state in $\H_x$ is the state $\ket{0,0}$. But the only state in $\H_x$ that has zero in the label register is exactly the all-zeros state. Hence, we can simply reflect through $\ket{0}$ on the first register, which has only $\O(\log(n))$ qubits. Thus, we can implement $\Cal R_{\ket{0}}$ in $\O(\log(n))$ gates.

	Next, we turn our attention to the implementation of $\Cal C_{\ket{w_0}}$. From \cref{lem:or-w0} we find that
	\[\frac{\ket{w_0}}{\norm{\ket{w_0}}} = \frac{1}{\norm{\alpha}}\sum_{j=1}^n \alpha_j \ket{j} \otimes \frac{\ket{w_0^{(j)}}}{\norm{\ket{w_0^{(j)}}}}.\]
	This allows for defining $\Cal C_{\ket{w_0}}$ as the following procedure.
	\begin{enumerate}
		\setlength\itemsep{-.4em}
		\item First, we prepare an auxiliary qubit in the state $\ket{1}$ whenever the data register is in the state $\ket{0}$, and $\ket{0}$ otherwise. This requires one controlled call to $\mathcal{R}_0$ alongside with $\O(1)$ auxiliary gates.
		\item Next, conditioned on this auxiliary qubit, we apply $\Cal C_{\alpha}$ to the label register.
		\item Now, we uncompute the auxiliary qubit with the gates from step 1 applied in reverse. This uncomputation succeeds with certainty as the all-zeros states in all the $\H_{x^{(j)}}^{(j)}$'s are the same, and hence permuting the labels effectively permutes between different all-zeros states in the $\H_{x^{(j)}}^{(j)}$'s.
		\item Finally, we call $\mathcal{C}$.
	\end{enumerate}
	We observe that the first three steps perform some unitary on all the $n+1$ states that have the all-zeros state in the data register. As all these states are part of $\H_x$, they leave $\H_x$ invariant. Similarly, the fourth step leaves $\H_x$ invariant as all of the individual subroutines that make up $\mathcal{C}$ leave their respective implementing subspace $\H_{x^{(j)}}^{(j)}$ invariant. Hence, the entire procedure $\Cal C_{\ket{w_0}}$ leaves $\H_x$ invariant.

	Furthermore, observe that if we start in the state $\ket{0,0}$, the mapping that is implemented is the following
	\[\ket{0,0} \overset{\text{steps }1-3}{\mapsto} \frac{1}{\norm{\alpha}}\sum_{j=1}^n \alpha_j \ket{j,0} \overset{\mathcal{C}}{\mapsto} \frac{1}{\norm{\alpha}} \sum_{j=1}^n \alpha_j \ket{j} \otimes \frac{\ket{w_0^{(j)}}}{\norm{\ket{w_0^{(j)}}}}.\]
	Thus, we conclude that we can implement $\Cal C_{\ket{w_0}}$ using $\O(1)$ calls to $\mathcal{R}_0$, $\Cal C_{\alpha}$ and $\mathcal{C}$ and $\O(1)$ extra gates and auxiliary qubits.

	We proceed by providing an implementation of the reflection through the kernel of $A$. To that end, remember from \cref{lem:or-kernel} that
	\[\Ker(A) = \Span\{\ket{0,0}\} \oplus \left[\Span\{\ket{w_0}\}^{\perp} \cap \underset{W_0}{\underbrace{\Span\{\ket{j}\ket{w_0^{(j)}} : j \in [n]\}}}\right] \oplus \bigoplus_{j=1}^n \Span\{\ket{j}\} \otimes \Ker(A^{(j)}).\]
	As $\Span\{\ket{w_0}\} \subseteq W_0$, we observe that
	\[2\Pi_{\Ker(A)} - \I = \left(2\ket{0,0}\bra{0,0} - \I\right)\left(\frac{2\ket{w_0}\bra{w_0}}{\norm{\ket{w_0}}^2} - \I\right)\left(2\Pi_{W_0} - \I\right)\left(\sum_{j=1}^n \ket{j}\bra{j} \otimes \left(2\Pi_{\Ker(A^{(j)})} - \I\right)\right).\]
	The first factor is simply $\Cal R_{\ket{0}}$ on $\H_x$. Similarly, the last factor is exactly the action of $\mathcal{R}_A$ on $\H_x$. The second factor can easily be implemented by the sequence $\Cal C_{\ket{w_0}}\Cal R_{\ket{0}}\Cal C_{\ket{w_0}}^{\dagger}$. So, it remains to implement the third factor, which we can achieve by observing that on $\H_x$ we have
	\[2\Pi_{W_0} - \I = \sum_{j=1}^n \ket{j}\bra{j} \otimes \left(\frac{2\ket{w_0^{(j)}}\bra{w_0^{(j)}}}{\norm{\ket{w_0^{(j)}}}^2} - \I\right) = \sum_{j=1}^n \ket{j}\bra{j} \otimes \left(\Cal C_{\ket{w_0^{(j)}}} \Cal R_{\ket{0}}^{(j)} \Cal C_{\ket{w_0^{(j)}}}^{\dagger}\right) = \mathcal{C}\mathcal{R}_0\mathcal{C}^{\dagger}.\]
	Thus, we have
	\[\Cal R_{\Ker(A)} = \Cal R_{\ket{0}} \Cal C_{\ket{w_0}} \Cal R_{\ket{0}} \Cal C_{\ket{w_0}}^{\dagger}\mathcal{C}\mathcal{R}_0\mathcal{C}^{\dagger}\mathcal{R}_A.\]
	As all the individual factors leave $\H_x$ invariant, so does their product. Hence, we can implement $\Cal R_{\Ker(A)}$ with $\O(1)$ calls to the subroutines mentioned in the statement of the lemma.

	It remains to implement the routine $\Cal R_{\H(x)}$. To that end, observe that
	\[2\Pi_{\H(x)} - \I = \sum_{j=1}^n \ket{j}\bra{j} \otimes \left(2\Pi_{\H^{(j)}(x^{(j)})} - \I\right),\]
	which implies that we can simply implement the reflection through $\H(x)$ with one call to $\mathcal{R}_\H$.

	We have implemented all routines in the statement of \cref{thm:span-program-time} with $\O(1)$ calls to the routines listed in the statement of this lemma. That means that the total number of calls to these routines is equal up to constants to the expression in \cref{thm:span-program-time}, which reduces to
	\[\O\left(\frac{\sqrt{W_+(P)\widetilde{W}_-(P)}}{(1-2n\lambda)^{3/2}}\log\frac{1}{1-2n\lambda}\right) = \O\left(\frac{\sqrt{\sum_{j=1}^n C_j^2}}{(1-2n\lambda)^{3/2}}\log\frac{1}{1-2n\lambda}\right).\]
	Moreover, it follows directly from the statement of \cref{thm:span-program-time} that the total number of extra gates is $\O(\mathrm{polylog}(\sqrt{\sum_{j=1}^n C_j^2},1/(1-2n\lambda)))$. This completes the proof.
\end{proof}

\subsection{Implementation of variable time quantum search}
\label{sec:var-time-search}

In this section, we prove \cref{thm:variable-time}. The core idea is to first convert the algorithms into span programs using the construction from \cref{sec:main_thm}, next merge them into an OR span program as in \cref{def:or-span-program}, and finally convert that back into a quantum algorithm using \cref{lem:or-span-time}.

There is one caveat though. If we naively use the span programs of the algorithms $\A^{(j)}$ from \cref{def:PA} with the upper bounds on the positive witness sizes that follow from \cref{lem:w+}, then we might end up with completely arbitrary coefficients $\alpha_j$ in \cref{lem:or-span-time}, making it too time-consuming to implement $\Cal C_{\alpha}$. We circumvent this using a technique that was already present in Ambainis's original paper \cite{Amb06}, which we dub the \textit{binning technique}. The next two lemmas formalize this idea and their proofs can be found in \cref{app:OR}.

\begin{restatable}{lemma}{binningTech}
	\label{lem:binning-technique}
	Let $0 < \gamma_{\min} = \gamma_1 \leq \dots \leq \gamma_n = \gamma_{\max}$. Then, we can efficiently find a sequence of integers $0 = j_0 \leq \dots \leq j_k = n$ such that $k \leq \lceil\log(\gamma_{\max}/\gamma_{\min})\rceil \cdot \lceil \log(n)\rceil$ and the following two properties hold:
	\begin{enumerate}
		\setlength\itemsep{-.4em}
		\item For all $\ell \in [k]$, $j_{\ell} - j_{\ell-1}$ is a power of $2$.
		\item For all $\ell \in [k]$ and $j \in [j_{\ell-1}+1,j_{\ell}]$,
		\[\frac{\gamma_{j_{\ell}}}{2} \leq \gamma_j \leq \gamma_{j_{\ell}}.\]
	\end{enumerate}
\end{restatable}

The above lemma is nothing more than a statement about how we can put a sequence of positive reals into several bins. We now use it to modify the upper bounds $W_+^{(j)}$, so as to make the cost of implementing $\Cal C_{\alpha}$ scale more favorably.

\begin{restatable}{lemma}{Calpha}
	\label{lem:Calpha}
	Let $\A = \{{\cal A}^{(j)}\}_{j=1}^n$ be a finite set of quantum algorithms, where $\A^{(j)}$ has query complexity $1 \leq S_j \leq S_{\max}$. Let $P^{(j)}$ be the span program of $\A^{(j)}$. Then, we can define positive reals $\{W_+^{(j)}\}_{j=1}^n$ such that $W_+(P^{(j)}) \leq W_+^{(j)} \leq 12(2S_j+1)$, and a sequence of integers $0 = j_0 \leq \cdots \leq j_k = n$ with $k \leq \lceil\frac12\log(6S_{\max})\rceil \cdot \lceil\log(n)\rceil$, such that for every $\ell \in [k]$, $j_{\ell} - j_{\ell-1}$ is a power of $2$ and for every $j \in [j_{\ell-1}+1,j_{\ell}]$,
	\[\alpha_j = \frac{\sqrt{W_+^{(j)}}}{\norm{\ket{w_0^{(j)}}}} = \frac{\sqrt{W_+^{(j_\ell)}}}{\norm{\ket{w_0^{(j_\ell)}}}} = \alpha_{j_{\ell}}.\]
	With this choice of upper bounds $\{W_+^{(j)}\}_{j=1}^n$, we can implement the circuit $\Cal C_{\alpha}$, as defined in \cref{lem:or-span-time}, with $\O(\log(S_{\max})\log^2(n))$ gates and $\O(\log(n))$ auxiliary qubits.
\end{restatable}

Now that we can implement $\Cal C_{\alpha}$ in a number of gates that scales polylogarithmically in both $S_{\max}$ and $n$, we turn our attention to the proof of \cref{thm:variable-time}.

\begin{proof}[Proof of \cref{thm:variable-time}]
	First of all, we consider the case where the $n$ algorithms $\{\A^{(j)}\}_{j=1}^n$ are clean quantum algorithms with error probabilities satisfying $\varepsilon_j < 1/(80n)$ and $\varepsilon_j = o(1/\sum_{j=1}^n S_j^2)$. In the final paragraph of this proof, we will lift this restriction.

	We modify the algorithms $\A^{(1)}, \dots, \A^{(n)}$ slightly. Similar to the proof of \cref{lem:Q_Uniformity}, we insert a sequence $\O_x, \I, \O_x$ into all $\A^{(j)}$'s at a spacing $B$ defined by
	\[B = \left\lceil\sqrt{\frac{\sum_{j=1}^n T_j^2}{\sum_{j=0}^n S_j^2}}\right\rceil.\]
	We denote the algorithm that we obtain after this modification by $\overline{\A}^{(j)}$, and its query and time complexity by $\overline{S}_j$ and $\overline{T}_j$, respectively. Using a similar analysis as in the proof of \cref{lem:Q_Uniformity}, we obtain
	\[\overline{S}_j = \Theta\left(S_j + \frac{T_j}{B}\right), \qquad \overline{T}_j = \Theta(T_j), \qquad \text{and} \qquad \frac{\overline{T}_j}{\overline{S}_j} = \O(B).\]

	Next, we turn these algorithms $\{\overline{\A}^{(j)}\}_{j=1}^n$ into span programs $\{P^{(j)}\}_{j=1}^n$ using \cref{def:PA}. According to \cref{lem:Calpha}, we can define the upper bounds $\{W_+^{(j)}\}_{j=1}^n$ such that
	\[W_+(P^{(j)}) \leq W_+^{(j)} \leq 12(2\overline{S}_j+1) = \O(\overline{S}_j),\]
	and such that we can implement $\Cal{C}_{\alpha}$ in a number of gates and auxiliary qubits that scales polylogarithmically in $S_{\max}$ and $n$. In addition, for all $j \in [n]$ we can take $W_-^{(j)} = \O(\overline{S}_j)$ by virtue of \cref{lem:w-}, which implies that $C_j = \O(\overline{S}_j)$.  From \cref{lem:w-} we have a negative witness $\ket{\tilde{\omega}}$ that satisfies
	\[\norm{\bra{\tilde{\omega}}A^{(j)}\Pi_{\H^{(j)}(x)}}^2 \leq \frac{5\varepsilon_j}{3(2\overline{S}_j+1)} \leq \frac{20\varepsilon_j}{W_+^{(j)}},\]
	which implies that all $P^{(j)}$'s are positive $\lambda$-approximating with $\lambda \leq 20\varepsilon_j < 1/(4n)$.

	We have now shown that we satisfy the requirements for constructing the OR span program $P$, as defined in \cref{def:or-span-program}. According to \cref{thm:or-span-program}, the complexity of this span program is now upper bounded by
	\begin{align*}
		C(P) &\leq \sqrt{\sum_{j=1}^n C_j^2} = \O\left(\sqrt{\sum_{j=1}^n \overline{S}_j^2}\right) = \O\left(\sqrt{\sum_{j=1}^n \left[S_j^2 + \frac{T_j^2}{B^2}\right]}\right) = \O\left(\sum_{j=1}^n S_j^2 + \frac{1}{B^2}\sum_{j=1}^n T_j^2\right) \\
		&= \O\left(\sum_{j=1}^n S_j^2 + \frac{\sum_{j=1}^n S_j^2}{\sum_{j=1}^n T_j^2} \cdot \sum_{j=1}^n T_j^2\right) = \O\left(\sum_{j=1}^n S_j^2\right).
	\end{align*}

	According to \cref{lem:or-span-time}, implementing the algorithm compiled from $P$ takes a number of calls to the subroutines $\mathcal{R}_\A$, $\mathcal{C}$, $\mathcal{R}_\H$ and $\mathcal{R}_0$ that satisfies
	\[\O\left(\frac{\sqrt{\sum_{j=1}^n C_j^2}}{(1-2n\lambda)^{3/2}} \cdot \log\frac{1}{1-2n\lambda}\right) = \O\left(\sqrt{\sum_{j=1}^n S_j^2}\right),\]
	and a number of extra gates and auxiliary qubits that satisfies
	\[\O\left(\mathrm{polylog}\left(\sqrt{\sum_{j=1}^n C_j^2}, \frac{1}{1-2n\lambda}\right)\right) = \O\left(\mathrm{polylog}(S_{\max},n)\right).\]

	According to \cref{lem:TAconcurrent,lem:TBconcurrent,lem:TCconcurrent}, we can construct $\Cal R_\A$, $\Cal C$ and $\Cal R_\H$ with $\O(1)$ calls to $\O_x$, a number of calls to $\O_{\A}$, $\O_{\S}$ that satisfies
	\[\O\left(\max_{j \in [n]} \frac{\overline{T}_j}{\overline{S}_j}\right) = \O\left(B\right),\]
	a number of auxiliary gates that satisfies
	\[\O\left(\max_{j \in [n]} \frac{\overline{T}_j}{\overline{S}_j} \cdot \mathrm{polylog}(\overline{T}_j)\right) = \O\left(B \cdot \mathrm{polylog}(T_{\max})\right),\]
	and a number of auxiliary qubits that is polylogarithmic in $T_{\max}$. If we ensure that the answer register is located on the same qubit for all the algorithms $\A^{(j)}$'s, we can implement $\mathcal{R}_0$ with $\O(\mathrm{polylog}(T_{\max}))$ gates. This implies that the total number of calls to $\O_{\A}$ and $\O_{\S}$ is
	\[\O\left(\sqrt{\sum_{j=1}^n S_j^2} \cdot B\right) = \O\left(\sqrt{\sum_{j=1}^n S_j^2} \cdot \sqrt{\frac{\sum_{j=1}^n T_j^2}{\sum_{j=1}^n S_j^2}}\right) = \O\left(\sqrt{\sum_{j=1}^n T_j^2}\right),\]
	and the total number of auxiliary gates is
	\[\O\left(\sqrt{\sum_{j=1}^n S_j^2} \cdot B \cdot \mathrm{polylog}(T_{\max},n)\right) = \O\left(\sqrt{\sum_{j=1}^n T_j^2} \cdot \mathrm{polylog}(T_{\max},n)\right).\]
	This completes the proof of the claimed complexities.

	It remains to check that the success probability of our algorithm compiled from $P$ is sufficiently high. We have $\O(\sqrt{\varepsilon_j})$-precise implementations of $\Cal R_{\Ker(A^{(j)})}$ and $\Cal R_{\ket{0}}^{(j)}$ w.r.t.\ operator norm. Thus, our resulting implementations of $\mathcal{R}_{\Ker(A)}$ and $\mathcal{R}_0$ are accurate in the operator norm up to error
	\[\max_{j \in [n]} 4\sqrt{2\varepsilon_j} = o\left(\frac{1}{\sqrt{\sum_{j=1}^n S_j^2}}\right).\]

	Similarly, the subroutines $\Cal C_{\ket{w_0^{(j)}}}$ only approximately stay within $\H_{x^{(j)}}$. Thus,
	\[\sup_{\substack{\ket{h} \in \H_x \\ \norm{\ket{h}} = 1}} \norm{\Pi_{\H_x^{\perp}}\Cal C_{\ket{w_0}}\ket{h}} \leq \max_{j \in [n]} 2\sqrt{2\varepsilon_j} = o\left(\frac{1}{\sqrt{\sum_{j=1}^n S_j^2}}\right).\]
	As we call these two subroutines a total of $\sqrt{\sum_{j=1}^n S_j^2}$ times, these errors influence the final success probability at most by $o(1)$, using a similar argument as in the proof of \cref{thm:main-time}. Thus, our implementation of the algorithm compiled from $P$ succeeds with bounded error.

	Finally, we remove the restriction that we imposed on the algorithms $\{\A^{(j)}\}_{j=1}^n$ at the beginning of this proof. We can always reduce the error probability of our algorithms to $o(1/\sum_{j=1}^n S_j^2)$ using standard techniques. This conversion incurs a multiplicative factor of $\O(\log(\sum_{j=1}^n S_j^2))$ in the query and time complexities, and in the worst case an additive term of $k_{\max}^{o(1)}$ in the number of auxiliary qubits. Accounting for them in the relevant complexities completes the proof.
\end{proof}

\section{Discussion and outlook}

In this paper, we reached two main results. First, we prove in \cref{sec:main_thm} that every quantum query algorithm can be converted into a span program and back into a quantum algorithm while keeping the query and time complexity unaffected up to polylogarithmic factors. This implies that span programs fully capture both query and time complexity up to polylogarithmic factors, which strengthens the motivation for considering span programs as an important formalism from which to derive quantum algorithms.

There remain some interesting follow-up questions in this direction. First, it feels like the error analysis of our algorithm compiled from the span program in \cref{sec:main_thm} is not yet optimal. At this point, we implement the span program unitary up to precision $\O(\sqrt{\varepsilon})$. As we call this unitary a total of $\O(S)$ times, we have to require that $\varepsilon = o(1/S^2)$ to make sure the total error in the operator norm throughout the execution of the algorithm scales as $o(1)$. Within the construction of the span program unitary, the subroutine $\Cal C_{\ket{w_0}}$ that constructs the minimal positive witness and the subroutine $\Cal R_{\ket{0}}$ that reflects around the all-zeros state both have an error that scales with $\O(\sqrt{\varepsilon})$. To attain an improvement, one would have to deal with both these error dependencies.

As a first step towards improvement, with a more clever choice of the implementing subspace it might be possible to move the error in the implementation of $\Cal C_{\ket{w_0}}$ to the implementation of $\Cal R_{\Ker(A)}$, and subsequently analyze the error of the implementation of $\Cal R_{\Ker(A)}$ more favorably. A candidate for the more clever choice of the implementing subspace would be a space that gradually transitions from the ideal initial state to the ideal final state with equally sized steps at all query time steps. The hope in this direction would be to show that the reflection through the kernel of $A$ is actually implemented with error $\O(\sqrt{\varepsilon}/S)$ in operator norm. An improvement in the error of $\Cal R_{\ket{0}}$ seems harder to achieve with this method.

Alternatively, one could dig deeper into the proofs of the success probabilities of the algorithms compiled from positively $\lambda$-approximating span programs to see if in our particular setting the error propagates more favorably than we argue at this point. The hope in this direction would be to prove that even though we implement the span program unitary up to error $\O(\sqrt{\varepsilon})$ in operator norm, the actual state vector does not drift away from the ideal state as quickly as $\O(\sqrt{\varepsilon})$ in norm error per step.

If both above ideas work, they would render redundant the boosting of the success probability at the end of the proof of \cref{thm:main-time}. The effect would be that the $\log(S)$-factors in the query complexities in the statement of that theorem could be removed, alongside with the term $k^{o(1)}$, making the theorem statement a lot cleaner and more elegant. Moreover, it would also remove the same factors from \cref{thm:variable-time}, which would allow us to recover the query complexity of Ambainis's result without multiplicative logarithmic overhead. Hence, we think that these ideas are very much worthwhile investigating.

Our second result, in \cref{sec:variable_time_search}, is an improvement on Ambainis's variable time search result -- we can obtain a Grover-like speed-up in both query and time complexity simultaneously, where the query complexity is measured in the number of calls to $\O_x$ providing access to the input $x$ and the time complexity is measured in the number of calls to $O_\A$ and $\O_\S$ providing access to the descriptions of the algorithms. Our construction goes via a composition of span programs. Even though the analysis of the time complexity of the algorithm compiled from this composed span program is quite involved, and in some way not yet complete, the actual composition is rather simple. This exemplifies the power of the span program framework.

This section leaves several open ends for further research as well. First, we do not rederive all of the results that Ambainis obtains in his work. For instance, we do not consider the case where the query and time complexities of the original algorithms are not known in advance, so it would be interesting to investigate whether we could match Ambainis's result in this setting as well. This would probably require somewhat modifying the input model that we describe in \cref{sec:model}.

Similarly, we handle the decision version of the search problem whereas Ambainis handles the full search version. It would be interesting to see if one can recover the full search algorithm as well. One possible direction would be to investigate whether one could use span programs with non-binary outputs for that, as described for instance in \cite{BT20}.

The most interesting direction of further research that we foresee, though, is whether the relative ease with which span programs can be composed can be exploited to obtain more composition results. The variable time search result composes a set of arbitrary functions with the OR function and obtains a Grover-like speed-up in the query and time complexity of the resulting algorithm. A natural next step would be to investigate if similar types of speed-ups can be obtained when one composes some arbitrary functions with threshold functions.

\section{Acknowledgements}

Stacey Jeffery thanks Tsuyoshi Ito for illuminating discussions on the topic of span programs and time complexity.

\bibliographystyle{alphaurl}
\bibliography{sptimebib}

\appendix

\section{Proof of \cref{thm:scaling}: Span program rescaling}\label{app:scaling}

\thmscaling*

\begin{proof}
	By \cite[Lemma~11]{IJ15}, $\ket{w_0^\beta}=\frac{\beta}{\beta^2+N}\ket{w_0}+\frac{N}{\beta^2+N}\ket{\hat0}+\frac{\beta}{\sqrt{\beta^2+N}}\ket{\hat1}$ and $\|\ket{w_0^\beta}\|=1$.

	By \cite[Lemma~12]{IJ15}, for any $x\in f^{-1}(1)$, $w_+(x,P^\beta)=\frac{1}{\beta^2}w_+(x,P)+\frac{\beta^2}{N+\beta^2}$, so
	$$W_+(P^\beta)\leq \frac{1}{\beta^2}W_+(P)+1\leq 2,$$
	under our assumption $\beta\geq \sqrt{W_+(P)}$.

	Next, for any $x\in f^{-1}(0)$, let $\ket{\widetilde\omega}\in\V$ be an approximate negative witness for $x$. Define $\ket{\widetilde\omega'}\in\V^\beta$ as
	$$\ket{\widetilde\omega'} = \frac{\beta^2+N}{\beta^2\lambda+\beta^2+N}\ket{\widetilde\omega} + \frac{\beta^2\lambda}{\beta^2\lambda+\beta^2+N}\ket{\hat 1}.$$
	Since $\ket{\widetilde\omega}$ is a negative witness, $\braket{\widetilde\omega}{\tau} = 1$ and so $\braket{\widetilde\omega'}{\tau^\beta}=1$.
	By substituting $A^\beta$ from \cref{eq:VAt},
	\begin{align*}
		\norm{\bra{\widetilde\omega'}A^{\beta}\Pi_{\H^{\beta}(x)}}^2 &= \norm{\frac{\beta^2+N}{\beta^2\lambda+\beta^2+N}\beta\bra{\widetilde\omega} A\Pi_{\H(x)}}^2+\norm{\frac{\beta^2\lambda}{\beta^2\lambda+\beta^2+N}\frac{\sqrt{\beta^2+N}}{\beta}\bra{\hat 1}}^2\\
		&= \frac{\beta^2(\beta^2+N)^2}{(\beta^2\lambda+\beta^2+N)^2}\norm{\bra{\tilde\omega} A\Pi_{\H(x)}}^2+\frac{\beta^2\lambda^2(\beta^2+N)}{(\beta^2\lambda+\beta^2+N)^2}\\
		&\leq \frac{\beta^2(\beta^2+N)((\beta^2+N)\lambda/\beta^2+\lambda^2)}{(\beta^2\lambda+\beta^2+N)^2}
		= \lambda \frac{(\beta^2+N)(\beta^2+N+\beta^2\lambda)}{(\beta^2\lambda+\beta^2+N)^2}\\
		&=\lambda\frac{\beta^2+N}{\beta^2\lambda+\beta^2+N}\leq \lambda \leq \frac{2\lambda}{W_+(P^\beta)}.
	\end{align*}
	Thus $P^{\beta}$ positively $2\lambda$-approximates $f$.

	To compute the negative complexity, note that $\braket{\widetilde\omega'}{\tau} \leq \braket{\widetilde\omega}{\tau} = 1$ so
	\begin{align*}
		\norm{\bra{\widetilde\omega'} A^{\beta}}^2 &= \norm{\frac{\beta^2+N}{\beta^2\lambda+\beta^2+N}\beta\bra{\widetilde\omega} A}^2+\norm{\braket{\widetilde\omega'}{\tau}\bra{\hat 0}}^2+\norm{\frac{\beta^2\lambda}{\beta^2\lambda+\beta^2+N}\frac{\sqrt{\beta^2+N}}{\beta}\bra{\hat 1}}^2\\
		&\leq \frac{(\beta^2+N)^2\beta^2}{(\beta^2\lambda+\beta^2+N)^2}\widetilde{W}_-(P)+1+\frac{\beta^2\lambda^2(\beta^2+N)}{(\beta^2\lambda+\beta^2+N)^2}
		\leq \beta^2\widetilde{W}_-(P)+1+1,
	\end{align*}
	which completes the proof.
\end{proof}

\section{Witness anatomy of the span program of an algorithm}\label{app:lemwzero}

\lemkernel*

\begin{proof}
	We can simply check by direct calculation that all vectors in the image of the linear maps $\Phi_{\ell}$ are elements in the kernel of $A$. Thus, it remains to show that any vector in the kernel of $A$ can be written as a linear combination of vectors in the image of the $\Phi_{\ell}$'s. To that end, let $\ket{\Psi} \in \Ker(A) \subseteq \H$. We first of all split this state in several disjointly supported parts, i.e.,
	\begin{align*}
		\ket{\Psi} = \;&\frac1M\sum_{t=0}^{q_1-2} \ket{t}\ket{0}\ket{\psi_{1,t}} + \ket{q_1-1}\frac{\ket{+}}{\sqrt{2}}\ket{\psi_{1,q_1-1}} \\
		&+ \sum_{\ell=2}^S \left(\ket{q_{\ell-1}-1}\frac{\ket{-}}{\sqrt{2}}\ket{\psi_{\ell,q_{\ell-1}-1}} + \frac1M\sum_{t=q_{\ell-1}}^{q_{\ell}-2} \ket{t}\ket{0}\ket{\psi_{\ell,t}} + \ket{q_{\ell}-1}\frac{\ket{+}}{\sqrt{2}}\ket{\psi_{\ell,q_{\ell}-1}}\right) \\
		&+ \ket{q_S-1}\frac{\ket{-}}{\sqrt{2}}\ket{\psi_{S+1,q_S-1}} + \frac1M \sum_{t=q_S}^{T-1} \ket{t}\ket{0}\ket{\psi_{S+1,t}} + \frac1a\ket{T}\ket{0}\ket{\psi_{S+1,T}},
	\end{align*}
	where all the amplitudes are absorbed in the unnormalized $\ket{\psi_{\ell,t}}$-vectors. Now, we apply $A$ to this vector to obtain
	\begin{align*}
		A\ket{\Psi} = \;& \sum_{t=0}^{q_1-2} \left(\ket{t}\ket{\psi_{1,t}} - \ket{t+1}U_{t+1}\ket{\psi_{1,t}}\right) + \ket{q_1-1}\ket{\psi_{1,q_1-1}} \\
		&+ \sum_{\ell=2}^S \left(-\ket{q_{\ell-1}}\ket{\psi_{\ell,q_{\ell-1}-1}} + \sum_{t=q_{\ell-1}}^{q_{\ell}-2} \left(\ket{t}\ket{\psi_{\ell,t}} - \ket{t+1}U_{t+1}\ket{\psi_{\ell,t}}\right) + \ket{q_{\ell}-1}\ket{\psi_{\ell,q_{\ell}-1}}\right) \\
		&- \ket{q_S}\ket{\psi_{S+1,q_S-1}} + \sum_{t=q_S}^{T-1} \left(\ket{t}\ket{\psi_{S+1,t}} - \ket{t+1}U_{t+1}\ket{\psi_{S+1,t}}\right) + \ket{T}\ket{\psi_{S+1,T}}\\
		=\; &\ket{0}\ket{\psi_{1,0}} + \sum_{t=1}^{q_1-1}\ket{t}(\ket{\psi_{1,t}} - U_t\ket{\psi_{1,t-1}})\\
		&+ \sum_{\ell=2}^S\left(\ket{q_{\ell-1}}(\ket{\psi_\ell,q_{\ell-1}}-\ket{\psi_{\ell,q_{\ell-1}-1}})+\sum_{t=q_{\ell-1}+1}^{q_\ell-1}\ket{t}(\ket{\psi_{\ell,t}}-U_t\ket{\psi_{\ell,t-1}}) \right)\\
		& + \ket{q_S}(\ket{\psi_{S+1,q_S}}-\ket{\psi_{S+1,q_S-1}})+\sum_{t=q_S+1}^{T}\ket{t}(\ket{\psi_{S+1,t}}-U_t\ket{\psi_{S+1,t-1}}).
	\end{align*}
	As $\ket{\Psi} \in \Ker(A)$, the above expression has to equal $0$. We learn by inspection that this happens if and only if the following conditions are satisfied:
\begin{align*}
\ket{\psi_{1,0}}&=0\\
\forall t\in\{1,\dots,q_1-1\},\quad \ket{\psi_{1,t}} &= U_t\ket{\psi_{1,t-1}}\\
\forall \ell\in\{2,\dots,S+1\},\quad \ket{\psi_{\ell,q_{\ell-1}}}&=\ket{\psi_{\ell,q_{\ell-1}-1}}\\
\forall \ell\in \{1,\dots,S\},t\in \{q_{\ell-1}+1,\dots,q_{\ell}-1\},\quad \ket{\psi_{\ell,t}}&=U_t\ket{\psi_{\ell,t-1}}\\
\forall t\in\{q_S+1,\dots,T\},\quad \ket{\psi_{S+1,t}} &= U_t\ket{\psi_{S+1,t-1}}.
\end{align*}
Using the abbreviation $\ket{\psi_\ell}=\ket{\psi_{\ell,q_{\ell-1}-1}}$ for $\ell\in\{2,\dots,S+1\}$, these conditions simplify to:
\begin{align*}
\forall t\in\{0,\dots,q_1-1\},\quad \ket{\psi_{1,t}} &= 0\\
\forall \ell\in\{2,\dots,S+1\},\quad \ket{\psi_{\ell,q_{\ell-1}}}&=\ket{\psi_{\ell,q_{\ell-1}-1}}=\ket{\psi_\ell}\\
\forall \ell\in \{1,\dots,S\},t\in \{q_{\ell-1}+1,\dots,q_{\ell}-1\},\quad \ket{\psi_{\ell,t}}&=U_t\dots U_{q_{\ell-1}+1}\ket{\psi_{\ell,q_{\ell-1}}}=U_t\dots U_{q_{\ell-1}+1}\ket{\psi_{\ell}}\\
\forall t\in\{q_S+1,\dots,T\},\quad \ket{\psi_{S+1,t}} &= U_t\dots U_{q_S+1}\ket{\psi_{S+1,q_S}}=U_t\dots U_{q_S+1}\ket{\psi_{S+1}}.
\end{align*}
Using these constraints, we can rewrite $\ket{\Psi}$ as
	\begin{align*}
		\ket{\Psi} &= \sum_{\ell=2}^S \left(\ket{q_{\ell-1}-1}\frac{\ket{-}}{\sqrt{2}}\ket{\psi_{\ell}} + \frac1M\sum_{t=q_{\ell-1}}^{q_{\ell}-2} \ket{t}\ket{0}U_t \cdots U_{q_{\ell-1}+1}\ket{\psi_{\ell}} + \ket{q_{\ell}-1}\frac{\ket{+}}{\sqrt{2}}U_{q_{\ell}-1} \cdots U_{q_{\ell-1}+1}\ket{\psi_{\ell}}\right) \\
		&\;\;\;\; + \ket{q_S-1}\frac{\ket{-}}{\sqrt{2}}\ket{\psi_{S+1}} + \frac1M \sum_{t=q_S}^{T-1} \ket{t}\ket{0}U_t \cdots U_{q_S+1}\ket{\psi_{S+1}} + \frac1a\ket{T}\ket{0}U_T \cdots U_{q_S+1}\ket{\psi_{S+1}} \\
		&= \sum_{\ell=2}^{S+1} \Phi_{\ell}(\ket{\psi_{\ell}}),
	\end{align*}
	completing the proof.
\end{proof}

\lemwzero*

\begin{proof}
	We first prove that $\ket{w_0}$ is orthogonal to all vectors in the kernel of $A$. By \cref{lem:kernel}, it suffices to take $\ket{\psi} \in \C^{[n] \times \mathcal{W}}$ arbitrarily and check that for all $\ell \in \{2, \dots, S+1\}$, $\bra{\psi}\Phi_{\ell}^{\dagger}\ket{w_0} = 0$. Observe that $\ket{w_0}$ does not have support in states with time $t\in \{q_1,\dots q_S-2\}$, hence, for all $\ell \in \{3,\dots,S-1\}$, we easily obtain that the vectors $\ket{w_0}$ and $\Phi_{\ell}\ket{\psi}$ are orthogonal. For $\ell = 2$, we find that the supports only overlap at $t = q_1 - 1$ with the term $\ket{q_1-1}\frac{\ket{-}}{\sqrt{2}}\ket{\psi}$ of $\Phi_1\ket{\psi}$, so
	\[\bra{\psi}\Phi_1^{\dagger}\ket{w_0} = \frac{1}{4} \bra{\psi} U_{q_1-1} \cdots U_1 \ket{\Psi_0} - \frac{1}{4} \bra{\psi} U_{q_1-1} \cdots U_1\ket{\Psi_0} = 0.\]
	A similar computation shows that $\ket{w_0}$ and $\Phi_{S}\ket{\psi}$ are orthogonal, as their supports only overlap at $t = q_S-1$ with the term $\ket{q_S-1}\frac{\ket{+}}{\sqrt{2}}U_{q_S-1}\dots U_{q_{S-1}+1}\ket{\psi}$ of $\Phi_S\ket{\psi}$. Finally,
	\begin{align*}
		\bra{\psi}\Phi_{S+1}^{\dagger}\ket{w_0} &= \frac{1}{Ca^2+1} \left[\frac14 \bra{\psi} U_{q_S+1}^{\dagger} \cdots U_T^{\dagger} \ket{\Psi_T} + \frac14 \bra{\psi} U_{q_S+1}^{\dagger} \cdots U_T^{\dagger} \ket{\Psi_T} + \frac{1}{M^2} \sum_{t = q_S}^{T-1} \bra{\psi} U_{q_S+1}^{\dagger} \cdots U_T^{\dagger} \ket{\Psi_T}\right] \\
		&\;\;\;\;\; - \frac{C}{Ca^2+1} \bra{\psi} U_{q_S+1}^{\dagger} \cdots U_T^{\dagger} \ket{\Psi_T} \\
		&= \left(\frac{1}{Ca^2+1} \cdot \left[\frac12 + \frac{T-q_S}{M^2}\right] - \frac{C}{Ca^2+1}\right) \bra{\psi} U_{q_S+1}^{\dagger} \cdots U_T^{\dagger} \ket{\Psi_T} \\
		&= \left[\frac{C}{Ca^2+1} - \frac{C}{Ca^2+1}\right]\bra{\psi} U_{q_S+1}^{\dagger} \cdots U_T^{\dagger} \ket{\Psi_T} = 0.
	\end{align*}
	Thus $\ket{w_0}$ is orthogonal to the kernel of $A$. It is also mapped to $\ket{\tau}$ by $A$, as
	\begin{align*}
		A\ket{w_0} &= \sum_{t=0}^{q_1-2} \left(\ket{t}U_t \cdots U_1\ket{\Psi_0} - \ket{t+1}U_{t+1} \cdots U_1\ket{\Psi_0}\right) + \left(\frac12 + \frac12\right)\ket{q_1-1}U_{q_1-1} \cdots U_1\ket{\Psi_0} \\
		&\;\;\;\;\; - \left(\frac12 - \frac12\right)\ket{q_1}U_{q_1-1} \cdots U_1\ket{\Psi_0} + \frac{1}{Ca^2+1} \left[ \left(\frac12-\frac12\right)\ket{q_S-1}U_{q_S+1}^{\dagger} \cdots U_T^{\dagger}\ket{\Psi_T} \right. \\
		&\;\;\;\;\;\;\;\;\;\; \left. - \left(\frac12 + \frac12\right)\ket{q_S}U_{q_S+1}^{\dagger} \cdots U_T^{\dagger}\ket{\Psi_T} + \sum_{t=q_S}^{T-1} \left(\ket{t}U_{t+1}^{\dagger} \cdots U_T^{\dagger}\ket{\Psi_T} - \ket{t+1} U_{t+2}^{\dagger} \cdots U_T^{\dagger}\ket{\Psi_T}\right) \right] \\
		&\;\;\;\;\; - \frac{Ca^2}{Ca^2+1} \ket{T}\ket{\Psi_T} = \ket{0}\ket{\Psi_0} - \frac{1}{Ca^2+1} \ket{T}\ket{\Psi_T} - \frac{Ca^2}{Ca^2+1} \ket{T}\ket{\Psi_T} = \ket{\tau}.
	\end{align*}
	This proves that $\ket{w_0}$ is the minimal witness of $A$. Calculating its squared norm, we obtain
	\begin{align*}
		N &= \norm{\ket{w_0}}^2 = \sum_{t=0}^{q_1-2}\frac{1}{M^2} + \frac14 + \frac14 + \frac{1}{(Ca^2+1)^2}\left[\frac14 + \frac14 + \sum_{t=q_S}^{T-1} \frac{1}{M^2}\right] + \frac{C^2a^2}{(Ca^2+1)^2} \\
		&= \frac{q_1-1}{M^2} + \frac12 + \frac{C}{(Ca^2+1)^2} + \frac{C^2a^2}{(Ca^2+1)^2} = \frac{q_1-1}{M^2} + \frac12 + \frac{C}{Ca^2+1},
	\end{align*}
	which completes the proof.
\end{proof}

\section{Time complexity analysis}
\label{app:time-complexity}

\subsection{Splitting maps}

First, we implement a subroutine that we will use throughout this section.

\begin{definition}\label{def:splitting-map}
	Let $\A$ be a clean quantum algorithm with time complexity $T$ and query complexity $S$ and let $P_\A = (\H,\V,A,\ket{\tau})$ be its span program. We define the map $\Cal S_{t,\alpha}$ for all $t \in [T-1]_0$ and $\alpha \in [0,1]$ as the map acting on $\H$ whose non-trivial action is described, for $\ket{\psi}\in \mathbb{C}^{[n]\times\W}$, as
	\begin{equation}
	\label{eq:splitting_map}
	\begin{array}{rl}
	\text{if } t + 1 \in \mathcal{S}, \qquad & \Cal S_{t,\alpha} : \left\{\begin{array}{rcl}
	\ket{t}\ket{-}\ket{\psi} &\mapsto& \alpha\ket{t}\ket{-}\ket{\psi} + \sqrt{1-\alpha^2}\ket{t+1}\ket{0}\ket{\psi} \\
	\ket{t+1}\ket{0}\ket{\psi} &\mapsto& -\sqrt{1-\alpha^2}\ket{t}\ket{-}\ket{\psi} + \alpha\ket{t+1}\ket{0}\ket{\psi}
	\end{array}\right. \\
	\text{if } t + 2 \in \mathcal{S}, \qquad & \Cal S_{t,\alpha} : \left\{\begin{array}{rcl}
	\ket{t}\ket{0}\ket{\psi} &\mapsto& \alpha\ket{t}\ket{0}\ket{\psi} + \sqrt{1-\alpha^2}\ket{t+1}\ket{+}U_{t+1}\ket{\psi} \\
	\ket{t+1}\ket{+}\ket{\psi} &\mapsto& -\sqrt{1-\alpha^2}\ket{t}\ket{0}U_{t+1}^\dagger\ket{\psi} + \alpha\ket{t+1}\ket{+}\ket{\psi}
	\end{array}\right. \\
	\text{otherwise}, \qquad & \Cal S_{t,\alpha} : \left\{\begin{array}{rcl}
	\ket{t}\ket{0}\ket{\psi} &\mapsto& \alpha\ket{t}\ket{0}\ket{\psi} + \sqrt{1-\alpha^2}\ket{t+1}\ket{0}U_{t+1}\ket{\psi} \\
	\ket{t+1}\ket{0}\ket{\psi} &\mapsto& -\sqrt{1-\alpha^2}\ket{t}\ket{0}U_{t+1}^\dagger\ket{\psi} + \alpha\ket{t+1}\ket{0}\ket{\psi}.
	\end{array}\right.
	\end{array}
	\end{equation}
	In all three cases we have $\Cal S_{t,\alpha}:\ket{t'}\ket{\phi}\mapsto\ket{t'}\ket{\phi}$ for $t'\notin\{t,t+1\}$. We refer to these maps as \emph{splitting maps}.
\end{definition}

Note that for all choices of $t$ and $\alpha$, $\Cal S_{t,\alpha}$ leaves both $\overline{\H}_x$ and $\widetilde{\H}_x$ invariant. We leave this to the reader to check. In the lemma below, we elaborate on how we can implement this splitting map efficiently.

\begin{lemma}
	\label{lem:splitting_map}
	Let $\A$ be a clean quantum algorithm with time complexity $T$ and query complexity $S$, and let $P_\A = (\H,\V,A,\ket{\tau})$ be its span program. Let $t \in [T-1]_0$ and $\alpha \in [0,1]$. We can implement $\Cal S_{t,a}$ with two controlled calls to (the inverse of) $\O_\A$ and $\O(\mathrm{polylog}(T))$ additional gates.

	Furthermore, if we have a binary description of $t$ and $\alpha$ in auxiliary registers, where the description of $\alpha$ is $\delta$-precise, we can implement $\Cal S_{t,\alpha}$ up to error $\delta > 0$ with two controlled calls to (the inverse of) $\O_\A$ and $\O(\mathrm{polylog}(T,1/\delta))$ additional gates.
\end{lemma}

\begin{proof}
	First of all, we check which of the three cases that are listed in \cref{eq:splitting_map} applies. This we can do with one call to $\O_\S$ and polylogarithmically many extra gates in $T$. Each of these three cases we treat separately and consecutively. We only give the explicit description of the last case here, as the others come down to the same circuit with some minor adjustments.

	We implement the bottom mapping in \cref{eq:splitting_map} in three steps.
	\begin{enumerate}
		\item First, controlled on the first register being in time $t + 1$, we call the inverse of $\O_\A$. This will map $\ket{t}\ket{0}\ket{\psi}$ to itself, and it will map $\ket{t+1}\ket{0}\ket{\psi}$ to $\ket{t+1}\ket{0}U_{t+1}^\dagger\ket{\psi}$. This takes $1$ call to $\O_\A$, and $\O(\mathrm{polylog}(T))$ other gates.
		\item Next, we apply the following mapping to the first register:
		\[\ket{t} \mapsto \alpha\ket{t} + \sqrt{1-\alpha^2}\ket{t+1} \qquad \text{and} \qquad \ket{t+1} \mapsto -\sqrt{1-\alpha^2}\ket{t} + \alpha\ket{t+1}\]
		As this is a two-level rotation, we can implement it with $\O(\mathrm{polylog}(T))$ single qubit gates and $\mathrm{CNOT}$s.
		\item Finally, we apply $\O_\A$, controlled on the first register being in time $t+1$. This will add the $U_{t+1}$'s to the description of the state vector wherever this is required in the statement of the lemma. Just as in step 1, this takes $1$ call to $\O_\A$ and $\O(\mathrm{polylog}(T))$ additional gates.
	\end{enumerate}
	One can easily check that this implements the third mapping in \cref{eq:splitting_map}.

	Furthermore, if we have a binary description of $t$ and $\alpha$ stored in an extra register, we can implement the desired mapping in a similar number of gates. While we cannot hardcode $t$ in steps 1 and 3, we can control on its value. Similarly, in step 2, as $\alpha$ is not hardcoded, we have to substitute the rotation with $\O(\mathrm{polylog}(1/\delta))$ rotations controlled on the qubits storing $\alpha$. All of the necessary computations are efficiently implementable classically, and hence only add (additive) polylogarithmic overhead in the error parameter to the time complexity.
\end{proof}

\subsection{Proof of \cref{lem:TA}: Implementation of reflection around $\ker A$}
\label{app:TA}

In this section, we prove \cref{lem:TA}, i.e., we provide an implementation of the routine that reflects around the kernel of the span program operator $A$ as defined in \cref{eq:A}. In addition, we also elaborate on how one would obtain concurrent access to these routines when considering multiple such span program operators, because we need that in the proof of \cref{thm:variable-time}. The result is summarized in the following lemma, of which \cref{lem:TA} is a special case.

\begin{lemma}
	\label{lem:TAconcurrent}
	\lemTA*

	Similarly, let $\{\A^{(j)}\}_{j=1}^n$ be a set of clean quantum query algorithms. For all $j \in [n]$, let $S_j$ and $T_j$ be the query and time complexity of $\A^{(j)}$, respectively. Let $P^{(j)} = (\H^{(j)}, \V^{(j)}, A^{(j)}, \ket{\tau^{(j)}})$ be the span program of $\A^{(j)}$. We can provide concurrent access to $\{2\Pi_{\Ker(A^{(j)})} - \I\}_{j=1}^n$ up to precision $\delta > 0$ with $\O(\max_{j \in [n]} T_j/S_j)$ calls to $\O_\A$ and $\O_\S$, $\O(\mathrm{polylog}(T_{\max}))$ auxiliary qubits and a number of extra gates that satisfies
	\[\O\left(\max_{j \in [n]} \frac{T_j}{S_j} \mathrm{polylog}\left(T_{\max},\frac{1}{\delta}\right)\right).\]
\end{lemma}

The main idea of the proof is to use the characterization of the kernel of $A$ given in Lemma \ref{lem:w0}, and to map this space isometrically to another space around which we can reflect more easily. The formal proof of \cref{lem:TA} is given at the end of this section.

First of all, we define what we call the left and right block oracles, $\O_L$ and $\O_R$. Intuitively, we can break the time indices into query blocks, $q_{\ell-1},\dots,q_{\ell}$, beginning with the index of a query, and going up to the index of the next query. Two consecutive blocks overlap in a single index, which is always the index of a query. When $t$ is the index of a non-query, it belongs to a unique query block, and so the left-endpoint of its query block, $q_{\ell-1}$, is uniquely defined, as is the right-endpoint, $q_{\ell}$. Thus, we can define operations $\O_L$ and $\O_R$ that, for any such $t$, compute these values, or, rather, for technical reasons, given $\ket{t}$ such that $q_{\ell-1}<t+1<q_{\ell}$, $\O_L$ and $\O_R$ return $q_{\ell-1}-1$ and $q_{\ell}-1$ respectively.

When $t+1$ is the index of a query, there is ambiguity, because it is part of two blocks -- it is the left-endpoint of one, and the right-endpoint of another. We use an auxiliary qubit to resolve this ambiguity: for a state $\ket{t}\ket{+}$, we interpret $t+1=q_\ell$ as the right-endpoint of a block, so $\O_L$ and $\O_R$ return $q_{\ell-1}-1$ and $q_{\ell}-1$ respectively; and for a state $\ket{t}\ket{-}$, we interpret $t+1=q_\ell$ as the left-endpoint of a block, so $\O_L$ and $\O_R$ return $q_{\ell}-1$ and $q_{\ell+1}-1$ respectively. In other words, blocks start with a query and finish immediately before the next one. The precise actions of $\O_L$ and $\O_R$ are defined as follows.

\begin{definition}
	Let $\A$ be a clean quantum algorithm and let $P_{\A} = (\H,\V,A,\ket{\tau})$ be its span program. We define the left and right block oracles as unitaries on $\H \otimes \C^{\{-1,\dots,T\}}$, acting as
	\[\O_L : \left\{\begin{array}{rcll}
		\ket{t}\ket{0}\ket{0} &\mapsto& \ket{t}\ket{0}\ket{q_{\ell-1}-1}, & \text{if $t+1\not\in\S$ and } q_{\ell-1} - 1 < t < q_{\ell}-1, \\
		\ket{t}\ket{+}\ket{0} &\mapsto& \ket{t}\ket{+}\ket{q_{\ell-1}-1}, & \text{if $t+1\in \S$ and } t=q_{\ell} - 1, \\
		\ket{t}\ket{-}\ket{0} &\mapsto& \ket{t}\ket{-}\ket{q_{\ell-1}-1}, & \text{if $t+1\in\S$ and } t=q_{\ell-1} - 1,
	\end{array}\right.\]
	and
	\[\O_R : \left\{\begin{array}{rcll}
		\ket{t}\ket{0}\ket{0} &\mapsto& \ket{t}\ket{0}\ket{q_{\ell}-1}, & \text{if $t+1\not\in\S$ and } q_{\ell-1} - 1 < t < q_{\ell}-1, \\
		\ket{t}\ket{+}\ket{0} &\mapsto& \ket{t}\ket{+}\ket{q_{\ell}-1}, & \text{if $t+1\in\S$ and } t=q_{\ell} - 1, \\
		\ket{t}\ket{-}\ket{0} &\mapsto& \ket{t}\ket{-}\ket{q_{\ell}-1}, & \text{if $t+1\in\S$ and } t=q_{\ell-1} - 1.
	\end{array}\right.\]
\end{definition}

Next, we show how to implement these block oracles efficiently.

\begin{lemma}
	\label{lem:block_oracles}
	Let $\A$ be a clean quantum algorithm with query complexity $S$ and time complexity $T$, and let $P_{\A} = (\H,\V,A,\ket{\tau})$ be its span program. Then we can implement $\O_L$ and $\O_R$ with $\O(T/S)$ queries to $\O_\S$, $\O(\mathrm{polylog}(T))$ ancillary qubits and a number of additional gates that scales as
	\[\O\left(\frac{T}{S}\mathrm{polylog}(T)\right).\]

	Similarly, suppose $\{\A^{(j)}\}_{j=1}^n$ is a set of clean quantum algorithms. Let $j \in [n]$ and let $S_j$ and $T_j$ be the query and time complexity of $\A^{(j)}$, respectively. Similarly, let $\O_L^{(j)}$ and $\O_R^{(j)}$ be the left and right block oracles of $\A^{(j)}$, respectively. We can provide concurrent access to $\{\O_L^{(j)}\}_{j=1}^n$ and $\{\O_R^{(j)}\}_{j=1}^n$ with $\O(\max_{j \in [n]} T_j/S_j)$ queries to $\O_\S$, $\O(\mathrm{polylog}(T_{\max}))$ ancillary qubits and a number of additional gates that scales as
	\[\O\left(\max_{j \in [n]} \frac{T_j}{S_j}\mathrm{polylog}(T_{\max})\right).\]
\end{lemma}

\begin{proof}
	We first focus on the case where we have just one algorithm and leave the case where we have multiple algorithms for the final paragraph. We only show how to implement $\O_L$, as the implementation of $\O_R$ is similar. First of all, we check if $t+1\in\S$ using one call to $\O_\S$ and $\O(\mathrm{polylog}(T))$ other gates and store the result in an auxiliary qubit. If this flag qubit is $\ket{1}$, we apply a Hadamard to the second register. If the second register is now $1$, we copy the time $t$ to the last register. Observe that if the input was $\ket{t}\ket{-}\ket{0}$, where $t+1=q_{\ell-1}$ then we are done (up to reapplying the Hadamard to the second register and uncomputing the flag qubit). The only interesting case now is when the second bit is $0$, so we apply all the following operations controlled on this bit being $0$.

	In this last case, what we would like to do is write in a new register the index of the last query before our timestep $t$. For that purpose we initialize a new counter register having $\lceil \log(3T/S) \rceil + 1$ qubits, in the state $\ket{0}$, and iteratively decrement the time register until we reach a time step that is one less than a query time step, and after that the counter register is incremented. This means that after these iterations, we have the correct query time step stored in the time register, while the counter will contain a function of $t$, $q_{\ell-1}$, and $\lfloor 3T/S \rfloor$. This task can be done by repeating the following operation $\lfloor 3T/S \rfloor$ times. First, we check whether the time register is one less than a query time step and if it is we increment the counter register. This can be done using $2$ queries to $\O_\S$ and a number of extra gates that is polylogarithmic in $T$. After that, we decrement the time register controlled on the counter being in the $\ket{0}$ state. This also takes a number of gates that is polylogarithmic in $T$.  We can now copy the time into the last register, and then uncompute all of these iterations, returning the time register to the state $\ket{t}$ and the counter to the state $\ket{0}$.

	At last, we undo the computations we did in the beginning, i.e., we apply the controlled Hadamard again and reset the flag that indicated whether $t + 1\in\S$ using one more query to $\O_\S$ and $\O(\mathrm{polylog}(T))$ extra gates. We easily check that the total cost of this construction matches the claim in the statement of the lemma.

	Finally, in order to provide concurrent access to $\{\O_L^{(j)}\}_{j=1}^n$, we can simply run the loop in the second paragraph for $\max_{j \in [n]} \lfloor 3T_j/S_j \rfloor$ iterations. The size of the time register now has to be $T_{\max}$ and so the arithmetic operations on this register take a number of gates that is polylogarithmic in $T_{\max}$. This completes the proof.
\end{proof}

Next, we define a mapping that generates the vectors of the kernel of $A$.

\begin{definition}
	\label{def:construction_kernel}
	Let $\A$ be a clean quantum algorithm with time complexity $T$ and query complexity $S$, and let $P_{\A} = (\H,\V,A,\ket{\tau})$ be its span program. Define $\Cal C$ as a unitary on $\H$, which, for all $\ell \in \{2,\dots,S+1\}$ and $\ket{\psi} \in \C^{[n] \times \mathcal{W}}$, acts as
	\[\Cal{C} : \ket{q_{\ell-1}-1}\ket{-}\ket{\psi} \mapsto \frac{\Phi_{\ell}\ket{\psi}}{\norm{\Phi_{\ell}\ket{\psi}}},\]
	and otherwise, $\Cal C$ acts arbitrarily, but leaves $\overline{\H}_x$ and $\widetilde{\H}_x$ invariant.
\end{definition}

\begin{lemma}
	\label{lem:constrkerA}
	Let $\A$ be a clean quantum algorithm with time complexity $T$ and query complexity $S$, and let $P_{\A} = (\H,\V,A,\ket{\tau})$ be its span program. We can implement a mapping $\Cal C$ that satisfies the conditions in \cref{def:construction_kernel}, up to error $\delta > 0$ in operator norm with $\O(T/S)$ queries to $\O_\A$ and $\O_\S$, $\O(\mathrm{polylog}(T))$ ancillary qubits, and with a number of additional gates that scales as
	\[\O\left(\frac{T}{S} \mathrm{polylog}\left(T,\frac{1}{\delta}\right)\right).\]

	Similarly, let $\{\A^{(j)}\}_{j=1}^n$ be a set of clean quantum algorithms. For all $j \in [n]$, let $S_j$ and $T_j$ denote the query and time complexity of $\A^{(j)}$, respectively. Let $\Cal{C}^{(j)}$ be the routine defined in \cref{def:construction_kernel} for $\A^{(j)}$. We can provide concurrent access to $\{\Cal{C}^{(j)}\}_{j=1}^n$ up to error $\delta > 0$ in operator norm with $\O(\max_{j \in [n]}T_j/S_j)$ queries to $\O_\A$ and $\O_\S$, with $\O(\mathrm{polylog}(T_{\max}))$ ancillary qubits, and with a number of additional gates that scales as
	\[\O\left(\max_{j \in [n]} \frac{T_j}{S_j} \mathrm{polylog}\left(T_{\max},\frac{1}{\delta}\right)\right).\]
\end{lemma}

\begin{proof}
	While the behavior of $\Cal C$ is only fully specified on states with $t$ such that $t+1\in\S$ in the first register, more generally, we must ensure that $\Cal C$ leaves $\overline{\H}_x$ and $\widetilde{\H}_x$ invariant, which leads to a more involved construction. We first consider the case where we have just one algorithm $\A$, and leave the case where we have multiple algorithms for the final paragraph of this proof.

	First of all, we call $\O_L$ and $\O_R$ and store the results in some auxiliary registers. According to \cref{lem:block_oracles}, this takes $\O(T/S)$ calls to $\O_\S$ and $\O(T/S \cdot \mathrm{polylog}(T))$ additional gates.

	Next, we distinguish between three cases. First, if the application of $\O_L$ amounts to $\ket{-1}$ in the last register, then necessarily $t$ belongs to the block before the first query, in which case we simply do nothing, i.e.,~we act as the identity. Second, if the result of $\O_R$ is $T$, then we started out in a state in which the time step $t$ was higher than the last query, which requires separate treatment. The third case is when neither of these happened. These cases can be distinguished with $\O(\mathrm{polylog}(T))$ gates, and they can be handled separately and consecutively. We only explain how we handle the final case, as the first one is trivial and the other is similar to the third.

	Hence, we assume that the auxiliary registers are in the states $\ket{q_{\ell-1}-1}$ and $\ket{q_\ell-1}$ for some $\ell \in \{2,\dots,S\}$. Now, we repeat the following procedure $\lfloor 3T/S \rceil$ times. In the $i$th iteration, where $i = 0,1, \dots, \lfloor 3T/S \rfloor$, we perform the following steps.
	\begin{enumerate}
		\setlength\itemsep{-.4em}
		\item First, we initialize an auxiliary qubit and set it to $\ket{1}$ if and only if $q_{\ell-1}+i < q_{\ell}$. This takes time $\O(\mathrm{polylog}(T))$. Steps 2 -- 4 we do controlled on the auxiliary qubit being $\ket{1}$.
		\item From the auxiliary registers, we calculate
		\[\alpha_i = \begin{cases}
			\frac{1}{\sqrt{2} \cdot \sqrt{1 + \frac{q_{\ell}-q_{\ell-1}-1}{M^2}}}, & \text{if } i = 0, \\
			\frac{1}{M \cdot \sqrt{\frac12 + \frac{q_{\ell}-q_{\ell-1}-i}{M^2}}}, & \text{otherwise},
		\end{cases}\]
		and store it in a binary representation in another auxiliary register. This calculation can be done up to precision $\Theta(\delta S/T)$ in a number of gates polylogarithmic in $T$ and $1/\delta$ using standard classical methods.
		\item Next, we apply the splitting map $\Cal S_{q_{\ell-1}+i-1,\alpha_i}$, where $q_{\ell-1}+i-1$ and $\alpha_i$ are stored in separate registers, with $\alpha_i$ up to error $\Theta(\delta S/T)$, which by \cref{lem:splitting_map} incurs $2$ controlled calls to (the inverse of) $\O_\A$ and a number of extra gates that is polylogarithmic in $T$ and $1/\delta$.
		\item We uncompute the parameter $\alpha_i$ from step 2.
		\item We uncompute the check that $q_{\ell - 1}+i < q_{\ell}$.
	\end{enumerate}
	These steps have the effect of applying
	\begin{align*}
		\Cal S_{q_{\ell}-2,\alpha_{q_{\ell}-q_{\ell-1}-1}} \Cal S_{q_{\ell}-3,\alpha_{q_{\ell}-q_{\ell-1}-2}} \dots \Cal S_{q_{\ell-1},\alpha_1} \Cal S_{q_{\ell-1}-1,\alpha_0},
	\end{align*}
	where each factor is implemented up to error $\Theta(S\delta/T)$. As there are at most $\O(T/S)$ factors, the total error is at most $\delta$.

	First of all, recall that $\Cal S_{t,\alpha}$ leaves $\overline{\H}_x$ and $\widetilde{\H}_x$ invariant.
	Moreover, observe that the only values of $t$ for which we execute $\Cal S_{t,\alpha}$ are the values $\{q_{\ell-1}-1, \dots, q_{\ell}-2\}$. By \cref{def:splitting-map}, for any $\alpha$:
	\begin{itemize}
		\setlength\itemsep{-.4em}
		\item $\Cal S_{q_{\ell-1}-1,\alpha}$ only acts non-trivially on $\Span\{\ket{q_{\ell-1}-1}\ket{-}, \ket{q_{\ell-1}}\ket{0}\}\otimes\mathbb{C}^{[n]\times\W}$, which it also leaves invariant (since $q_{\ell-1}\in \S$);
		\item $\Cal S_{q_{\ell}-2,\alpha}$ only acts non-trivially on $\Span\{\ket{q_{\ell}-2}\ket{0},\ket{q_{\ell}-1}\ket{+}\}\otimes\mathbb{C}^{[n]\times\W}$, which it also leaves invariant (since $q_{\ell}\in \S$);
		\item for all $t\in \{q_{\ell-1},\dots,q_{\ell}-3\}$, $\Cal S_{t,\alpha}$ only acts non-trivially on $\Span\{\ket{t}\ket{0},\ket{t+1}\ket{0}\}\otimes\mathbb{C}^{[n]\times\W}$, which it also leaves invariant (since $t+1,t+2\not\in \S$).
	\end{itemize}
	This means that we only act non-trivially on the vectors $\ket{q_{\ell-1}-1}\ket{-}\ket{\psi}$, $\ket{q_{\ell}-1}\ket{+}\ket{\psi}$ and $\ket{t}\ket{0}\ket{\psi}$ where $q_{\ell-1}-1 < t < q_{\ell}-1$ and $\ket{\psi} \in \C^{[n] \times \mathcal{W}}$,
	and we leave the space
	\[\left(\Span\{\ket{q_{\ell-1}-1}\ket{-}\} \oplus \Span\{\ket{t}\ket{0} : q_{\ell-1} - 1 < t < q_{\ell} - 1\} \oplus \Span\{\ket{q_{\ell}-1}\ket{+}\}\right) \otimes \C^{[n] \times \mathcal{W}}\]
	invariant. This implies that the time register always contains a value $t$ such that $t+1$ is within the query block bounded by $q_{\ell-1}$ from the left, and $q_{\ell}$ from the right, where having $\ket{-}$ in the second register, we interpret $q_{\ell-1}$ as a left endpoint, and having $\ket{+}$ in the second register, we interpret $q_{\ell}-1$ as a right endpoint; so we can uncompute the values obtained from $\O_L$ and $\O_R$ (that is, $q_{\ell-1}-1$ and $q_{\ell}-1$) by simply calling their inverses, which is what we do as  the final operation in the circuit.

	We claim that this mapping leaves $\overline{\H}_x$ and $\widetilde{\H}_x$ invariant. This is clear since $\Cal S_{t,\alpha}$ leaves $\overline{\H}_x$ and $\widetilde{\H}_x$ invariant, and all other operations do not matter as they are uncomputed.

	Moreover, we claim that this circuit implements a mapping $\Cal C$ that satisfies the conditions from \cref{def:construction_kernel}. As we are considering the third case, suppose that we start with the state $\ket{q_{\ell}-1}\ket{-}\ket{\psi}$, for some $\ket{\psi} \in \C^{[n] \times \mathcal{W}}$. Now, in the first iteration ($i=0$) we apply $\Cal S_{q_{\ell-1}-1,\alpha_0}$, to arrive at the state
	\[\frac{1}{\sqrt{1 + \frac{q_{\ell}-q_{\ell-1}-1}{M^2}}} \left[\frac{1}{\sqrt{2}}\ket{q_{\ell}-1}\ket{-}\ket{\psi} + \sqrt{\frac12 + \frac{q_{\ell}-q_{\ell-1}-1}{M^2}}\ket{q_{\ell}}\ket{0}\ket{\psi}\right].\]
	We easily check by induction that after the $i$th iteration with $1 \leq i < q_{\ell} - q_{\ell-1}$, we are in the state
	\[\frac{1}{\sqrt{1 + \frac{q_{\ell}-q_{\ell-1}-1}{M^2}}} \left[\begin{array}{l}
		\displaystyle \frac{1}{\sqrt{2}} \ket{q_{\ell-1}-1}\ket{-}\ket{\psi} + \frac{1}{M} \sum_{t=q_{\ell-1}}^{q_{\ell-1}+i-1} \ket{t}\ket{0}U_t \cdots U_{q_{\ell-1}+1}\ket{\psi} \\
		\displaystyle \;\;\;\;\; + \sqrt{\frac12 + \frac{q_{\ell} - q_{\ell-1} - (i+1)}{M^2}}\ket{q_{\ell-1}+i}\ket{0}U_{q_{\ell-1}+i} \cdots U_{q_{\ell-1}+1}\ket{\psi}
	\end{array}\right],\]
	which implies that after the iteration where $i = q_{\ell} - q_{\ell-1} - 1$, we are in the desired state. Whenever $i \geq q_{\ell} - q_{\ell-1}$, we don't do anything due to the condition that is checked in step 1. Hence, this circuit indeed implements a mapping $\Cal C$ that satisfies the conditions outlined in \cref{def:construction_kernel}.

	We observe that there are $\O(T/S)$ iterations, each of which uses $\O(1)$ calls to $\O_\A$ and $\O(\mathrm{polylog}(T,1/\delta))$ extra gates. In addition we do $\O(T/S)$ calls to $\O_\S$ and $\O(T/S \cdot \mathrm{polylog}(T))$ extra gates when we call $\O_L$ and $\O_R$ and their inverses.

	In order to implement concurrent access to $\{\Cal C^{(j)}\}_{j=1}^n$, we can run the loop a total of $\max_{j \in [n]} \lfloor 3T_j/S_j \rfloor$ iterations. The time register now has to be of size $T_{\max} = \max_{j \in [n]} T_j$, and hence the arithmetic operations on this register now take $\O(\mathrm{polylog}(T_{\max}))$ gates. We can now calculate the coefficients $\alpha_i$ with precision $\O(\delta \min_{j \in [n]} S_j/T_j)$. Now all the maps $\Cal C^{(j)}$ are implemented up to precision $\delta$, which implies that the concurrent access is also implemented up to precision $\delta$. This completes the proof.
\end{proof}

It now remains to round up the proof of the main lemma in this section.

\begin{proof}[Proof of \cref{lem:TAconcurrent}]
	We first focus on the case where we have just one algorithm $\A$. Using the characterization of the kernel of $A$ in \cref{lem:kernel}, we see that
	\[\Cal C : \underset{X}{\underbrace{\Span\{\ket{q_{\ell-1}-1}\ket{-} : \ell \in [S+1] \setminus \{1\}\}}} \otimes \C^{[n] \times \mathcal{W}} \mapsto \Ker(A)\]
	isometrically. Hence, we obtain that
	\[2\Pi_{\Ker(A)} - I = \Cal C\left[\left(2\Pi_X - I\right) \otimes I\right]\Cal C^{\dagger}.\]
	As we can implement the reflection around $X$ using $\O(\log(T))$ extra gates and a single controlled query to $\O_\S$, the cost of reflecting around $\Ker(A)$ essentially becomes twice the cost of implementing $\Cal C$, which is given in \cref{lem:constrkerA}.

	If we have multiple algorithms $\A^{(j)}$, we can use the exact same idea, but now we should use concurrent access to the $\Cal C^{(j)}$'s and a concurrent reflection around the spaces $X^{(j)}$'s. The cost of implementing concurrent access to $\{\Cal C^{(j)}\}_{j=1}^n$ is analyzed in \cref{lem:constrkerA}, and the concurrent reflection around the $X^{(j)}$'s can be implemented with $\O(\mathrm{polylog}(T_{\max}))$ gates and one controlled call to $\O_\S$. This completes the proof.
\end{proof}

As a final remark, we would like to point out that this is not the only possible construction of the reflection around the kernel of $A$. One could alternatively employ a more general method of constructing a block-encoding of $A$ and using phase estimation to separate the vectors in the kernel of $A$ from those that are orthogonal to it. Implementing this construction carefully yields the same time and query complexity, but requires a spectral analysis of $A$, which is possible but turns out to be quite involved.

\subsection{Proof of \cref{lem:TC}: Implementation of reflection around $\H(x)$}
\label{app:TC}

In this section, we prove \cref{lem:TC}, i.e., we provide an implementation of the circuit that reflects through $\H(x)$. We also elaborate on how one would do this concurrently. The results are summarized in the lemma below, of which \cref{lem:TC} is a special case.

\begin{lemma}
	\label{lem:TCconcurrent}
	\lemTC*

	Similarly, let $\{\A^{(j)}\}_{j=1}^n$ be a set of clean quantum query algorithms. For all $j \in [n]$, let $S_j$ and $T_j$ denote the query and time complexity of $\A^{(j)}$. Let $P^{(j)}$ be the span program of $\A^{(j)}$. Then we can implement concurrent access to $\{2\Pi_{\H^{(j)}(x^{(j)})} - \I\}_{j=1}^n$ with $\O(1)$ calls to $\O_x$ and $\O_\S$ and auxiliary qubits, and $\O(\mathrm{polylog}(T_{\max}))$ extra gates.
\end{lemma}

\begin{proof}
	First, we consider the case where we just have one algorithm, $\A$. For all $x \in \{0,1\}^n$, recall that
	\begin{align*}
	\H(x) &= \bigoplus_{i=1}^n \H_{i,x_i} \oplus \H_{\text{true}}= \Span\{\ket{t,x_i,i,j} : t + 1 \in \mathcal{S}, i \in [n], j \in \mathcal{W}\}\\
	&\oplus \Span\{\ket{t,0,i,j} : t + 1 \in [T+1] \setminus \mathcal{S}, i \in [n], j \in \mathcal{W}\}.
	\end{align*}
	From this, and the definition of $\H$, it readily follows that the orthogonal complement of $\H(x)$ is given by
	\[\H(x)^{\perp} = \Span\{\ket{t,1-x_i,i,j} : t+1 \in \mathcal{S}, i \in [n], j \in \mathcal{W}\}.\]
	In order to reflect around $\H(x)$, all we have to do is put a minus phase if we are in $\H(x)^{\perp}$. To that end, call the oracle $\O_{\mathcal{S}}$ once to distinguish whether the time step in the first register is a state $\ket{t}$ for which $t + 1 \in \mathcal{S}$ (e.g. by first incrementing the first register, performing the call and then decrementing again). Store this bit in an auxiliary flag register. Next, conditioned on the flag qubit being $1$, perform one query to $\O_x$ to get a phase $(-1)^{x_i}$. Finally, apply $-Z$ to the second register, also controlled on the flag qubit, where $Z$ is the Pauli-$Z$ gate. Then if the second register is in the state $\ket{1-x_i}$, the overall phase will be $-(-1)^{x_i+1-x_i}=(-1)$, and if it is in the state $\ket{x_i}$, the overall phase will be $(+1)$, as desired. Finally, we need to uncompute the flag qubit, which again takes one call to $\O_{\mathcal{S}}$. All the other operations can be implemented in a number of elementary gates that is polylogarithmic in $T$.

	If we instead have multiple algorithms $\{\A^{(j)}\}_{j=1}^n$, then all that changes is the size of the time register. It is now of size $T_{\max} = \max_{j \in [n]} T_j$, and hence the arithmetic operations on it now require $\O(\mathrm{polylog}(T_{\max}))$ gates. This completes the proof.
\end{proof}

\subsection{Proof of \cref{lem:TB}: Construction of $\ket{w_0}$}
\label{app:TB}

The goal of this section is to prove \cref{lem:TB}, i.e., we provide an implementation of the circuit that constructs the minimal positive witness that is analytically calculated in \cref{lem:w0}. Additionally, we also elaborate on how one would do this concurrently, because we need this is in the proof of \cref{thm:variable-time}. The results are summarized in the following lemma, of which \cref{lem:TB} is a special case.

\begin{lemma}
	\label{lem:TBconcurrent}
	\lemTB*

	Similarly, let $\{\A^{(j)}\}_{j=1}^n$ be a set of clean quantum query algorithms. For all $j \in [n]$, let $S_j$ and $T_j$ denote the query and time complexity of $\A^{(j)}$, respectively, and let $\varepsilon_j$ be the error probability. Let $P^{(j)}$ be the span program of $\A^{(j)}$, and let $\ket{w_0^{(j)}}$ be the minimal positive witness of $P^{(j)}$. Then we can implement concurrent access $\mathcal{C}$ to $\{\Cal C_{\ket{w_0^{(j)}}}\}_{j=1}^n$ such that
	\[\sup_{\substack{\ket{h} \in \H_x \\ \norm{\ket{h}} = 1}} \norm{\Pi_{\H_x^{\perp}}\mathcal{C}\ket{h}} \leq \max_{j \in [n]} 2\sqrt{2\varepsilon_j},\]
	up to error $\delta > 0$ in operator norm with $\O(T/S)$ calls to $\O_\A$ and $\O_\S$, $\O(1)$ auxiliary qubits and a number of gates that satisfies
	\[\O\left(\max_{j \in [n]} \frac{T_j}{S_j} \mathrm{polylog}\left(T_{\max},\frac{1}{\delta}\right)\right).\]
\end{lemma}

We first name the three parts that form $\ket{w_0}$. From the form of $\ket{w_0}$ in \cref{lem:w0} we have
\[\frac{\ket{w_0}}{\norm{\ket{w_0}}} = \frac{\ket{\psi}}{\sqrt{N}} + \frac{\ket{\chi}}{\sqrt{N}} + \frac{\ket{\phi}}{\sqrt{N}},\]
where
\begin{align*}
	\ket{\psi} &= \frac{1}{M}\sum_{t=0}^{q_1-2} \ket{t}\ket{0}U_t \cdots U_1\ket{\Psi_0} + \ket{q_1-1}\left(\frac12\ket{0} + \frac12\ket{1}\right)U_{q_1-1} \cdots U_1\ket{\Psi_0}, \\
	\ket{\chi} &= \frac{1}{Ca^2+1}\left[\ket{q_S-1}\left(\frac12\ket{0} - \frac12\ket{1}\right)U_{q_S+1}^{\dagger} \cdots U_T^{\dagger}\ket{\Psi_T} + \frac{1}{M}\sum_{t=q_S+1}^{T-1} \ket{t}\ket{0} U_{t+1}^{\dagger} \cdots U_T^{\dagger}\ket{\Psi_T}\right], \\
	\ket{\phi} &= -\frac{Ca}{Ca^2+1}\ket{T}\ket{0}\ket{\Psi_T}.
\end{align*}
We can easily calculate the norms of the respective vectors, which results in
\[\norm{\ket{\psi}}^2 = \frac{q_1-1}{M^2} + \frac12, \qquad \norm{\ket{\chi}}^2 = \frac{C}{(Ca^2+1)^2}, \qquad \text{and} \qquad \norm{\ket{\phi}}^2 = \frac{C^2a^2}{(Ca^2+1)^2}.\]

By assumption on the state $\ket{\Psi_T}$, the state $\ket{\phi}/\norm{\ket{\phi}}$ can be generated in $\O(1)$ gates. The generation of the other two states is somewhat harder and is the focus of the following lemma.

\begin{lemma}
	\label{lem:Cpsi}
	We can implement routines $\Cal C_{\ket{\psi}}$ and $\Cal C_{\ket{\chi}}$ that map $\ket{0}\ket{0}\ket{\Psi_0}$ to $\ket{\psi}/\norm{\ket{\psi}}$ and $\ket{T-1}\ket{0}U_T^{\dagger}\ket{\Psi_T}$ to $\ket{\chi}/\norm{\ket{\chi}}$, respectively, leave $\overline{\H}_x$ and $\widetilde{\H}_x$ invariant, and moreover leave all the states that have disjoint support from $\ket{\psi}$ resp.\ $\ket{\chi}$ invariant, up to error $\delta > 0$ in operator norm with $\O(T/S)$ calls to $\O_\A$, $\O(1)$ auxiliary qubits and a number of gates that satisfies
	\[\O\left(\frac{T}{S}\mathrm{polylog}\left(T,\frac{1}{\delta}\right)\right).\]

	Similarly, we can provide concurrent access to $\{\Cal C_{\ket{\psi}}^{(j)}\}_{j=1}^n$ and $\{\Cal C_{\ket{\chi}}^{(j)}\}_{j=1}^n$ up to precision $\delta > 0$ with a number of calls to $\O_\A$ and $\O_\S$ of $\O(\max_{j \in [n]} T_j/S_j)$, $\O(1)$ auxiliary qubits and a number of extra gates that scales as
	\[\O\left(\max_{j \in [n]} \frac{T_j}{S_j} \mathrm{polylog}\left(T_{\max},\frac{1}{\delta}\right)\right).\]
\end{lemma}

\begin{proof}
	Note that $\Cal C_{\ket{\psi}}$ and $\Cal C_{\ket{\chi}}$ are very similar to the circuit $\Cal C$ defined in \cref{def:construction_kernel}. We can use the exact same techniques as we used in implementing $\Cal C$ in \cref{lem:constrkerA} to implement $\Cal C_{\ket{\psi}}$ and $\Cal C_{\ket{\chi}}$. The cost of implementing these routines carries over from the proof of \cref{lem:constrkerA}. This completes the proof.
\end{proof}

\begin{proof}[Proof of \cref{lem:TBconcurrent}]
	We first restrict ourselves to the case where we just have one algorithm, $\A$, and we postpone the treatment of the case where we have a set of algorithms to the final paragraph of this proof.

	First of all, we implement a circuit $\Cal C_1$ whose action is
	\[\Cal C_1 : \left\{\begin{array}{rcl}
		\ket{0}\ket{0}\ket{\Psi_0} &\mapsto& \alpha\ket{0}\ket{0}\ket{\Psi_0} + \sqrt{1-\alpha^2}\ket{T}\ket{0}\ket{\Psi_T} \\
		\ket{T}\ket{0}\ket{\Psi_T} &\mapsto& -\sqrt{1-\alpha^2}\ket{0}\ket{0}\ket{\Psi_0} + \alpha\ket{T}\ket{0}\ket{\Psi_T}
	\end{array}\right., \qquad \text{with} \qquad \alpha = \frac{\norm{\ket{\psi}}}{\norm{\ket{w_0}}}.\]
	This circuit can be implemented in a similar way as we implemented the splitting map in \cref{lem:splitting_map}. Conditioned on the first register being in state $\ket{T}$, one first applies the map $\ket{\Psi_T} \mapsto \ket{\Psi_0}$ to the last register, which amounts to applying a controlled $X$ operation targeted to the answer register. Then, one applies a rotation on a two-dimensional subspace of the state space of the first register spanned by $\Span\{\ket{0},\ket{T}\}$, which can be implemented with $\O(\mathrm{polylog}(T))$ gates as it is a register on $\log(T)$ qubits. Finally, one applies the mapping $\ket{\Psi_0} \mapsto \ket{\Psi_T}$, again controlled on the first register being in state $\ket{T}$. Counting the auxiliary qubits and gates reveals that we can do this with $\O(1)$ auxiliary qubits and $\O(\mathrm{polylog}(T))$ gates.

	Next, one applies the mapping $\Cal S_{T-1,\beta}^{\dagger}$, with $\beta = \norm{\ket{\chi}}/\norm{\ket{\chi} + \ket{\phi}}$. The combined mapping now acts as
	\[\Cal S_{T-1,\beta}^{\dagger} \Cal C_1 : \ket{0}\ket{0}\ket{\Psi_0} \mapsto \frac{\norm{\ket{\psi}}}{\sqrt{N}}\ket{0}\ket{0}\ket{\Psi_0} + \frac{\norm{\ket{\chi}}}{\sqrt{N}}\ket{T-1}\ket{0}U_T^{\dagger}\ket{\Psi_T} + \frac{\norm{\ket{\phi}}}{\sqrt{N}}\ket{T}\ket{0}\ket{\Psi_T}.\]
	Thus, all that is left is mapping the first term to $\ket{\psi}$ and the second to $\ket{\chi}$ using the circuits that we already have for them, meaning that
	\[\Cal C_{\ket{w_0}} = \Cal C_{\ket{\psi}}\Cal C_{\ket{\chi}}\Cal S_{T_1,\beta}^{\dagger}\Cal C_1 : \ket{0}\ket{0}\ket{\Psi_0} \mapsto \frac{\ket{w_0}}{\sqrt{N}}.\]
	Moreover, all the four subroutines can be implemented with $\O(T/S)$ calls to $\O_\A$, polylogarithmically in $T$ many auxiliary qubits, and a number of gates that satisfies
	\[\O\left(\frac{T}{S}\mathrm{polylog}(T)\right).\]
	Now all that is left to check is that $\Cal C_{\ket{w_0}}$ leaves $\H_x$ approximately invariant. We already know from \cref{lem:splitting_map,lem:Cpsi} that $\Cal S_{T-1,\beta}^{\dagger}$, $\Cal C_{\ket{\psi}}$ and $\Cal C_{\ket{\chi}}$, respectively, leave both $\overline{\H}_x$ and $\widetilde{\H}_x$ invariant. Moreover, $\Cal C_1$ keeps $\overline{\H}_x + \widetilde{\H}_x$ invariant, and hence whenever $f(x) = 0$, then $\Cal C_{\ket{w_0}}$ leaves $\H_x$ invariant. On the other hand, when $f(x) = 1$, we have, for any $\ket{h} \in \H_x = \overline{\H}_x$ with $\norm{\ket{h}} = 1$:
	\[\norm{\Pi_{\H_x^{\perp}}\Cal C_{\ket{w_0}}\ket{h}} = \norm{\Pi_{\H_x^{\perp}}\Cal C_1\ket{h}}.\]
	As $\Cal C_1$ acts as the identity operation on all states that have disjoint support on the time register from $\{0,T\}$, we can without loss of generality assume that $\ket{h} = a\ket{0}\ket{0}\ket{\Psi_0} + b\ket{T}\ket{0}\ket{\Psi_T(x)}$ with $a,b \in \C$. Thus, we find that
	\begin{align*}
		\sup_{\substack{\ket{h} \in \H_x \\ \norm{\ket{h}} = 1}} \norm{\Pi_{\H_x^{\perp}}\Cal C_{\ket{w_0}}\ket{h}} &= \sup_{\substack{a,b \in \C \\ |a|^2 + |b|^2 = 1}} \norm{\Pi_{\H_x^{\perp}}\Cal C_1(a\ket{0}\ket{0}\ket{\Psi_0} + b\ket{T}\ket{0}\ket{\Psi_T(x)})} \\
		&\leq \sup_{\substack{a,b \in \C \\ |a|^2 + |b|^2 = 1}} \norm{\Pi_{\H_x^{\perp}}\left[\left(a\alpha - b\sqrt{1-\alpha^2}\right)\ket{0}\ket{0}\ket{\Psi_0} + \left(a\sqrt{1-\alpha^2} + b\alpha\right)\ket{T}\ket{0}\ket{\Psi_T}\right]} \\
		&\;\;\;\;\; + |b|\norm{\ket{\Psi_T} - \ket{\Psi_T(x)}} \\
		&\leq \sup_{\substack{a,b \in \C \\ |a|^2 + |b|^2 = 1}} \left|a\sqrt{1-\alpha^2} + b\alpha\right| \norm{\ket{\Psi_T} - \ket{\Psi_T(x)}} + |b|\norm{\ket{\Psi_T} - \ket{\Psi_T(x)}} \\
		&\leq \left[\sup_{\substack{a,b \in \C \\ |a|^2 + |b|^2 = 1}} \left(|a|^2 + |b|^2\right)\left(\alpha^2 + 1 - \alpha^2\right) + |b|\right] \cdot \norm{\ket{\Psi_T} - \ket{\Psi_T(x)}} \leq 2\sqrt{2\varepsilon},
	\end{align*}
	where we used Cauchy-Schwarz in the last line.

	If we have multiple algorithms $\{\A^{(j)}\}_{j=1}^n$, we simply run the concurrent versions of $\Cal C_{\ket{\psi}}$ and $\Cal C_{\ket{\chi}}$, and we supply a concurrent version of $\Cal C_1$ which we build using the same techniques as in \cref{lem:constrkerA}. With this, we can successfully implement concurrent access $\mathcal{C}$ to $\{\Cal C_{\ket{w_0^{(j)}}}\}_{j=1}^n$ with the desired complexities. As all the individual $\Cal C_{\ket{w_0^{(j)}}}$'s are implemented up to precision $\delta$, so is their concurrent access routine.

	It remains to check that the resulting concurrent access routine approximately leaves $\H_x$ invariant. To that end, take $\ket{h}$ in $\H_x$, such that $\norm{\ket{h}} = 1$. We can write
	\[\ket{h} = \sum_{j=1}^n \alpha_j\ket{j}\ket{h^{(j)}} \qquad \text{with} \qquad \norm{\alpha} = 1,\]
	where all the $\ket{h^{(j)}}$'s are unit vectors in $\H^{(j)}_{x^{(j)}}$. Now, we find that
	\[\norm{\Pi_{\H_x^{\perp}}\mathcal{C}\ket{h}}^2 = \sum_{j=1}^n |\alpha|^2\norm{\Pi_{(\H^{(j)}_{x^{(j)}})^{\perp}}\Cal C_{\ket{w_0^{(j)}}}\ket{h^{(j)}}}^2 \leq \sum_{j=1}^n |\alpha|^2 \cdot 8\varepsilon_j \leq \max_{j \in [n]} 8\varepsilon_j.\]
	This completes the proof.
\end{proof}

\subsection{Proof of \cref{lem:consistent-refl}: Implementation of reflection around $\ket{0}$}
\label{app:ConRef}

The goal of this final section in \cref{app:time-complexity} is to prove \cref{lem:consistent-refl}, i.e., we provide an implementation of a quantum circuit that reflects around the all-zeros state on the implementing subspace. The details can be found below.

\lemConRef*

\begin{proof}
	As in \cref{def:QA}, we assume the basis states of the workspace are labeled by $\W = \W'\times\{0,1\}$, so the last qubit is the answer register.
	The basis states of the overall space $\H$ are $\ket{t}\ket{b}\ket{i,j,a}$, where $i\in [n]$, $j\in \W'$, and $a\in\{0,1\}$ is the content of the answer register.
	The map $G$ reflects around states with $t=0$, $b=0$ and $a=0$:
	\begin{align*}
		G\ket{t}\ket{b}\ket{i,j,a} &= \left\{\begin{array}{ll}
		\ket{t}\ket{b}\ket{i,j,a} & \mbox{if }t=b=a=0\\
		-\ket{t}\ket{b}\ket{i,j,a} & \mbox{else.}
		\end{array}\right.
	\end{align*}
	To implement this reflection, we simply compute a bit $b_{\text{flag}}$ in a new register such that $b_{\text{flag}}=0$ if and only if $t=0$, $b=0$ and $a=0$. Then we can apply a $Z$-gate on this register, and then uncompute. Since $t\in [T]_0$, $a,b\in \{0,1\}$, this can be done in time $\O(\log T)$.

	Let $\ket{h} \in \H_x$. Any part of $\ket{h}$ supported on $\ket{t}$ in the first register for $t\neq 0$ will be reflected by $G$, which is the desired behavior. If $f(x) = 1$, we can without loss of generality assume that $\ket{h} = \ket{0}\ket{0}\ket{\Psi_0}$. Observe that
	\[\bigl(G - \left(2\ket{0,0,\Psi_0}\bra{0,0,\Psi_0} - \I\right)\bigr)\ket{0,0,\Psi_0(x)} = 0,\]
	so $G$ acts the same as $\Cal R_{\ket{0}}$ on $\H_x$.

	On the other hand, suppose that $f(x) = 0$. Now we can without loss of generality assume that $\ket{h}\in \Span\{\ket{0}\ket{0}\ket{\Psi_0(x)},\ket{0}\ket{0}\ket{\widetilde{\Psi}_0(x)}\}$. Thus, without loss of generality,
	\[\ket{h} \in \alpha_1\ket{0}\ket{0}\ket{\Psi_0(x)} + \alpha_2\ket{0}\ket{0}\ket{\widetilde{\Psi}_0(x)}.\]
	Hence,
	\begin{align*}
		&\sup_{\ket{h} \in \H_x} \frac{\norm{\bigl(G - \left(2\ket{0,0,\Psi_0}\bra{0,0,\Psi_0} - \I\right)\bigr)\ket{h}}}{\norm{\ket{h}}}\\
		={}& \sup_{\alpha_1,\alpha_2 \in \C} \frac{\norm{\bigl(G - \left(2\ket{0,0,\Psi_0}\bra{0,0,\Psi_0} - \I\right)\bigr)\bigl[\alpha_1\ket{0}\ket{0}\ket{\Psi_0(x)} + \alpha_2\ket{0}\ket{0}\ket{\widetilde{\Psi}_0(x)}\bigr]}}{\norm{\alpha_1\ket{0}\ket{0}\ket{\Psi_0(x)} + \alpha_2\ket{0}\ket{0}\ket{\widetilde{\Psi}_0(x)}}} \\
		={}&\sup_{\alpha \in \C} \frac{\norm{\bigl(G - \left(2\ket{0,0,\Psi_0}\bra{0,0,\Psi_0} - \I\right)\bigr)\ket{0}\ket{0}\ket{\widetilde{\Psi}_0(x)}}}{\norm{\alpha\ket{0}\ket{0}\ket{\Psi_0(x)} + \ket{0}\ket{0}\ket{\widetilde{\Psi}_0(x)}}}.
	\end{align*}
	For any $\alpha \in \C$, we can rewrite the denominator as
	\begin{align*}
		\norm{\alpha\ket{0}\ket{0}\ket{\Psi_0(x)} + \ket{0}\ket{0}\ket{\widetilde{\Psi}_0(x)}} &= \sqrt{1 + |\alpha|^2 + 2\Re\left(\alpha\braket{\Psi_0(x)}{\widetilde{\Psi}_0(x)}\right)} \geq \sqrt{1 + |\alpha|^2 - 2|\alpha|p_1(x)} \\
		&= \sqrt{\left(|\alpha| - p_1(x)\right)^2 + 1 - p_1(x)^2} \geq \sqrt{1 - p_1(x)^2} \geq \sqrt{1-\varepsilon^2}.
	\end{align*}
	For the numerator, observe by \cref{lem:almost-orthog} that
	\[\ket{\widetilde{\Psi}_0(x)} = p_1(x)\ket{\Psi_0(x)} + p_0(x)(I \otimes X)\ket{\Psi_0(x)} + \sqrt{2p_0(x)p_1(x)}\ket{\bot}\]
	for some normalized state $\ket{\bot}$ that is orthogonal to $\ket{\Psi_0(x)}$ and $(I \otimes X)\ket{\Psi_0(x)}$. Note that $\ket{\Psi_0(x)} = \ket{\Psi_0}=\ket{0,0,0}$ and hence it is handled correctly by $G$. Similarly, $(I \otimes X)\ket{\Psi_0(x)}=\ket{0,0,1}$ has answer bit $1$, so it is handled correctly by $G$ as well. Thus, we end up with
	\[\norm{\bigl(G - \left(2\ket{0,0,\Psi_0}\bra{0,0,\Psi_0} - \I\right)\bigr)\ket{0,0,\widetilde{\Psi}_0(x)}} \leq 2\sqrt{2p_0(x)p_1(x)} \leq 2\sqrt{2\varepsilon},\]
	which implies that
	\[\norm{G - \left(2\ket{0,0,\Psi_0}\bra{0,0,\Psi_0} - \I\right)} \leq \frac{2\sqrt{2\varepsilon}}{\sqrt{1-\varepsilon^2}} \leq 4\sqrt{2\varepsilon}.\]
	This completes the proof.
\end{proof}

\section{OR span program analysis}
\label{app:OR}

In this final appendix, we first prove the two lemmas about the structure of the minimal positive witness $\ket{w_0}$ and the kernel of the span program operator of $P$, defined in \cref{eq:or-span}. After that, we provide the proofs of two technical lemmas that formalize the procedure to produce states in superposition with certain amplitudes.

\ORw*

\begin{proof}
	Observe that for every choice of $\beta_j$'s that sum to $1$, we have
	\[A\left[\sum_{j=1}^n \frac{\beta_j}{\sqrt{W_+^{(j)}}} \ket{j} \otimes \ket{w_0^{(j)}}\right] = \sum_{j=1}^n \beta_j A^{(j)}\ket{w_0^{(j)}} = \sum_{j=1}^n \beta_j \ket{\tau} = \ket{\tau},\]
	and that the minimal positive witness has to be of this form. Moreover, for all such choices of $\beta_j$, we have
	\[\norm{\alpha} \cdot \norm{\sum_{j=1}^n \frac{\beta_j}{\sqrt{W_+^{(j)}}}\ket{j} \otimes \ket{w_0^{(j)}}} = \sqrt{\sum_{k=1}^n \alpha_k^2} \cdot \sqrt{\sum_{j=1}^n \frac{|\beta_j|^2\norm{\ket{w_0^{(j)}}}^2}{W_+^{(j)}}} \geq \sum_{j=1}^n \alpha_j \cdot \frac{\beta_j}{\alpha_j}  = \sum_{j=1}^n \beta_j = 1,\]
	where we used the Cauchy-Schwarz inequality. Thus, we find that for all choices of $\beta_j$,
	\[\norm{\sum_{j=1}^n \frac{\beta_j}{\sqrt{W_+^{(j)}}}\ket{j} \otimes \ket{w_0^{(j)}}}^2 \geq \frac{1}{\norm{\alpha}^2},\]
	and the tightness of this inequality is shown by picking $\beta_j = \alpha_j^2 / \norm{\alpha}^2$. Thus, the minimal witness is
	\[\ket{w_0} = \frac{1}{\norm{\alpha}^2} \sum_{j=1}^n \alpha_j \ket{j} \otimes \frac{\ket{w_0^{(j)}}}{\norm{\ket{w_0^{(j)}}}},\]
	completing the proof.
\end{proof}

\ORkernel*

\begin{proof}
	First, observe that $\bigoplus_{j=1}^n\Span\{\ket{j}\}\otimes \ker (A^{(j)})\subseteq \ker(A)$, since for any $\ket{h_j}\in \ker(A^{(j)})$, $A\ket{j}\ket{h_j}=\sqrt{W_+^{(j)}}A^{(j)}\ket{h_j} = 0$. Similarly, observe that $\ket{0,0}$ vanishes under $A$, so it is part of the kernel of $A$ as well.

	Thus, suppose $\ket{h}=\sum_{j=1}^n\ket{j}\ket{h_j}\in \H$ is in $\ker(A)$, and for all $j\in [n]$, $\ket{h_j}\in\mathrm{row}(A^{(j)})$. For all $k \in [n]$, we have
	\[0 = \Pi_{\overline{\V}_k}\left[A\sum_{j=1}^n\ket{j}\ket{h_j}\right] = \Pi_{\overline{\V}_k}\left[\sum_{j=1}^n\sqrt{W_+^{(j)}}A^{(j)}\ket{h_j}\right] = \Pi_{\overline{\V}_k}A^{(k)}\ket{h_k},\]
	where we use the fact that $\Pi_{\overline{\V}_k}A^{(j)} = 0$ whenever $j\neq k$. Thus, we have that for all $k\in [n]$, $A^{(k)}\ket{h_k}\in \Span\{\ket{0}\}$. Since we also have that $\ket{h_k}\in\mathrm{row}(A^{(k)})$, it must be the case that $\ket{h_k}\in \Span\{\ket{w_0^{(k)}}\}$, so let $\ket{h_k}=\beta_k\ket{w_0^{(k)}}$. Then, we have:
	\begin{align*}
		0 &= \sum_{j=1}^n \beta_j A^{(j)}\ket{w_0^{(j)}} = \sum_{j=1}^n \beta_j\sqrt{W_+^{(j)}} \ket{0} = \left[\sum_{j=1}^n \beta_j\alpha_j \cdot \frac{\braket{w_0^{(j)}}{w_0^{(j)}}}{\norm{\ket{w_0^{(j)}}}}\right]\ket{0} = N\braket{w_0}{h}\ket{0}
	\end{align*}
	Hence, $\ket{h}=\sum_{j=1}^n\alpha_j\ket{j}\ket{w_0^{(j)}}\in \Ker A$ if it is orthogonal to $\ket{w_0}$, which completes the proof.
\end{proof}

Now we prove \cref{lem:binning-technique} and \cref{lem:Calpha}.

\binningTech*

\begin{proof}
	We let $k' = \lceil\log(\gamma_{\max}/\gamma_{\min})\rceil$. Then, with every $j \in [n]$, we associate the unique integer $m_j$ such that $\gamma_j \in [\gamma_{\min} \cdot 2^{m_j-1}, \gamma_{\min} \cdot 2^{m_j})$. As
	\[\gamma_{\min} \cdot 2^0 \leq \gamma_{\min} \leq \gamma_j \leq \gamma_{\max} = \gamma_{\min} \cdot 2^{\log(\gamma_{\max}/\gamma_{\min})} \leq \gamma_{\min} \cdot 2^{k'},\]
	we find that $m_j \leq k'$, and hence $m_j \in [k']$. Moreover, as the $\gamma_j$'s are non-decreasing, so are the $m_j$'s. Now, for every $\ell \in [k']$, we define $j'_{\ell} = \max \{j \in [n] : m_j = \ell\}$ and we let $j'_0 = 0$. We find that for all $\ell \in [k']$ and $j \in [j'_{\ell-1}+1,j'_{\ell}]$,
	\[\frac{\gamma_{j'_{\ell}}}{2} \leq \frac{\gamma_{\min} \cdot 2^{\ell}}{2} = \gamma_{\min} \cdot 2^{\ell-1} \leq \gamma_j \leq \gamma_{j'_{\ell}},\]
	and hence the second condition is verified.

	Now, for every $\ell \in [k']$, we write $j'_{\ell} - j'_{\ell-1}$ in terms of its binary expansion, i.e., we write
	\[j'_{\ell} - j'_{\ell-1} = 2^{p_{\ell,1}} + \cdots + 2^{p_{\ell,k_{\ell}}},\]
	where $p_{\ell,1} > \cdots > p_{\ell,k_{\ell}}$. As $j'_{\ell} - j'_{\ell-1} \leq n$, we have that $k_{\ell} \leq \lceil\log(n)\rceil$. Finally, we let
	\[(j_{\ell})_{\ell=1}^k = (j'_0, j'_0 + 2^{p_{1,1}}, \dots, j'_0 + 2^{p_{1,1}} + \dots + 2^{p_{1,k_1-1}}, j'_1, \dots, j'_{k'})\]
	The difference between two consecutive terms is always a power of two by construction, and the length satisfies
	\[k = \sum_{\ell=1}^{k'} k_{\ell} \leq k' \cdot \lceil \log(n) \rceil \leq \lceil \log(\gamma_{\max}/\gamma_{\min}) \rceil \cdot \lceil \log(n) \rceil,\]
	completing the proof.
\end{proof}

\Calpha*

\begin{proof}
	For all $j \in [n]$, let
	\[\gamma_j = \frac{\sqrt{3(2S_j+1)}}{\norm{\ket{w_0^{(j)}}}}.\]
	Assume without loss of generality that the algorithms $\A^{(j)}$ are ordered such that $0 < \gamma_{\min} = \gamma_1 \leq \cdots \leq \gamma_n = \gamma_{\max}$. From \cref{lem:w0} we deduce that $1/\sqrt{2} \leq \norm{\ket{w_0^{(j)}}} \leq \sqrt{3}$, and hence
	\[\frac{\gamma_{\max}}{\gamma_{\min}} \leq \sqrt{\frac{3(2S_{\max}+1)}{9}} \cdot \sqrt{6} \leq \sqrt{6S_{\max}}.\]
	According to \cref{lem:binning-technique}, we can now find a sequence $0 = j_1 \leq \cdots \leq j_k = n$ with $k \leq \lceil\frac12\log(6S_{\max})\rceil \cdot \lceil\log(n)\rceil$, such that for all $\ell \in [k]$ we have that $j_{\ell} - j_{\ell-1}$ is a power of two and for all $j \in [j_{\ell-1}+1,j_{\ell}]$, we have
	\[\frac{\gamma_{j_{\ell}}}{2} \leq \gamma_j \leq \gamma_{j_{\ell}}.\]
	Given such a $j$, we define $W_+^{(j)} = \gamma_{j_{\ell}}^2 \cdot \norm{\ket{w_0^{(j)}}}^2$.
	Then we find
	\[W_+^{(j)} = \gamma_{j_{\ell}}^2 \cdot \norm{\ket{w_0^{(j)}}}^2 \leq 4\gamma_j^2 \cdot \norm{\ket{w_0^{(j)}}}^2 = 12(2S_j+1),\]
	and according to \cref{lem:w+},
	\[W_+(P^{(j)}) \leq 3(2S_j+1) = \gamma_j^2\norm{\ket{w_0^{(j)}}}^2 \leq \gamma_{j_{\ell}}^2 \cdot \norm{\ket{w_0^{(j)}}}^2 = W_+^{(j)}.\]
	Moreover, we have for all such $j$ that
	\[\alpha_j = \frac{\sqrt{W_+^{(j)}}}{\norm{\ket{w_0^{(j)}}}} = \gamma_{j_{\ell}} = \frac{\sqrt{W_+^{(j_{\ell})}}}{\norm{\ket{w_0^{(j_{\ell})}}}} = \alpha_{j_{\ell}}.\]

	Thus, it remains to show that we can implement $\Cal C_{\alpha}$ in $\O(\log(S_{\max})\log(n))$ gates. To that end, we first of all define the mapping $\Cal S$ that acts as the identity on $\ket{0}\ket{0}\ket{0}$ and that given a $j \in [j_{\ell-1}+1,j_{\ell}]$ implements
	\[\Cal S : \ket{j}\ket{0}\ket{0} \mapsto \ket{0}\ket{\ell}\ket{j - j_{\ell-1} - 1},\]
	where the registers are of size $\lceil \log(n+1) \rceil$, $\lceil \log(k+1) \rceil$ and $\lceil \log(n+1) \rceil$, respectively. Moreover, as the values of the $j_{\ell}$'s are known beforehand, this can be implemented with $\O(k)$ arithmetic circuits that all have $\O(\log(n))$ gates, so the number of gates needed to implement $\Cal S$ is $\O(\log(S_{\max})\log^2(n))$.

	We define the subspace
	\[\Cal X = \Span\{\ket{0}\ket{0}\ket{0}\} \oplus \Span\{\ket{0}\ket{\ell}\ket{j-j_{\ell-1}-1} : j \in [j_{\ell-1}+1,j_{\ell}]\}.\]
	Observe that $\Cal S$ maps any state $\ket{j}\ket{0}\ket{0}$ into $\Cal X$, and moreover that $\Cal S^{\dagger}$ will set the final two registers to $\ket{0}$ when it is applied to any state in $\Cal X$. Hence, as long as we stay in $\Cal X$, we can always uncompute the final two registers.

	Now, we implement $\Cal C_{\alpha}$ as follows, where we treat the final two registers as ancilla registers.
	\begin{enumerate}
		\setlength\itemsep{-.4em}
		\item We apply $\Cal S$. This maps our state into $\Cal X$, and will leave $\ket{0}\ket{0}\ket{0}$ unaltered.
		\item Controlled on the last register being in the state $\ket{0}$, we apply on the second register the map
		\[\Cal C : \ket{0} \mapsto \frac{1}{\norm{\alpha}} \sum_{\ell=1}^k \alpha_{\ell}\sqrt{j_{\ell}-j_{\ell-1}} \ket{\ell}.\]
		This leaves $\Cal X$ invariant as $\ket{0}\ket{\ell}\ket{0} \in \Cal X$ for every $\ell \in [k]_0$. As this is a map on $\O(\log(k))$ qubits, it can be implemented in $\O(k) = \O(\log(S_{\max})\log(n))$ gates.
		\item Next, controlled on the second register being in state $\ket{\ell}$, we perform the map $I \otimes H^{\otimes \log(j_{\ell}-j_{\ell-1})}$ to the final register. This only affects the basis states in $\Cal X$ and implements
		\[\ket{0}\ket{\ell}\ket{0} \mapsto \ket{0}\ket{\ell} \frac{1}{\sqrt{j_{\ell}-j_{\ell-1}}} \sum_{m=0}^{j_{\ell}-j_{\ell-1}-1} \ket{m}.\]
		This circuit can be built using $k$ times at most $\log(n)$ controlled Hadamards, plus some arithmetic circuits on $\log(k)$ qubits to set the controls. Hence, the number of gates needed to implement this step is $\O(k\log(n) + \log(k)) = \O(\log(S_{\max})\log^2(n))$.
		\item We implement $\Cal S^{\dagger}$. Since steps 2 and 3 left $\Cal X$ invariant, we can now uncompute the final two registers.
	\end{enumerate}

	The total time complexity of $\Cal C_{\alpha}$ hence now becomes the sum of the time complexities of the above steps, which is $\O(\log(S_{\max})\log^2(n))$, and it maps
	\[\ket{0}\ket{0}\ket{0} \mapsto \ket{0}\frac{1}{\norm{\alpha}} \sum_{\ell=1}^k \alpha_{\ell}\sqrt{j_{\ell} - j_{\ell-1}}\ket{\ell}\ket{0} \mapsto \ket{0} \frac{1}{\norm{\alpha}}\sum_{\ell=1}^k \alpha_{\ell}\ket{\ell} \sum_{m=0}^{j_{\ell}-j_{\ell-1}-1}\ket{m} \mapsto \frac{1}{\norm{\alpha}}\sum_{j=1}^n \alpha_{\ell}\ket{j}\ket{0}\ket{0}.\]
	This completes the proof.
\end{proof}
\end{document}